\pdfoutput=1
\documentclass{lmcs}
\pdfoutput=1

\usepackage{lastpage}
\lmcsdoi{15}{2}{2}
\lmcsheading{}{\pageref{LastPage}}{}{}%
{Jan.~09,~2019}{Apr.~11,~2019}{}

\ACMCCS{Theory of computation $\to$ Modal and temporal logics; Complexity theory and logic;}

\usepackage[utf8]{inputenc}
\usepackage[T1]{fontenc}
\usepackage[english]{babel}
\usepackage{amsthm}

\usepackage{bookmark}

\usepackage{pdfcolmk}

\usepackage[heavycircles]{stmaryrd}
\usepackage{xcolor}

\usepackage{braket}

\usepackage[babel,autostyle=true,strict]{csquotes}
\MakeOuterQuote{"}

\usepackage{amssymb}

\usepackage[fixamsmath,disallowspaces]{mathtools}

\usepackage{stackengine}

\usepackage{fixmath}

\usepackage[xspace]{ellipsis}

\usepackage{relsize}
\usepackage{scalerel}

\usepackage{booktabs}
\usepackage{array}
\newcolumntype{L}{>{$}l<{$}}
 \newcolumntype{C}{>{$}c<{$}}

\usetikzlibrary{arrows}
\usetikzlibrary{fit}
\usetikzlibrary{positioning}
\usetikzlibrary{patterns}
\usetikzlibrary{shapes,backgrounds}
\usetikzlibrary{decorations.pathmorphing}
\usetikzlibrary{decorations.pathreplacing}

\usepackage{upgreek}

\theoremstyle{plain}\newtheorem*{claim}{Claim}

\newcommand*{\altqed}{\hfill\small\ensuremath{\triangleleft}}
\newenvironment{cproof}
{\begin{proof}[Proof of claim.]\let\oldqed\qed\let\qed\relax}
  {\altqed\end{proof}}

\makeatletter
\newtheorem*{rep@theorem}{\rep@title}
\newcommand{\newreptheorem}[2]{
\newenvironment{rep#1}[1]{
 \def\rep@title{#2~\ref{##1}}
 \begin{rep@theorem}}
 {\end{rep@theorem}}}
\makeatother
\newreptheorem{prop}{Proposition}
\newreptheorem{lem}{Lemma}

\newcommand{\ie}{i.e.\@\xspace}
\newcommand{\eg}{e.g.\@\xspace}
\newcommand{\wrt}{w.\,r.\,t.\@\xspace}
\newcommand{\wloss}{w.l.o.g.\@\xspace}
\newcommand{\Wloss}{W.l.o.g.\@\xspace}

\newcommand{\red}[1]{\color{red}#1}
\newcommand{\blue}[1]{\color{blue}#1}

\newcommand\+{\mkern2mu}

\DeclareMathOperator*{\bigovee}{\scalerel*{\ovee}{\sum}}

\newcommand{\pow}[1]{\ensuremath{\mathfrak{P}(#1)}}
\newcommand{\size}[1]{{\ensuremath{\vert\nobreak#1\nobreak\vert}}}

\newcommand{\imp}{\rightarrow}

\newcommand{\timp}{\rightarrowtriangle}

\newcommand{\qa}[1]{\forall^{1}_{#1}}
\newcommand{\qaa}[1]{\forall^{\subseteq}_{#1}}
\newcommand{\qe}[1]{\exists^1_{#1}}
\newcommand{\qee}[1]{\exists^{\subseteq}_{#1}}

\newcommand{\qeet}[2]{\exists^{\approx#2}_{#1}}

\newcommand{\xstate}[1]{#1\text{-}\mathsf{state}}

\renewcommand{\max}{\mathsf{max}}
\newcommand{\maxsf}{\mathrm{max}}

\providecommand{\dfn}{\vcentcolon=}

\providecommand{\ddfn}{\vcentcolon\vcentcolon=}
\newcommand{\N}{\ensuremath{\mathbb{N}}\xspace}
\newcommand{\E}{\ensuremath\mathsf{E}}
\newcommand{\PS}{\ensuremath\mathcal{PS}}

\newcommand{\VAL}{\ensuremath{\mathsf{VAL}}}
\newcommand{\SAT}{\ensuremath{\mathsf{SAT}}}

\newcommand{\negg}{{\sim}}

\newcommand{\dep}[1]{{=\!\!(#1)}}
\newcommand{\bis}{\rightleftharpoons}

\newcommand{\Prop}{\mathsf{Prop}}
\newcommand{\canon}{\mathsf{canon}}
\newcommand{\scopes}{\mathsf{scopes}}
\newcommand{\hook}{\hookrightarrow}

\newcommand{\md}{\mathsf{md}}

\newcommand{\logic}[1]{\ensuremath{\mathsf{#1}}\xspace}
\newcommand{\compclass}[1]{\ensuremath{\mathrm{#1}}\xspace}

\newcommand{\ML}{\logic{ML}}
\newcommand{\PTL}{\logic{PTL}}
\newcommand{\FO}{\logic{FO}}
\newcommand{\MTL}{\logic{MTL}}

\newcommand{\ATIMEALT}[1]{\ensuremath{\compclass{ATIME}\text{-}\compclass{ALT}(#1,\mathrm{poly})}}
\newcommand{\AEXPPOLY}{\ATIMEALT{\exp}}
\newcommand{\TOWERPOLY}{\ensuremath{\compclass{TOWER}(\mathrm{poly})}}
\newcommand{\TOWER}{\compclass{TOWER}}

\newcommand{\NEXPTIME}{\compclass{NEXPTIME}}

\newcommand{\PSPACE}{\compclass{PSPACE}}

\newcommand{\fraka}{\mathfrak{s}}

\newcommand{\frakt}{\mathfrak{t}}
\newcommand{\frakp}{\mathfrak{p}}
\newcommand{\frakq}{\mathfrak{q}}

\newcommand{\calF}{\ensuremath{\mathcal{F}}}

\newcommand{\calL}{\ensuremath{\mathcal{L}}}
\newcommand{\calK}{\ensuremath{\mathcal{K}}}

\newcommand{\calN}{\ensuremath{\mathcal{N}}}

\newcommand{\calR}{\ensuremath{\mathcal{R}}}

\newcommand{\sfD}{\ensuremath{\mathsf{D}}}

\newcommand{\sfL}{\ensuremath{\mathsf{L}}}

\newcommand{\bigO}[1]{\ensuremath{{\mathcal{O}(#1)}}}
\newcommand{\leqpm}{\ensuremath{\leq^\mathrm{P}_\mathrm{m}}}
\newcommand{\leqlogm}{\ensuremath{\leq^{\log}_\mathrm{m}}}
\newcommand{\leqelem}{\ensuremath{\leq^\mathrm{elem}_\mathrm{m}}}
\newcommand{\type}[1]{\llbracket{}#1\rrbracket{}}

\begin{document}

\title[Canonical Models and the Complexity of Modal Team Logic]{Canonical Models and the \texorpdfstring{\\}{}Complexity of Modal Team Logic}
\titlecomment{{\lsuper*}This article is an extension of the conference publication~\cite{mtl-csl}.}

\author[M.~Lück]{Martin Lück}
\address{Leibniz Universität Hannover, Institut für Theoretische Informatik, Appelstraße 4, 30167 Hannover}
\email{lueck@thi.uni-hannover.de}

\begin{abstract}
 \noindent We study modal team logic MTL, the team-semantical extension of modal logic ML closed under Boolean negation.
 Its fragments, such as modal dependence, independence, and inclusion logic, are well-understood.
 However, due to the unrestricted Boolean negation, the satisfiability problem of full MTL has been notoriously resistant to a complexity theoretical classification.

 In our approach, we introduce the notion of canonical models into the team-semantical setting.
 By construction of such a model, we reduce the satisfiability problem of MTL to simple model checking.
 Afterwards, we show that this approach is optimal in the sense that MTL-formulas can efficiently enforce canonicity.

 Furthermore, to capture these results in terms of complexity, we introduce a non-elementary complexity class, TOWER(poly), and prove that it contains satisfiability and validity of MTL as complete problems. % chktex 36
 We also prove that the fragments of MTL with bounded modal depth are complete for the levels of the elementary hierarchy (with polynomially many alternations).
 The respective hardness results hold for both strict or lax semantics of the modal operators and the splitting disjunction, and also over the class of reflexive and transitive frames.
\end{abstract}

\maketitle

\section{Introduction}

It is well-known that non-linear quantifier dependencies, such as $w$ depending only on $z$ in the sentence $\forall x \,\exists y\, \forall z\, \exists w \, \varphi$, cannot be expressed in first-order logic.
To overcome this restriction, logics of incomplete information such as \emph{independence-friendly logic}~\cite{hintikka_informational_1989} have been studied.
Later, Hodges~\cite{Hodges97} introduced \emph{team semantics} to provide these logics with a compositional interpretation.
The fundamental idea is to not consider single assignments to free variables, but instead whole sets of assignments, called \emph{teams}.

In this vein, Väänänen~\cite{vaananen_dependence_2007} expressed non-linear quantifier dependencies by the \emph{dependence atom} $\dep{x_1,\ldots,x_n,y}$, which intuitively states that the values of $y$ in the team functionally depend on those of $x_1,\ldots,x_n$.
Logics with numerous other non-classical atoms such as \emph{independence} $\perp$~\cite{GradelV13}, \emph{inclusion} $\subseteq$ and \emph{exclusion} $\mid$~\cite{Galliani12} have been studied since, and manifold connections to scientific areas such as statistics, database theory, physics, cryptography and social choice theory have emerged (see also Abramsky et al.~\cite{dep_book}).

\smallskip

Team semantics have also been adapted to a range of propositional~\cite{pdl,HannulaKLV16}, modal~\cite{vaananen_modal_2008}, and temporal logics~\cite{teamctl,KrebsMV018}.
Besides \emph{propositional dependence logic} $\logic{PDL}$~\cite{pdl} and \emph{modal dependence logic} $\logic{MDL}$~\cite{vaananen_modal_2008}, also propositional and modal logics of independence and inclusion have been studied~\cite{mil,Hannula2015,minc,hannula_lipics}.
Unlike in the first-order setting, the atoms such as the dependence atom range over flat formulas.
For example, the instance $\dep{p_1,\ldots,p_n,\Diamond\mathsf{unsafe}}$ of a modal dependence atom
may specify that the reachability of an unsafe state is a function of $p_1\cdots p_n$, but instead of exhibiting the explicit function, the atom only stipulates its existence.

Most team logics lack the Boolean negation, and adding it as a connective $\negg$ usually increases both the expressive power and the complexity tremendously.
The respective extensions of propositional and modal logic are called \emph{propositional team logic} $\PTL$~\cite{HannulaKLV16,YANG20171406,ptl2017} and \emph{modal team logic} $\MTL$~\cite{mueller14,mtl}.
With $\negg$, these logics can express all the non-classical atoms mentioned above, and in fact are expressively complete for their respective class of models~\cite{mtl,YANG20171406}.
For these reasons, they are both interesting and natural logics.

The expressive power of $\MTL$ is well-understood~\cite{mtl}, and a complete axiomatization was presented by the author~\cite{axiom}.
Yet the complexity of the satisfiability problem has been an open question~\cite{mueller14,mtl,Durand2016,HellaKMV17}.
Recently, certain fragments of $\MTL$ with restricted negation were shown $\AEXPPOLY$-complete using the well-known filtration method~\cite{filtration}.
In the same paper, however, it was shown that no elementary upper bound for full $\MTL$ can be established by the same approach, whereas
the best known lower bound is $\AEXPPOLY$-hardness, inherited from propositional team logic~\cite{ptl2017}.

\begin{table*}\centering
 \begin{tabular}{LCCl}
  \toprule
  \text{Logic} & \text{Satisfiability} & \text{Validity}                            & References                                 \\
  \midrule
  \logic{PDL}  & \compclass{NP}        & \NEXPTIME                                  & \cite{Lohmann2013,virtema_complexity_2017} \\ % chktex 2
  \logic{MDL}  & \NEXPTIME             & \NEXPTIME                                  & \cite{sevenster,hannula_lipics}            \\ % chktex 2
  \logic{PIL}  & \compclass{NP}        & \NEXPTIME\text{-hard}, \text{in }\Uppi^E_2 & \cite{Hannula2015}                         \\ % chktex 2
  \logic{MIL}  & \NEXPTIME             & \Uppi^E_2\text{-hard}                      & \cite{mil,hannula_arxiv}                   \\ % chktex 2
  \logic{PInc} & \compclass{EXPTIME}   & \text{co-}\compclass{NP}                   & \cite{Hannula2015}                         \\ % chktex 2
  \logic{MInc} & \compclass{EXPTIME}   & \text{co-}\NEXPTIME\text{-hard}            & \cite{HellaKMV15}                          \\ % chktex 2
  \PTL         & \AEXPPOLY             & \AEXPPOLY                                  & \cite{HannulaKLV16,ptl2017}                \\ % chktex 2
  \MTL_k       & \ATIMEALT{\exp_{k+1}}            & \ATIMEALT{\exp_{k+1}}                                 & Theorem~\ref{thm:mtl-main}                 \\
  \MTL         & \TOWERPOLY            & \TOWERPOLY                                 & Theorem~\ref{thm:mtl-main}                 \\
  \bottomrule
 \end{tabular}
 \bigskip
 \caption{Complexity landscape of propositional and modal logics of dependence ($*\sfD\sfL$), independence ($*\mathsf{IL}$), inclusion ($*\mathsf{Inc}$) and team logic ($*\mathsf{TL}$).  Entries are completeness results unless stated otherwise.\label{tab:table}}
\end{table*}

\medskip

\textbf{Contribution.}
We show that $\MTL$ is complete for a non-elementary class we call $\TOWERPOLY$, which contains the problems decidable in a runtime that is a tower of nested exponentials of polynomial height.
Likewise, we show that the fragments $\MTL_k$ of bounded modal depth $k$ are complete for classes we call $\ATIMEALT{\exp_{k+1}}$ and which corresponds to $(k+1)$-fold exponential runtime and polynomially many alternations.
These results fill a long-standing gap in the active field of propositional and modal team logics
(see Table~\ref{tab:table}).

\smallskip

In our approach, we consider so-called \emph{canonical models}.
Loosely speaking, a canonical model satisfies every satisfiable formula in some of its submodels, and such models have been long known for, \eg, many systems of modal logic~\cite{blackburn_modal_2001}.
In Section~\ref{sec:canon}, we adapt this notion for modal logics with team semantics, and prove that such models exist for $\MTL$.
This enables us to reduce the satisfiability problem to simple model checking, albeit on models that are of non-elementary size with respect to $\size{\Phi}+k$, where $\Phi$ are the available propositional variables and $k$ is a bound on the modal depth.

Nonetheless, this approach is essentially optimal:
In Section~\ref{sec:encoding} and~\ref{sec:staircase}, we show that $\MTL$ can, in a certain sense, \emph{efficiently enforce} canonical models, that is, with formulas that are of size polynomial in $\size{\Phi} + k$.
In this vein, we then obtain the matching complexity lower bounds in Section~\ref{sec:order} and~\ref{sec:reduction}, where we encode computations of non-elementary length in such large models.

Finally, in Section~\ref{sec:variants} we extend the preliminary version of this paper~\cite{mtl-csl} and consider restrictions of $\MTL$ to specific frame classes, and to so-called \emph{strict} team-semantical connectives.

\section{Preliminaries}
The length of (the encoding of) $x$ is denoted by $\size{x}$.
We assume the reader to be familiar with alternating Turing machines~\cite{alternation} and basic complexity theory.
When a problem is hard or complete for a complexity class, in this paper we are always referring to logspace reductions.

The class $\AEXPPOLY$ (also known as $\compclass{AEXPTIME}(\mathrm{poly})$) contains the problems decidable by an alternating Turing machine in time $2^{p(n)}$ with $p(n)$ alternations, where $p$ is a polynomial.
We generalize it to capture the \emph{elementary hierarchy} as follows.

Let $\exp_0(n) \dfn n$ and $\exp_{k+1}(n) \dfn 2^{\exp_{k}(n)}$.\label{p:exp}
A function $f \colon \N \to \N$ is \emph{elementary} if it is computable in time $\bigO{\exp_k(n)}$ for some fixed $k$.
In this paper, we consider the elementary hierarchy with polynomially many alternations:

\begin{defi}
 For $k \geq 0$, $\ATIMEALT{\exp_{k}}$ is the class of problems decidable by an alternating Turing machine with at most $p(n)$ alternations and runtime at most $\exp_k(p(n))$, for a polynomial $p$.
\end{defi}
Note that setting $k = 0$ or $k = 1$ yields the classes $\PSPACE$ and $\AEXPPOLY$, respectively~\cite{alternation}.
Schmitz~\cite{tower} proposed the following non-elementary class that contains $\ATIMEALT{\exp_{k}}$ for all $k$.
\begin{defiC}[\cite{tower}]\label{def:tower1}
 $\TOWER$ is the class of problems decidable by a deterministic Turing machine in time (or equivalently, space) $\exp_{f(n)}(1)$ for an elementary function $f$.
\end{defiC}

A suitable notion of reduction for this class is the following:
An \emph{elementary reduction} from $A$ to $B$ is an elementary function $f$ such that $x \in A \Leftrightarrow f(x) \in B$.
$A \leqelem B$ means that there exists an elementary reduction from $A$ to $B$.

\begin{propC}[\cite{tower}]
 $\TOWER$ is closed under $\leqelem$.
\end{propC}

The next class results from imposing a polynomial bound on the number of exponentials in the definition of $\TOWER$, which leads to a strict subclass.
\begin{defi}\label{def:towerpoly}
 $\TOWERPOLY$ is the class of problems that are decided by a deterministic Turing machine in time (or equivalently, space) $\exp_{p(n)}(1)$ for some polynomial $p$.
\end{defi}

The reader may verify that both $\ATIMEALT{\exp_{k}}$ and $\TOWERPOLY$ are closed under $\leqpm$ and $\leqlogm$.
Furthermore, by the time hierarchy theorem, $\TOWERPOLY \subsetneq \TOWER$.

To the author's best knowledge, neither has been explicitly considered before.
However, candidates for natural complete problems exist.
Although not proved complete, several problems in $\TOWERPOLY$ are provably non-elementary, such as the satisfiability problem of separated first-order logic~\cite{separated},
the equivalence problem for star-free expressions~\cite{stockmeyer}, or the first-order theory of finite trees~\cite{ComptonH90}, to only name a few.
We refer the reader also to the survey of Meyer~\cite{Meyer74}.

Another example is the two-variable fragment of first-order team logic, $\FO^2(\negg)$.
It is related to $\MTL$ in the same fashion as classical two-variable logic $\FO^2$ to $\ML$.
By reduction from $\MTL$ to $\FO^2(\negg)$, the satisfiability problem of $\FO^2(\negg)$ is $\TOWERPOLY$-complete problems as a corollary of our main result, Theorem~\ref{thm:mtl-main}, while its fragments $\FO^2_k(\negg)$ of bounded quantifier rank $k$ are $\ATIMEALT{\exp_{k+1}}$-hard~\cite{fo2-mfcs}.

Next, we justify why we use only $\leqlogm$-reductions (or polynomial time reductions in general) in this paper instead of $\leqelem$.

\begin{prop}
 Every problem that is $\leqelem$-complete for $\TOWERPOLY$ is also $\leqelem$-complete for $\TOWER$.
\end{prop}
\begin{proof}
 Clearly, $\TOWERPOLY \subseteq \TOWER$.
 For the lower bound, let $A$ be $\leqelem$-complete for $\TOWERPOLY$, and let $B \in \TOWER$ be arbitrary.
 $B$ is decidable in time $\exp_{r(n)}(1)$ for some elementary $r$.
 Define the set
 $C \dfn \{ x\#0^{r(\size{x})} \mid x \in B \}$.
 First, we show that $C \in \TOWERPOLY$.
 Consider the algorithm that first checks if the input $z$ is of the form $x\#0^*$, computes $r(\size{x})$ in elementary time, checks whether $z = x\#0^{r(\size{x})}$, and then whether $x \in B$.
 The first two steps clearly take elementary time in $n$, where $n \dfn \size{x\#0^{r(\size{x})}}$, and the final step runs in time $\exp_{r(\size{x})}(1) \leq \exp_{n}(1)$.

 By assumption, $C \leqelem A$ via an elementary reduction $f$.
 But clearly also $B \leqelem C$ by the elementary reduction $g \colon x \mapsto x\#0^{r(\size{x})}$.
 As a consequence, the function $h \dfn f \circ g$ is a reduction from $B$ to $A$.
 $h$ is computable in time $\exp_{k_1}(\exp_{k_2}(n)) = \exp_{k_1+k_2}(n)$ for fixed $k_1,k_2 \geq 0$ depending on $f$ and $g$, and hence again elementary.
\end{proof}

\begin{cor}
 $\TOWERPOLY$ is not closed under $\leqelem$-reductions.
\end{cor}
\begin{proof}
 Suppose $\TOWERPOLY$ is closed under $\leqelem$-reductions, and let $A$ be any problem complete for $\TOWERPOLY$ (such $A$ exists; see also our main result, Theorem~\ref{thm:mtl-main}).
 By the previous proposition, then $\TOWER \subseteq \TOWERPOLY$, contradiction.
\end{proof}

\section{Modal team logic}%
\label{sec:mtl-prelim}

We fix a countably infinite set $\PS$ of propositional symbols.
\emph{Modal team logic} $\MTL$, introduced by Müller~\cite{mueller14}, extends classical modal logic $\ML$.
Formulas of classical $\ML$ are built following the grammar
\begin{align*}
 \alpha & \ddfn \neg \alpha \mid \alpha  \land \alpha  \mid \alpha  \lor \alpha \mid \Box\alpha \mid \Diamond\alpha \mid p \mid \top\text{,}
\end{align*}
where $p \in \PS$ and $\top$ is constant truth.
$\MTL$ extends $\ML$ by the grammar
\begin{align*}
 \varphi & \ddfn \negg \varphi \mid \varphi \land \varphi \mid \varphi \lor \varphi \mid  \Box \varphi \mid \Diamond \varphi \mid  \alpha\text{,}
\end{align*}
where $\alpha$ denotes an $\ML$-formula.

The set of propositional variables occurring in a formula $\varphi \in \MTL$ is $\Prop(\varphi)$.
We use the common abbreviations $\bot \dfn \neg \top$, $\alpha \imp \beta \dfn \neg \alpha \lor \beta$ and $\alpha \leftrightarrow \beta \dfn (\alpha\land\beta)\lor(\neg\alpha\land\neg\beta)$.
For easier distinction, we have classical $\ML$-formulas denoted by $\alpha,\beta,\gamma,\ldots$ and reserve $\varphi,\psi,\vartheta,\ldots$ for general $\MTL$-formulas.

The \emph{modal depth} $\md(\varphi)$ of a formula $\varphi$ is recursively defined:
\begin{alignat*}{3}
  & \md(p)                  &  & \dfn \md(\top)              &  & \dfn 0                                      \\
  & \md(\negg \varphi)      &  & \dfn \md(\neg\varphi)       &  & \dfn \md(\varphi)                           \\
  & \md(\varphi \land \psi) &  & \dfn \md(\varphi \lor \psi) &  & \dfn \maxsf\{\md(\varphi),\md(\psi)\} \\
  & \md(\Diamond \varphi)   &  & \dfn \md(\Box\varphi)       &  & \dfn \md(\varphi) + 1
\end{alignat*}
$\ML_k$ and $\MTL_k$ are the fragments of $\ML$ and $\MTL$ with modal depth  $\leq k$, respectively.
If the propositions are restricted to a fixed set $\Phi \subseteq \PS$ as well, then the fragment is denoted by $\ML_k^\Phi$, or $\MTL^\Phi_k$, respectively.

\medskip

Let $\Phi \subseteq \PS$ be finite.
A \emph{Kripke structure} (over $\Phi$) is a tuple $\calK = (W, R, V)$, where $W$ is a set of \emph{worlds} or \emph{points}, $(W,R)$ is a directed graph called \emph{frame}, and $V \colon \Phi \to \pow{W}$ is the \emph{valuation}, with $\pow{X}$ being the power set of $X$.

Occasionally, by slight abuse of notation, we use the inverse mapping $V^{-1} \colon W \to \pow{\Phi}$ defined by $V^{-1}(w) \dfn \{ p \in \Phi \mid w \in V(p)\}$ instead of $V$, \ie, the set of propositions that are true in a given world.
If $w \in W$, then $(\calK,w)$ is called \emph{pointed structure}.
$\ML$ is evaluated on pointed structures in the classical Kripke semantics.

\medskip

By contrast, $\MTL$ is evaluated on pairs $(\calK,T)$ called \emph{structures with teams}, where $\calK = (W,R,V)$ is a Kripke structure and $T \subseteq W$ is called \emph{team} (in $\calK$).
Every team $T$ has an \emph{image} $RT \dfn \{ v \mid w \in T, (w,v) \in R \}$, and for $w \in W$, we simply write $Rw$ instead of $R\{w\}$.
$R^{i}T$ is inductively defined as $R^0T \dfn T$ and $R^{i+1}T \dfn RR^{i}T$.
An $R$-\emph{successor team} (or simply successor team) of $T$ is a team $S$ such that $S \subseteq RT$ and $T \subseteq R^{-1}S$, where $R^{-1} \dfn \{(v,w)\mid(w,v) \in R\}$.
Intuitively, $S$ is formed by picking at least one $R$-successor of every world in $T$.
The semantics of $\MTL$ can now be defined as follows.
\footnote{Often, the "atoms" of $\MTL$ are restricted to literals $p,\neg p$ instead of $\ML$-formulas $\alpha$.
 However, this implies a restriction to formulas in negation normal form, and both definitions are equivalent due to the \emph{flatness} property of $\ML$ (cf.~\cite[Proposition 2.2]{mtl}).}
\begin{alignat*}{3}
  & (\calK, T) \vDash \alpha           &  & \Leftrightarrow\;\forall w \in T \colon (\calK, w) \vDash \alpha \; \text{ if } \alpha \in \ML\text{, and otherwise as}                           \\
  & (\calK, T) \vDash \negg \psi       &  & \Leftrightarrow\;(\calK,T) \nvDash \psi\text{,}                                                                                                   \\
  & (\calK,T) \vDash \psi \land \theta &  & \Leftrightarrow\;(\calK,T) \vDash \psi \text{ and }(\calK,T) \vDash \theta\text{,}                                                                \\
  & (\calK,T) \vDash \psi \lor \theta  &  & \Leftrightarrow\;\exists S, U \subseteq T \text{ such that }T = S \cup U\text{, }(\calK,S) \vDash \psi\text{, and }(\calK,U)\vDash \theta\text{,} \\
  & (\calK,T)\vDash \Diamond \psi      &  & \Leftrightarrow\;(\calK, S)\vDash \psi \text{ for some successor team }S\text{ of }T\text{,}                                                      \\
  & (\calK,T)\vDash \Box\psi           &  & \Leftrightarrow\;(\calK,RT) \vDash \psi\text{.}
\end{alignat*}
We often omit $\calK$ and write only $T \vDash \varphi$ (for team semantics) or $w \vDash \alpha$ (for Kripke semantics).

\smallskip

An $\MTL$-formula $\varphi$ is \emph{satisfiable} if it is true in some structure with team over $\Prop(\varphi)$, which is then called a \emph{model} of $\varphi$.
Analogously, $\varphi$ is \emph{valid} if it is true in every structure with team (over $\Prop(\varphi)$).
For a logic $\calL$, the sets of all satisfiable resp.\ valid formulas of $\calL$ are $\SAT(\calL)$ and $\VAL(\calL)$, respectively.

In the literature on team semantics, the empty team is usually excluded in the above definition, since most $\negg$-free logics with team semantics have the \emph{empty team property}, \ie, the empty team satisfies every formula~\cite{vaananen_modal_2008,mil,minc}.
However, this distinction is unnecessary for $\MTL$:
$\varphi$ is satisfiable iff $\top \lor \varphi$ is satisfied by some non-empty team\footnote{Note that $\top \lor \varphi$ is not a tautology in general, since $\lor$ is not the Boolean disjunction.
Rather, $\top \lor \varphi$ existentially quantifies a subteam where $\varphi$ holds.
In fact, $\top \lor \varphi$ is a tautology if and only if $\varphi$ holds in the empty team.}, and $\varphi$ is satisfied by some non-empty team iff $\negg\bot \land \varphi$ is satisfiable.

\smallskip

The modality-free fragment $\MTL_0$ syntactically coincides with \emph{propositional team logic} $\PTL$~\cite{HannulaKLV16,ptl2017,YANG20171406}.
The usual interpretations of the latter, \ie, sets of Boolean assignments, can easily be represented as teams in Kripke structures.
For this reason, we treat $\PTL$ and $\MTL_0$ as identical in this article.

Note that the connectives $\lor$, $\imp$ and $\neg$ are not the Boolean disjunction, implication and negation, except on singleton teams, which correspond to Kripke semantics.
Using $\land$ and $\negg$ however,
we can define team-wide Boolean disjunction $\varphi_1\ovee\varphi_2 \dfn \negg(\negg\varphi_1\land\negg\varphi_2)$ and material implication $\varphi_1 \timp \varphi_2 \dfn \negg \varphi_1 \ovee \varphi_2$.

The notation $\Box^i\varphi$ is defined via $\Box^0\varphi \dfn \varphi$ and $\Box^{i+1}\varphi \dfn \Box\Box^i\varphi$, and analogously for $\Diamond^i\varphi$.
To express that at least one element of a team satisfies $\alpha \in \ML$, we use $\E\alpha \dfn \negg\neg\alpha$.

\medskip

$\MTL$ can express the \emph{(extended) dependence atom} $\dep{\alpha_1,\ldots,\alpha_{n-1},\alpha_n}$ of (extended) modal dependence logic~\cite{vaananen_modal_2008,EbbingHMMVV13}, which states that the truth value of $\alpha_n$ is a function of the truth values of $\alpha_1,\ldots,\alpha_{n-1}$, where $\alpha_1,\ldots,\alpha_n \in \ML$.
It is definable in $\MTL$ as
$\negg\left[\top \lor \negg \left(\bigwedge_{i = 1 }^{n-1}\dep{\alpha_i} \timp \dep{\alpha_n}\right) \right]$,
where $\dep{\alpha} \dfn \alpha \ovee \neg \alpha$ is the \emph{constancy atom}, stating that the truth value of $\alpha \in \ML$ is constant throughout the team.

\medskip

The well-known \emph{bisimulation} relation $\bis^\Phi_k$ fundamentally characterizes the expressive power of modal logic~\cite{blackburn_modal_2001} and plays a key role in our results.

\begin{defi}\label{def:point-bisim}
 Let $\Phi \subseteq \PS$ and $k\geq 0$.
 For $i \in \{1,2\}$, let $(\calK_i,w_i)$ be a pointed structure, where $\calK_i = (W_i,R_i,V_i)$.
 Then
 $(\calK_1,w_1)$ and $(\calK_2,w_2)$ are \emph{$(\Phi,k)$-bisimilar}, in symbols $(\calK_1,w_1) \bis^\Phi_k (\calK_2,w_2)$, if
 \begin{itemize}
  \item $\forall p \in \Phi \colon w_1 \in V_1(p) \Leftrightarrow w_2 \in V_2(p)$,
  \item and if $k > 0$,
        \begin{itemize}
         \item $\forall v_1 \in R_1w_1 \colon \exists v_2 \in R_2w_2 \colon (\calK_1,v_1)\bis^\Phi_{k-1} (\calK_2,v_2)$ (\emph{forward} condition),
         \item $\forall v_2 \in R_2w_2 \colon \exists v_1 \in R_1w_1 \colon (\calK_1,v_1)\bis^\Phi_{k-1} (\calK_2,v_2)$ (\emph{backward} condition).\qedhere
        \end{itemize}
 \end{itemize}
\end{defi}

\noindent
So-called \emph{characteristic formulas} or \emph{Hintikka formulas} capture the essence of the bisimulation relation in the following sense:

\begin{propC}[{\cite[Theorem 32]{handbook}}]\label{prop:hintikka}
 Let $\Phi\subseteq \PS$ be finite, $k \geq 0$, and let $(\calK,w)$ be a pointed structure.
 Then there is a formula $\zeta \in \ML^\Phi_k$ such that for all pointed structures $(\calK',w')$ we have $(\calK',w') \vDash \zeta$ if and only if $(\calK,w) \bis^\Phi_k (\calK',w')$.
\end{propC}

\smallskip

The notion of bisimulation was lifted to team semantics by Hella et al.~\cite{HellaLSV14,mil,mtl}:

\begin{defi}\label{def:team-bisim}
 Let $\Phi \subseteq \PS$ and $k\geq 0$.
 For $i \in \{1,2\}$, let $(\calK_i,T_i)$ be a structure with team.
 Then $(\calK_1,T_1)$ and $(\calK_2,T_2)$ are \emph{$(\Phi,k)$-team-bisimilar}, written $(\calK_1,T_1) \bis^\Phi_k (\calK_2,T_2)$, if
 \begin{itemize}
  \item $\forall w_1 \in T_1 \colon \exists w_2 \in T_2 \colon (\calK_1,w_1) \bis^\Phi_k (\calK_2,w_2)$,
  \item $\forall w_2 \in T_2 \colon \exists w_1 \in T_1 \colon (\calK_1,w_1) \bis^\Phi_k (\calK_2,w_2)$.\qedhere
 \end{itemize}
\end{defi}

\noindent
If no confusion can arise, we will also refer to teams $T_1,T_2$ that are $(\Phi,k)$-team-bisimilar simply as $(\Phi,k)$-bisimilar.
Throughout the paper, we will make use of the following characterizations of bisimilarity.

\begin{prop}\label{prop:ml-bisim-types}
 Let $\Phi \subseteq \PS$ be finite, and $k \geq 0$.
 For $i \in \{1,2\}$, let $(\calK_i,w_i)$ be a pointed structure, where $\calK_i  = (W_i,R_i,V_i)$.
 The following statements are equivalent:
 \begin{enumerate}
  \item $\forall \alpha \in \ML^\Phi_k \colon (\calK_1,w_1) \vDash \alpha \Leftrightarrow (\calK_2,w_2) \vDash \alpha$,
  \item $(\calK_1,w_1) \bis^\Phi_k (\calK_2,w_2)$,
  \item $(\calK_1,\{w_1\})\bis^\Phi_k (\calK_2,\{w_2\})$,
 \end{enumerate}
 \quad\, and if $k > 0$,
 \begin{enumerate}\setcounter{enumi}{3}
  \item $(\calK_1,w_1)\bis^\Phi_0(\calK_2,w_2)$ and $(\calK_1,R_1w_1) \bis^\Phi_{k-1} (\calK_2,R_2w_2)$.\label{enum:image-bisim}
 \end{enumerate}
\end{prop}
\begin{proof}
 (1) $\Leftrightarrow$ (2) is a standard result (\cite[Theorem 32]{handbook}).
 (2) $\Leftrightarrow$ (3) follows from Definition~\ref{def:team-bisim}.
 For $k > 0$, we show that (2) + (3) implies (4).
 Clearly, $(\calK_1,w_1) \bis^\Phi_0 (\calK_2,w_2)$ follows from (2).
 Due to Hella et al.~\cite[Lemma 3.3]{HellaLSV14}, (3) implies $(\calK_1,R_1w_1) \bis^\Phi_{k-1} (\calK_2,R_2w_2)$.

 Finally, we show (4) $\Rightarrow$ (2).
 Suppose $(\calK_1,w_1) \bis^\Phi_0 (\calK_2,w_2)$ and $(\calK_1,R_1w_1) \bis^\Phi_{k-1} (\calK_2,R_2w_2)$.
 Then to show $(\calK_1,w_1) \bis^\Phi_k (\calK_2,w_2)$, it is sufficient to prove the \emph{forward} and \emph{backward} conditions of Definition~\ref{def:point-bisim}.
 Suppose $v_1 \in R_1w_1$.
 Since $(\calK_1,R_1w_1) \bis^\Phi_{k-1} (\calK_2,R_2w_2)$, by Definition~\ref{def:team-bisim} there exists $v_2\in R_2w_2$ such that $(\calK_1,v_1) \bis^\Phi_{k-1} (\calK_2,v_2)$, proving the \emph{forward} condition.
 The \emph{backward} condition is symmetric.
\end{proof}

As a consequence, the \emph{forward} and \emph{backward} condition from Definition~\ref{def:point-bisim} can be equivalently stated in terms of team-bisimilarity of the respective image teams.
A similar characterization exists for team-bisimilarity:

\begin{prop}\label{prop:mtl-bisim-types}
 Let $\Phi \subseteq \PS$ be finite, and $k \geq 0$.
 Let $(\calK_i,T_i)$ be a structure with team for $i \in \{1,2\}$.
 Then the following statements are equivalent:
 \begin{enumerate}
  \item $\forall \alpha \in \ML^\Phi_k \colon (\calK_1,T_1) \vDash \alpha \Leftrightarrow (\calK_2,T_2) \vDash \alpha$,
  \item $\forall \varphi \in \MTL^\Phi_k \colon (\calK_1,T_1) \vDash \varphi \Leftrightarrow (\calK_2,T_2) \vDash \varphi$,
  \item $(\calK_1,T_1) \bis^\Phi_k (\calK_2,T_2)$.
 \end{enumerate}
\end{prop}
\begin{proof}
 The above statements are all true if $T_1 = T_2 = \emptyset$, and they are all false if exactly one of the teams is empty, since a team $T$ satisfies the $\ML$-formula $\bot$ precisely if $T = \emptyset$.
 For this reason, we can assume that both $T_1$ and $T_2$ are non-empty.

 By Kontinen et al.~\cite[Proposition 3.10]{mtl}, for non-empty $T_1,T_2$ there exists an $\MTL^\Phi_k$-formula $\varphi$ that is true in $(\calK_1,T_1)$, but holds in $(\calK_2,T_2)$ if and only if $(\calK_1,T_1) \bis^\Phi_k (\calK_2,T_2)$.
 This immediately proves (2) $\Rightarrow$ (3).
 The direction (3) $\Rightarrow$ (2) is due to Kontinen et al.~\cite[Proposition 2.8]{mtl} as well.

 Finally, (1) $\Leftrightarrow$ (2) follows from the fact that $\ML^\Phi_k \subseteq \MTL^\Phi_k$, and that conversely every $\MTL^\Phi_k$-formula is equivalent to a formula of the form
 \[
  \bigovee_{i=1}^n \Big(\alpha_{i} \land \bigwedge_{j=1}^{m_i} \E \beta_{i,j} \Big)\text{,}
 \]
 where $\{\alpha_1, \ldots, \alpha_n, \beta_{1,1}, \ldots, \beta_{n,m_n}\} \subseteq \ML^\Phi_k$ (see~\cite[Theorem 5.2]{axiom} or~\cite[p.\,11]{mtl}).
\end{proof}

Note that the analog of condition~\ref{enum:image-bisim} in Proposition~\ref{prop:ml-bisim-types} for team bisimulation is not equivalent: It is possible that $(\calK_1,T_1) \bis^\Phi_0 (\calK_2,T_2)$ and
$(\calK_1,R_1T_1) \bis^\Phi_{k-1} (\calK_2,R_2T_2)$, but $(\calK_1,T_1) \not\bis^\Phi_k (\calK_2,T_2)$.

\section{Types and canonical models}%
\label{sec:canon}

Many modal logics admit a "universal" model, also called \emph{canonical model}.
The defining property of a canonical model is that it simultaneously witnesses all satisfiable (sets of) formulas in some of its points.
These models are a popular tool for proving the completeness of manifold systems of modal logics; for the explicit construction of such a model for $\ML$, consult, \eg, Blackburn et al.~\cite[Section 4.2]{blackburn_modal_2001}.

Unfortunately, any canonical model for $\ML$ is necessarily infinite, and consequently impractical for complexity theoretic considerations.
Instead, we use so-called \emph{$(\Phi,k)$-canonical models} for finite $\Phi \subseteq \PS$ and $k \in \N$; as the name suggests they are canonical for the fragment $\ML^\Phi_k$.
While these models are finite, by Proposition~\ref{prop:ml-bisim-types} their size is at least the number of equivalence classes of $\bis^\Phi_k$.
We call the equivalence classes of $\bis^\Phi_k$ \emph{types}.

A first issue arises since types are then proper classes, and in team semantics, we need to speak about \emph{sets} of types.
For this reason, we begin this section by defining types on proper set-theoretic grounds, by indentifying the type of a point with the set of formulas that are true in it, which is a standard approach in first-order model theory.

\subsection{Types}%
\label{sec:types}

\begin{defi}
 A set $\tau \subseteq \ML^\Phi_k$ is a \emph{$(\Phi,k)$-type} if it is satisfiable and for all $\alpha \in \ML^\Phi_k$ contains either $\alpha$ or $\neg\alpha$.
 The $(\Phi,k)$-type of a pointed structure $(\calK,w)$ is
 \[
    \type{\calK,w}^\Phi_k \dfn \big\{ \alpha \in \ML^\Phi_k \mid (\calK,w) \vDash \alpha \big\}.
 \]
\end{defi}

The set of all $(\Phi,k)$-types is $\Delta^\Phi_k$.
Given a team $T$ in $\calK$, the types in $T$ are
\[
    \type{\calK,T}^\Phi_k \dfn \big\{ \type{\calK,w}^\Phi_k \mid w \in T\big\}.
\]
The following assertions ascertain that the above definition of types properly reflects the bisimulation relation.

\begin{prop}\label{prop:types}
 Let $\Phi \subseteq \PS$ and $k \geq 0$.
 Then
 \begin{enumerate}
  \item The unique $(\Phi,k)$-type satisfied by $(\calK,w)$ is $\type{\calK,w}^\Phi_k$.
  \item $(\calK,w) \bis^\Phi_k (\calK',w')$ if and only if $\type{\calK,w}^\Phi_k = \type{\calK',w'}^\Phi_k$.
  \item $(\calK,T) \bis^\Phi_k (\calK',T')$ if and only if $\type{\calK,T}^\Phi_k = \type{\calK',T'}^\Phi_k$.
 \end{enumerate}
\end{prop}
\begin{proof}
 Property (1) is straightforward: two distinct types $\tau,\tau'$ satisfied by $(\calK,w)$ differ in some $\alpha \in \ML^\Phi_k$.
 But then $(\calK,w) \vDash \alpha,\neg\alpha$, contradiction.
 Property (2) immediately follows from Proposition~\ref{prop:ml-bisim-types}.
 For (3), first consider "$\Rightarrow$".
 Due to symmetry, we only show that $(\calK,T) \bis^\Phi_k (\calK',T')$ implies $\type{\calK,T}^\Phi_k \subseteq \type{\calK',T'}^\Phi_k$.
 Hence suppose $\tau \in \type{\calK,T}^\Phi_k$.
 Then there exists $w \in T$ of type $\type{\calK,w}^\Phi_k = \tau$.
 By Definition~\ref{def:team-bisim}, there is $w' \in T'$ with $(\calK,w)\bis^\Phi_k (\calK',w')$.
 Then $\type{\calK',w'}^\Phi_k = \tau \in \type{\calK',T'}^\Phi_k$ by property (2).
 The direction "$\Leftarrow$" of (3) is shown analogously.
\end{proof}

It is unsurprising that the type of a point $w$ is determined solely by the propositions in $w$ and the types in the image $Rw$.
In other words, all pointed structures of type $\tau$ satisfy the same propositions in their roots, viz.\ $\tau \cap \Phi$, and have the same types contained in their image teams.
Regarding the latter, we define
$\calR\tau \dfn \big\{ \tau' \in \Delta^\Phi_k \mid \{ \alpha \mid \Box\alpha \in \tau \} \subseteq \tau'\big\}$, given a $(\Phi,k+1)$-type $\tau$.
Intuitively, $\calR\tau$ is the set of $(\Phi,k)$-types that occur in the image team of a world of type $\tau$.

The following proposition shows that types are indeed uniquely determined by the above constituents:

\begin{prop}\label{prop:det-types}
 Let $\Phi \subseteq \PS$ be finite and $k \geq 0$.
 \begin{enumerate}
  \item $\type{w}^\Phi_k \cap \Phi = V^{-1}(w) \cap \Phi$ and $\type{Rw}^\Phi_k = \calR\type{w}^\Phi_{k+1}$, for all pointed structures $(W,R,V,w)$.\vspace{1mm}
  \item The mapping $h \colon \tau \mapsto \tau \cap \Phi$ is a bijection from $\Delta^\Phi_0$ to $\pow{\Phi}$.
  \item The mapping $h \colon \tau \mapsto (\tau \cap \Phi, \calR\tau)$ is a bijection from $\Delta^\Phi_{k+1}$ to $\pow{\Phi} \times \pow{\Delta^\Phi_k}$.
 \end{enumerate}
\end{prop}
\begin{proof}
  See the appendix.
\end{proof}

\begin{lem}\label{lem:types-equal-const}
 Let $(W,R,V,w)$ be a pointed structure.
 \begin{enumerate}
  \item If $\tau \in \Delta^\Phi_0$, then $\type{w}^\Phi_0 = \tau$ if and only if $V^{-1}(w) = \tau \cap \Phi$.
  \item If $\tau \in \Delta^\Phi_{k+1}$, then $\type{w}^\Phi_{k+1} = \tau$ if and only if $V^{-1}(w) = \tau \cap \Phi$ and $\type{Rw}^\Phi_k = \calR\tau$.
 \end{enumerate}
\end{lem}
\begin{proof}
 The direction "$\Rightarrow$" of 1.\ and 2.\ follows directly from Proposition~\ref{prop:det-types}.
 Moreover, we prove "$\Leftarrow$" only for statement 2., as the proof is analogous for 1.

 Suppose that there are $\tau,\tau' \in \Delta^\Phi_{k+1}$ such that $V^{-1}(w) = \tau \cap \Phi$ and $\type{Rw}^\Phi_k = \calR\tau$, but $\type{w}^\Phi_{k+1} = \tau'$.
 Then, by "$\Rightarrow$", we have $V^{-1}(w) = \tau' \cap \Phi$ and $\type{Rw}^\Phi_k = \calR\tau'$ as well.
 In other words, $\tau \cap \Phi = \tau' \cap \Phi$ and $\calR\tau = \calR\tau'$.
 However, since the mapping $h\colon \tau \mapsto (\tau \cap \Phi, \calR\tau)$ is bijective according to Proposition~\ref{prop:det-types}, we have $\tau = \tau' = \type{w}^\Phi_{k+1}$.
\end{proof}

We are now ready to state the formal definition of canonicity by the notion of types:

\begin{defi}
 A structure with team $(\calK,T)$ is \emph{$(\Phi,k)$-canonical} if $\type{\calK,T}^\Phi_k = \Delta^\Phi_k$.
\end{defi}

In the following, we often omit $\Phi$ and $\calK$ and instead write $\type{w}_k$ and $\type{T}_k$, respectively, and simply say that $T$ is $(\Phi,k)$-canonical if $\calK$ is clear.

\subsection{Canonical models in team semantics}

It is a standard result that for every $\Phi$ and $k \geq 0$ there exists a $(\Phi,k)$-canonical model~\cite{blackburn_modal_2001}, or in other words, that the logic $\ML^\Phi_k$ admits canonical models.

We will show that, given a $(\Phi,k)$-canonical model $\calK$, every satisfiable $\MTL^\Phi_k$-formula can be satisfied in some team of $\calK$ as well, despite $\MTL$ being significantly more expressive than $\ML$~\cite{mtl}.
In other words, the canonical models for $\MTL^\Phi_k$ and $\ML^\Phi_k$ coincide:

\begin{thm}\label{thm:canonical}
 Let $(\calK,T)$ be $(\Phi,k)$-canonical and $\varphi \in \MTL^\Phi_k$.
 Then $\varphi$ is satisfiable if and only if $(\calK,T') \vDash \varphi$ for some $T' \subseteq T$.
\end{thm}
\begin{proof}
 Assume $(\calK,T)$ and $\varphi$ are as above.
 As the direction from right to left is trivial, suppose that $\varphi$
 is satisfiable, \ie, has a model $(\hat{\calK},\hat{T})$.
 As a team in $\calK$ that satisfies $\varphi$, we define
 \[
  T' \dfn \Set{ w \in T | \type{\calK,w}^\Phi_k \in \type{\hat{\calK},\hat{T}}^\Phi_k }\text{.}
 \]
 By Proposition~\ref{prop:mtl-bisim-types} and~\ref{prop:types}, it suffices to prove $\type{\hat{\calK},\hat{T}}^\Phi_k = \type{\calK,T'}^\Phi_k$.
 Moreover, the direction "$\supseteq$" is clear by definition.
 As $T$ is $(\Phi,k)$-canonical, for every $\tau \in \type{\hat{\calK},\hat{T}}^\Phi_k$ there exists a world $w \in T$ of type $\tau$.
 Consequently, $\type{\hat{\calK},\hat{T}}^\Phi_k \subseteq \type{\calK,T'}^\Phi_k$.
\end{proof}

\medskip

How large is a $(\Phi,k)$-canonical model at least?
The number of types is captured by the function $\exp^*_k$, defined by
\[
 \exp^*_0(n) \dfn n\qquad \qquad \exp^*_{k+1}(n) \dfn n \cdot 2^{\exp^*_k(n)}\text{.}
\]

\begin{prop}\label{prop:number-of-types}
 $\size{\Delta^\Phi_k} = \exp^*_k\big(2^\size{\Phi}\big)$ for all $k \geq 0$ and finite $\Phi \subseteq \PS$.
\end{prop}
\begin{proof}
 By induction on $k$.
 For the base case $k = 0$, this follows from Proposition~\ref{prop:det-types}, as there is a bijection between $\Delta^\Phi_0$ and $\pow{\Phi}$ and $\exp^*_0\big(2^\size{\Phi}\big) = 2^\size{\Phi} = \size{\Delta^\Phi_0}$.

 We proceed with the inductive step, \ie, $k + 1$.
 First note that by induction hypothesis
 \[
  \exp^*_{k+1}\big(2^\size{\Phi}\big) = 2^\size{\Phi} \cdot 2^{\exp^*_k(2^\size{\Phi})} = \size{\pow{\Phi} \times \pow{\Delta^\Phi_k}}\text{.}
 \]
 Again, there exists a bijection from $\Delta^\Phi_{k+1}$ to $\pow{\Phi} \times \pow{\Delta^\Phi_k}$ by Proposition~\ref{prop:det-types}.
\end{proof}

Next, we present an algorithm that solves the satisfiability and validity problems of $\MTL_k$ by computing a canonical model.
Let us first explicate this construction in a lemma.

\begin{lem}\label{lem:construct-canon}
 There is an algorithm that, given $\Phi \subseteq \PS$ and $k \geq 0$, computes a $(\Phi,k)$-canonical model in time polynomial in $\size{\Delta^\Phi_k}$.
\end{lem}
\begin{proof}
 The idea is to construct sets $L_0 \cup L_1 \cup \cdots \cup L_k$ of worlds in stage-wise manner such that $L_i$ is $(\Phi,i)$-canonical.
 For $L_0$, we simply add a world $w$ for each $\Phi' \in \pow{\Phi}$ such that $V^{-1}(w) = \Phi'$.
 For $i > 0$, we iterate over all $L' \in \pow{L_{i-1}}$ and $\Phi' \in \pow{\Phi}$ and insert a new world $w$ into $L_i$ such that $L'$ is the image of $w$ and such that again $V^{-1}(w) = \Phi'$.
 An inductive argument based on Proposition~\ref{prop:mtl-bisim-types} and~\ref{prop:det-types} shows that $L_i$ is $(\Phi,i)$-canonical for all $i \in \{0,\ldots,k\}$.
 As $k \leq \size{\Delta^\Phi_k}$, and each $L_i$ is constructed in time polynomial in $\size{\Delta^\Phi_i} \leq \size{\Delta^\Phi_k}$, the overall runtime is polynomial in $\size{\Delta^\Phi_k}$.
\end{proof}

With the help of a small lemma, we conclude the upper bound for the satisfiability and validity problem of $\MTL$ and its fragments.

\begin{lem}\label{lem:poly-in-tower}
 For every polynomial $p$ there is a polynomial $q$ such that
 \[
    p(\exp^*_k(n)) \leq \exp_k( q((k+1)\cdot n))
 \]
 for all $k \geq 0$ and $n \geq 1$.
\end{lem}
\begin{proof}
 See the appendix.
\end{proof}

\begin{thm}\label{thm:k-membership}
 $\mathsf{SAT}(\MTL_k)$ and $\mathsf{VAL}(\MTL_k)$ are in $\ATIMEALT{\exp_{k+1}}$.
\end{thm}
\begin{proof}
 Consider the following algorithm.
 Let $\varphi\in \MTL_k$ be the input, $n \dfn \size{\varphi}$, and $\Phi \dfn \Prop(\varphi)$.
 Construct deterministically, as in Lemma~\ref{lem:construct-canon}, a $(\Phi,k)$-canonical structure $\calK = (W,R,V)$ in time $p(\size{\Delta^\Phi_k})$ for a polynomial $p$.

 By a result of Müller~\cite{mueller14}, the model checking problem of $\MTL$ is solvable by an alternating Turing machine that has runtime polynomial in $\size{\varphi} + \size{\calK}$, and alternations polynomial in $\size{\varphi}$.
 We call this algorithm as a subroutine:
 by Theorem~\ref{thm:canonical}, $\varphi$ is satisfiable (resp.\ valid) if and only if
 for at least one subteam (resp.\ all subteams) $T \subseteq W$ we have $(\calK,T) \vDash \varphi$.
 Equivalently, this is the case if and only if $(\calK,W)$ satisfies $\top \lor \varphi$ (resp.\ $\negg(\top \lor \negg \varphi)$).

 Let us turn to the overall runtime.
 $\calK$ is constructed in time polynomial in $\size{\Delta^\Phi_k} = \exp^*_k(2^\size{\Phi}) \leq \exp^*_{k+1}(\size{\Phi}) \leq \exp^*_{k+1}(n)$.
 The subsequent model checking runs in time polynomial in $\size{\calK} + n$, and hence polynomial in $\exp^*_{k+1}(n)$ as well.
 By Lemma~\ref{lem:poly-in-tower}, we obtain a total runtime of $\exp_{k+1}(q((k+2)\cdot n))$ for a polynomial $q$.
\end{proof}

The upper bound for $\MTL$ is proved identically, since $k \dfn \md(\varphi)$ is polynomial in $\size{\varphi}$.

\begin{cor}\label{cor:membership}
 $\SAT(\MTL)$ and $\VAL(\MTL)$ are in $\TOWERPOLY$.
\end{cor}

The usual definition of a canonical model is a structure that has all (infinite) maximal consistent subsets of a certain class of modal formulas as worlds (see virtually any textbook on modal logic, \eg,~\cite{blackburn_modal_2001}).
This indeed results in a finite number of worlds in the case of, say, $\ML^\Phi_k$ (cf.~\cite{Cresswell83,Cresswell1996}).
Truly finitary constructions of canonical models can be traced back to Fine~\cite{Fine75}, whose work has been extended towards various other modal systems (\eg, by Moss~\cite{Moss07}).
Furthermore, Cresswell and Hughes~\cite{Cresswell1996} used \emph{mini canonical models}, models that are "canonical" only with respect to all subformulas of a fixed $\ML$-formula, which allows them to be finite models with finite sets of formulas as worlds.

All these approaches have in common that they still are non-constructive and intended for completeness proofs.
Even computing a "mini canonical model" would not be guaranteed to be feasible enough for $\MTL$: This would require an explicit translation of a given input $\MTL^\Phi_k$-formula to a Boolean combination of $\ML^\Phi_k$-formulas first (see the proof of Proposition~\ref{prop:mtl-bisim-types}), and it is open whether there is an elementary translation for every fixed $k$ (cf.~\cite{axiom}).

In this light, our approach yields a purely constructive definition of a canonical model (in Lemma~\ref{lem:construct-canon}), which can easily be plugged into the algorithms used for the above results, and has optimal runtime up to a polynomial.

\section{Scopes and Subteam Quantifiers}%
\label{sec:encoding}

Kontinen et al.~\cite{mtl} proved that $\MTL$ is expressively complete up to bisimulation: it can define every property of teams that is $(\Phi,k)$-bisimulation invariant, that is closed under $\bis^\Phi_k$, for some finite $\Phi$ and $k$.
Two team properties that fall into this category are in fact $(\Phi,k)$-bisimilarity itself---in the sense that all worlds in a team have the same $(\Phi,k)$-type---as well as $(\Phi,k)$-canonicity.
Consequently, these properties are definable by $\MTL^\Phi_k$-formulas.
However, by a simple counting argument, formulas defining arbitrary team properties require non-elementary size \wrt $\Phi$ and $k$.

\medskip

In this section, we consider a special class of structures, and on these, define $k$-bisimilarity by a formula $\chi_k$ of polynomial size in $\Phi$ and $k$.
(From now on, we always assume some finite $\Phi \subseteq \PS$ and omit it in the notation, \ie, we write \emph{$k$-canonicity}, \emph{$k$-bisimilarity}, $\bis_k$, and so on.)
Afterwards, in Section~\ref{sec:staircase} we devise a formula $\canon_k$ of polynomial size that expresses $k$-canonicity.

\subsection{Scopes}%
\label{p:scopes}

It is natural to implement $k$-bisimilarity by mutual recursion in the spirit of Proposition~\ref{prop:ml-bisim-types}: the $(k+1)$-bisimilarity of two points $w,v$ is expressed in terms of $k$-team-bisimilarity of $Rw$ and $Rv$, and
conversely, to verify $k$-team-bisimilarity of $Rw$ and $Rv$, we proceed analogously to the \emph{forward} and \emph{backward} conditions of Definition~\ref{def:point-bisim} and reduce the problem to checking $k$-bisimilarity of pairs of points in $Rw$ and $Rv$.

$\MTL$-formulas define team properties, but we want to express a \emph{relation} between teams such as $Rw$ and $Rv$.
For this reason, we consider the "marked union" of $Rw$ and $Rv$ as a single team using the following tool.
Formally, if $\alpha \in \ML$, then the "conditioned" subteam $T_\alpha \subseteq T$ is defined as
\[
 T_\alpha \dfn \Set{ w \in T | w \vDash \alpha}\text{.}
\]
In the literature, $T_\alpha$ is also written $T \upharpoonright \alpha$~\cite{galliani_upwards_2015,Galliani16,abs-1808-00710}.
The corresponding "decoding" operator\label{p:hook}
\[
 \alpha \hook \varphi \dfn \neg \alpha \lor (\alpha \land \varphi)
\]
was introduced by Galliani~\cite{galliani_upwards_2015,Galliani16,abs-1808-00710} as well: $\alpha \hook \varphi$ is true in $T$ if and only if $T_\alpha \vDash \varphi$.

Now, instead of defining an $n$-ary relation on teams, a formula $\varphi$ can define a unary relation---a team property---parameterized by formulas $\alpha_1,\ldots,\alpha_n \in \ML$.
We emphasize this by writing $\varphi(\alpha_1,\ldots,\alpha_n)$.

It will be useful if the "markers" of the constituent teams are invariant under traversing edges in the structure.
In that case,
we call these formulas \emph{scopes}:

\begin{defi}
 Let $\calK = (W,R,V)$ be a Kripke structure.
 A formula $\alpha \in \ML$  is called a \emph{scope (in $\calK$)} if $(w,v)\in R$ implies $w \vDash \alpha \Leftrightarrow v \vDash \alpha$.
 Two scopes $\alpha,\beta$ are called \emph{disjoint (in $\calK$)} if $W_\alpha$ and $W_{\beta}$ are disjoint.
\end{defi}

To avoid interference, we always assume that scopes are formulas in $\ML^{\PS\setminus \Phi}_0$, \ie, they are always purely propositional and do not contain propositions from $\Phi$.

\begin{figure}\centering
 \begin{tikzpicture}
  \tikzset{world/.style={draw,circle,inner sep=.3mm,black,fill}}
  \tikzset{team/.style={draw,rounded corners,thick,inner sep=1.3mm}}
  \foreach \i in {1,...,8} {
    \node[world] (w\i) at ({\i*.5},1) {};
    \node[fit=(w\i)] (fitw\i) {};
   }
  \node[fit=(w4)(w5),team,fill=red,fill opacity=.2,draw=red,dotted] {};
  \node[fit=(w1)(w8),team] (team) {};
  \node[left= 2mm of fitw1] {\Large$T$};
  \node[below right = 2mm and 0 of w5] {\Large\red{$S$}};
  \draw[dashed] (1.25,.5) -- (1.25,2);
  \draw[dashed] (3.25,.5) -- (3.25,2);
  \node at (.5,1.8) {$\alpha_1$};
  \node at (4,1.8) {$\alpha_3$};
  \node at (2.25,1.8) {$\alpha_2$};
  \node at (5.85,1) {\LARGE$\Rightarrow$};

  \begin{scope}[xshift=8cm]
   \foreach \i in {1,...,8} {
     \node[world] (w\i) at ({\i*.5},1) {};
     \node[fit=(w\i)] (fitw\i) {};
    }
   \node[fit=(w4)(w5),team,fill=red,fill opacity=.2,draw=red,dotted] {};
   \draw[thick, rounded corners] (fitw1.north west)
   -- (fitw2.north east) -- ([yshift=.3mm]fitw2.south east) -- ([yshift=.3mm]fitw4.south west) -- (fitw4.north west)
   -- (fitw5.north east) -- ([yshift=.3mm]fitw5.south east) -- ([yshift=.3mm]fitw7.south west) -- (fitw7.north west)
   -- (fitw8.north east) -- ([yshift=-.2mm]fitw8.south east) -- ([yshift=-.2mm]fitw1.south west) -- cycle
   ;
   \draw[dashed] (1.25,.5) -- (1.25,2);
   \draw[dashed] (3.25,.5) -- (3.25,2);
   \node at (.5,1.8) {$\alpha_1$};
   \node at (4,1.8) {$\alpha_3$};
   \node at (2.25,1.8) {$\alpha_2$};
   \node[below right = 2mm and 0 of w5] {\Large\red{$S$}};
  \end{scope}

  \node[left= 2mm of fitw1] {\Large$T^{\alpha_2}_{\red{S}}$};
 \end{tikzpicture}
 \caption{Example of subteam selection in the scope $\alpha_2$\label{fig:shrinking}}
\end{figure}

It is desirable to be able to speak about subteams in a specific scope.
If $S$ is a team, let $T^\alpha_S \dfn T_{\neg\alpha} \cup (T_\alpha \cap S)$.
For singletons $\{w\}$, we simply write $T^\alpha_w$ instead of $T^\alpha_{\{w\}}$.
Intuitively, $T^\alpha_S$ is obtained from $T$ by "shrinking" the subteam $T_\alpha$ down to $S$ without impairing $T \setminus T_\alpha$ (see Figure~\ref{fig:shrinking} for an example).
Scopes have several desirable properties:

\begin{prop}\label{prop:of-scopes}
 Let $\alpha,\beta$ be disjoint scopes and $S,U,T$ teams in a Kripke structure $\calK = (W,R,V)$.
 Then the following laws hold:
 \begin{enumerate}
  \item Distributive laws: ${(T \cap S)}_\alpha = T_\alpha \cap S = T \cap S_\alpha = T_\alpha \cap S_\alpha$ and ${(T \cup S)}_\alpha = T_\alpha \cup S_\alpha$.
  \item Disjoint selection commutes: ${\big(T^\alpha_S\big)}^{\beta}_{U} = {\big(T^{\beta}_{U}\big)}^{\alpha}_S$.
  \item Disjoint selection is independent: ${\big({(T^\alpha_S)}^{\beta}_{U}\big)}_\alpha = T_\alpha \cap S$.
  \item Image and selection commute: ${(RT)}_\alpha = {\big(R(T_\alpha)\big)}_\alpha = R(T_\alpha)$
  \item Selection propagates: If $S \subseteq T$, then $R{\big(T^\alpha_{S}\big)} = {(RT)}^\alpha_{RS}$.
 \end{enumerate}
\end{prop}
\begin{proof}
  Straightforward; see the appendix.
\end{proof}

Accordingly, we write $R^{i}T_\alpha$ instead of ${(R^{i}T)}_\alpha$ or $R^i(T_\alpha)$ and $T^{\alpha_1,\alpha_2}_{S_1,S_2}$ for ${(T^{\alpha_1}_{S_1})}^{\alpha_2}_{S_2}$.

\subsection{Subteam quantifiers}

We refer to the following abbreviations as \emph{subteam quantifiers}, where $\alpha \in \ML$:\label{p:quantifiers}
\begin{align*}
 \qee{\alpha} \; \varphi & \dfn \alpha \lor \varphi                                                            & \qaa{\alpha} \;\varphi & \dfn \negg\qee{\alpha} \negg\varphi \\
 \qe{\alpha} \; \varphi  & \dfn \qee{\alpha}\left[ \E\alpha \land \qaa{\alpha}(\E\alpha \timp \varphi) \right] & \qa{\alpha} \; \varphi & \dfn \negg\qe{\alpha} \negg \varphi
\end{align*}

Intuitively, they quantify over subteams $S \subseteq T_\alpha$  or worlds $w \in T_\alpha$ such that $T^\alpha_S$ resp.\ $T^\alpha_w$ satisfies $\varphi$.

\begin{prop}\label{prop:quantifiers}
 The subteam quantifiers have the following semantics:
 \begin{alignat*}{8}
   & \qquad T\; \vDash \;\; &  & \qee{\alpha} \varphi &  & \; \Leftrightarrow \; \exists S \subseteq T_\alpha &  & \colon T^\alpha_S \vDash \varphi &  & \qquad\qquad T\; \vDash  \;\; &  & \qe{\alpha} \varphi &  & \;\Leftrightarrow \;\exists w \in T_\alpha &  & \colon T^\alpha_w \vDash \varphi \\
   & \qquad T\; \vDash \;\; &  & \qaa{\alpha} \varphi &  & \;\Leftrightarrow \; \forall S \subseteq T_\alpha  &  & \colon T^\alpha_S \vDash \varphi &  & \qquad\qquad T\; \vDash  \;\; &  & \qa{\alpha} \varphi &  & \;\Leftrightarrow \;\forall w \in T_\alpha &  & \colon T^\alpha_w \vDash \varphi
 \end{alignat*}
\end{prop}
\begin{proof}
 We prove the existential cases, as the other ones work dually.

 Let us first consider the "$\Rightarrow$" direction for $\qee{\alpha}$.
 Accordingly, suppose $T \vDash \qee{\alpha}\, \varphi$, \ie, $T \vDash \alpha \lor \varphi$.
 Then there exist $S \subseteq T$ and $U \subseteq T_\alpha$ such that $S \vDash \varphi$ and $T = S \cup U$.
 Since $U \cap T_{\neg\alpha} = \emptyset$, it holds $T_{\neg\alpha} \subseteq S$.
 Moreover, $
  S = (S \cap T_\alpha) \cup (S \cap T_{\neg \alpha}) = ((S \cap T_\alpha) \cap T_{\alpha}) \cup T_{\neg\alpha} = T^\alpha_{S \cap T_\alpha}$.
 Consequently, $T^\alpha_{S \cap T_\alpha} \vDash \varphi$ for some set $S \cap T_\alpha \subseteq T_\alpha$.

 For "$\Leftarrow$", suppose $T^\alpha_S \vDash \varphi$ for some $S \subseteq T_\alpha$.
 Then $T^\alpha_S$ and $T \setminus T^\alpha_{S}$ form a division of $T$.
 Since $T \setminus T^\alpha_S = T \setminus \left( T_{\neg\alpha} \cup (T_\alpha \cap S) \right) \subseteq T\setminus T_{\neg\alpha} = T_\alpha$, it holds  $T \setminus T^\alpha_{S} \vDash \alpha$.
 As a consequence, $T \vDash \alpha \lor \varphi$.

 \smallskip

 We proceed with $\qe{\alpha}$.
 For "$\Rightarrow$", suppose that $T \vDash \qe{\alpha} \varphi$.
 Then there exists $S \subseteq T_\alpha$ such that $T^\alpha_S \vDash \E \alpha \land \qaa{\alpha} (\E\alpha \timp \varphi)$.
 Since $T^\alpha_S \vDash \E\alpha$, there exists $w \in {(T^\alpha_S)}_\alpha$.
 As $\qaa{\alpha}$ now applies to ${(T^\alpha_S)}^\alpha_{\{w\}} = T^\alpha_w$ as well, it follows $T^\alpha_w \vDash \E\alpha \timp\varphi$, and consequently $T^\alpha_w \vDash \varphi$.

 Suppose for "$\Leftarrow$" that $T^\alpha_w \vDash \varphi$ for some $w \in T_\alpha$.
 Let $S \subseteq T_\alpha$ be arbitrary.
 If $w \notin S$, then ${(T^\alpha_w)}^\alpha_S = T^\alpha_\emptyset \nvDash \E \alpha$, and if $w \in S$, then ${(T^\alpha_w)}^\alpha_S = T^\alpha_w \vDash \varphi$.
 Therefore, for any $S \subseteq T_\alpha$ it holds ${(T^\alpha_w)}^\alpha_S \vDash (\E \alpha \timp \varphi)$, so $T^\alpha_w \vDash \qaa{\alpha} (\E\alpha \timp \varphi)$.
 Since also $T^\alpha_w \vDash \E\alpha$, it follows $T \vDash \qee{\alpha} \left[ \E\alpha \land \qaa{\alpha}(\E\alpha \timp \varphi) \right]$.
\end{proof}

\subsection{Implementing bisimulation}

With scopes and subteam quantifiers at our hands, we have all ingredients to implement $k$-bisimulation.\label{p:xi}
\begin{align*}
 \chi_0(\alpha,\beta)     & \dfn (\alpha\lor\beta) \hook \bigwedge_{p\in\Phi}\dep{p}                                                                                                                                            \\[2mm]
 \chi_{k+1}(\alpha,\beta) & \dfn \chi_0(\alpha,\beta) \land \Box\chi^*_{k}(\alpha,\beta)                                                                                                                                        \\[2.5mm]
 \chi^*_k(\alpha,\beta)   & \dfn (\neg\alpha\land\neg\beta) \ovee \Big(\E \alpha\land\E\beta \land \negg \big[(\alpha \ovee \beta)  \lor (\E\alpha\land\E\beta \land \negg \qe{\alpha}\qe{\beta}\chi_k(\alpha,\beta))\big]\Big)
\end{align*}

Note that a literal translation of the forward and backward condition would rather result in the formula $\chi^*_k(\alpha,\beta) \dfn \qa{\alpha}\qe{\beta}\chi_k(\alpha,\beta) \land \qa{\beta}\qe{\alpha}\chi_k(\alpha,\beta)$.
The more complicated formula shown above however avoids the exponential size that would come with two recursive calls.

\begin{thm}\label{thm:main-bisim}
 Let $k \geq 0$.
 For all Kripke structures $\calK$, teams $T$ and disjoint scopes $\alpha,\beta$ in $\calK$, and points $w \in T_\alpha$ and $v \in T_\beta$ it holds:
 \begin{alignat*}{4}
   &  & T^{\alpha,\beta}_{w,v} & \;\vDash\; \chi_k(\alpha,\beta) \;     &  & \text{ if and only if } \; & w \bis_k       \,   & v\text{,}       \\
   &  & T                      & \;\vDash\; \chi^*_{k}(\alpha,\beta) \; &  & \text{ if and only if } \; & T_\alpha \bis_k  \, & T_\beta\text{.}
 \end{alignat*}

 Moreover, both $\chi_k(\alpha,\beta)$ and $\chi^*_k(\alpha,\beta)$ are $\MTL_k$-formulas that are constructible in
 space $\bigO{\log(k + \size{\Phi} + \size{\alpha} + \size{\beta})}$.
\end{thm}

\begin{figure}\centering
 \begin{tikzpicture}[>=stealth,scale=0.6,node distance=0]
  \tikzset{world/.style={draw,circle,inner sep=.3mm,black,fill}}
  \tikzset{team/.style={draw,rounded corners,thick,inner sep=1mm}}
  \node[world] (r1) at (2.5,-2) {};
  \node[world] (r2) at (2.5,2) {};

  \node[color=red] at (-.5,2) {$\alpha$};
  \node[color=blue] at (-.5,-2) {$\beta$};
  \node[above=1mm of r2] {$T$};

  \foreach \i in {1,2,3,4} {
    \node[world] (w1\i) at (\i,-1) {};
    \draw[->] (r1) -- (w1\i);
   }
  \foreach \i in {1,2,3} {
    \node[world] (v1\i) at (\i,1) {};
    \draw[->] (r2) -- (v1\i);
   }
  \node[right=0mm of w14] {$z$};

  \foreach \i in {1,2} {
    \path (w1\i) edge[draw,-,dotted,thick,color=green!50!black] node[left,pos=.5] {\color{green!50!black}{\tiny$\bis_0\!$}} (v1\i);
   }
  \foreach \i in {3} {
    \path (w1\i) edge[draw,-,dotted,thick,color=green!50!black] node[right,pos=.5] {\color{green!50!black}{\tiny\!$\bis_0$}} (v1\i);
   }

  \begin{scope}
   \clip (0,3) rectangle (4,0);
   \node[team,fit=(r1)(r2),draw=red,draw opacity=.5,fill=red,fill opacity=.1] {};
  \end{scope}
  \begin{scope}
   \clip (0,-3) rectangle (4,0);
   \node[team,fit=(r1)(r2),draw=blue,draw opacity=.5,fill=blue,fill opacity=.1] {};
  \end{scope}

  \node (bis) at (-1.5cm,0) {\color{green!50!black}$\bis_1$?};

  \node[fit=(r1)] (fitr1){};
  \node[fit=(r2)] (fitr2){};

  \path[draw,-,dashed,thick,color=green!50!black] (fitr2.west) to [out=180,in=90] (bis.east) to [out=270,in=180] (fitr1.west);
  \node at (5.5,0) {\Large$\Rightarrow$};

  \begin{scope}[xshift=6cm]
   \node[world,opacity=.3] (r1) at (2.5,-2) {};
   \node[world,opacity=.3] (r2) at (2.5,2) {};

   \foreach \i in {1,2,3,4} {
     \node[world] (w1\i) at (\i,-1) {};
     \draw[->,opacity=.3] (r1) edge (w1\i);
    }
   \foreach \i in {1,2,3} {
     \node[world] (v1\i) at (\i,1) {};
     \draw[->,opacity=.3] (r2) edge (v1\i);
    }

   \begin{scope}
    \clip (0,3) rectangle (5,0);
    \node[team,fit=(w14)(v11),draw=red,draw opacity=.5,fill=red,fill opacity=.1] {};
   \end{scope}
   \begin{scope}
    \clip (0,-3) rectangle (5,0);
    \node[team,fit=(w14)(v11),draw=blue,draw opacity=.5,fill=blue,fill opacity=.1] (team) {};
   \end{scope}
   \foreach \i in {1,2,3} {
     \path (w1\i) edge[draw,-,dotted,thick,color=green!50!black] node[right,pos=.5] {\color{green!50!black}{\tiny\!$\bis_0$}} (v1\i);
    }
   \node at (5.5,0) {\Large$\Rightarrow$};
   \node[above right=0mm and -4mm of team] {$RT$};
   \node[right=1mm of w14] {$z$};
  \end{scope}

  \begin{scope}[xshift=12cm]
   \node[world,opacity=.3] (r1) at (2.5,-2) {};
   \node[world,opacity=.3] (r2) at (2.5,2) {};

   \foreach \i in {1,2,3,4} {
     \node[world] (w1\i) at (\i,-1) {};
     \draw[->,opacity=.3] (r1) edge (w1\i);
     \node[inner sep=1mm,fit=(w1\i)] (fitw1\i) {};
    }
   \foreach \i in {1,2,3} {
     \node[world] (v1\i) at (\i,1) {};
     \draw[->,opacity=.3] (r2) edge (v1\i);
     \node[inner sep=1mm,fit=(v1\i)] (fitv1\i) {};
    }
   \node[inner sep=.3mm] (v14) at (4,1){};
   \node[inner sep=1mm,fit=(v14)] (fitv14) {};

   \begin{scope}
    \clip (0,3) rectangle (5,0);
    \draw[draw=red,draw opacity=.5,fill=red,fill opacity=.1,thick,rounded corners]
    (fitv11.north west)
    -- ([yshift=-1cm]fitv11.north west)
    -- ([yshift=1cm]fitw14.north west)
    -- (fitw14.south west) -- (fitw14.south east) -- (fitv14.north east) -- cycle;
   \end{scope}
   \begin{scope}
    \clip (0,-3) rectangle (5,0);
    \draw[draw=blue,draw opacity=.5,fill=blue,fill opacity=.1,thick,rounded corners]
    (fitv11.north west)
    -- ([yshift=-1cm]fitv11.north west)
    -- ([yshift=1cm]fitw14.north west)
    -- (fitw14.south west) -- (fitw14.south east) -- (fitv14.north east) -- cycle;
   \end{scope}

   \draw[draw=blue,draw opacity=.25,fill=blue,fill opacity=.05,thick,rounded corners]
   (fitw11.south west) -- (fitw13.south east) -- ([yshift=1cm]fitw13.south east)
   -- ([yshift=1cm]fitw11.south west) -- cycle;

   \node[right=1mm of w14] {$z$};
   \node[above right=0mm and -2mm of v14] {$RT^\beta_z$};
  \end{scope}

 \end{tikzpicture}
 \caption{As $z$ violates the \emph{backward} condition, shrinking $RT_\beta$ leads to a {\color{green!50!black}$\bis_0$}-free subteam, falsifying $\qe{\red{\alpha}}\qe{\blue{\beta}}\color{green!50!black}\chi_0(\alpha,\beta)$.}\label{fig:bisimulation}
\end{figure}

\begin{proof}
 The idea is to isolate a single point in $z \in T_\alpha \cup T_\beta$ that serves as a \emph{counter-example} against $\type{T_\alpha}_k = \type{T_\beta}_k$ by, say, $\type{z}_k \in \type{T_\beta}_k \setminus \type{T_\alpha}_k$.
 We erase $T_\beta \setminus \{z\}$ from $T$ using the disjunction $\lor$, as $T_\beta \setminus\{z\} \vDash \alpha \ovee \beta$.
 The remaining team is exactly $T^\beta_z$, in which $\qe{\alpha}\qe{\beta}\chi_k(\alpha,\beta)$ fails (see Figure~\ref{fig:bisimulation}).
 The case $\type{z}_k \in \type{T_\alpha}_k \setminus \type{T_\beta}_k$ is detected analogously.

 We proceed with a formal correctness proof by induction on $k$.
 Let $\calK = (W,R,V)$ as in the theorem.
 The base case $k = 0$ is straightforward,
as no proposition $p \in \Phi$ occurs in $\alpha$ or $\beta$.
 The induction step is split into two parts.

\smallskip

"$\chi_k \Rightarrow \chi^*_k$":
  Let $T$ be a team and $\alpha,\beta$ disjoint scopes.
 Observe that $\chi^*_k$ is always true if $T_\alpha$ and $T_\beta$ are both empty (then $\type{T_\alpha}_k = \type{T_\beta}_k$), and that it is always false if exactly one of them is empty (then $\type{T_\alpha}_k \neq \type{T_\beta}_k$).
 Therefore, let $T_\alpha \neq \emptyset$ and $T_\beta \neq \emptyset$.
 Then $\chi^*_k(\alpha,\beta)$ boils down to $\negg((\alpha \ovee \beta) \lor\E\alpha\land\E\beta\land\negg\qe{\alpha}\qe{\beta}\chi_k(\alpha,\beta))$,
 which we prove equivalent to $\type{T_\alpha}_k = \type{T_\beta}_k$.

 The first direction is proved by contradiction.
 Suppose $\type{T_\alpha}_k = \type{T_\beta}_k$ but $T \vDash (\alpha \ovee \beta) \lor \E\alpha\land\E\beta\land\negg\qe{\alpha}\qe{\beta} \chi_k(\alpha,\beta)$.
 The disjunction is witnessed by some division $T = S \cup U$, where \wloss $S \subseteq T_\alpha$ satisfies $\alpha \ovee \beta$,
 (if $S \subseteq T_\beta$, the proof is symmetric), and $U \vDash \E\alpha\land\E\beta\land\negg\qe{\alpha}\qe{\beta}\chi_k(\alpha,\beta)$.
 Since $T_\alpha \cap T_\beta = \emptyset$, then $T_\beta \subseteq U$, and clearly $T_\beta \subseteq U_\beta$.
 By the formula, some $w \in U_\alpha$ exists.
 By assumption that $\type{T_\alpha}_k = \type{T_\beta}_k$, $U_\beta$ must contain a world $v$ of type $\type{w}_k$ as well.
 But then $U^{\alpha,\beta}_{w,v} \vDash \chi_k(\alpha,\beta)$ by induction hypothesis, contradiction to $U \vDash \negg \qe{\alpha}\qe{\beta} \chi_k(\alpha,\beta)$.

 For the other direction, suppose $\type{T_\alpha}_k \neq \type{T_\beta}_k$.
 \Wloss there exists $w \in T_\alpha$ such that $\type{w}_k \notin \type{T_\beta}_k$.
 (For $w \in T_\beta$, the proof is again symmetric.)
 Consider $S \dfn T_\alpha \setminus \{w\}$ and $U \dfn T^\alpha_w$ as a division of $T$.
 Then $S \vDash \alpha \ovee \beta$ and $U \vDash \E\alpha\land\E\beta$.
It remains to show $U \vDash \negg\qe{\alpha}\qe{\beta} \chi_k(\alpha,\beta)$.
 However, this is easy to see:
 $U \vDash \qe{\alpha}\qe{\beta}\chi_k(\alpha,\beta)$ if and only if $U\vDash \qe{\beta}\chi_k(\alpha,\beta)$, but $T_\beta$ and hence $U_\beta$ contains no world of type $\type{w}_k$, so by induction hypothesis $U$ cannot satisfy $\qe{\beta}\chi_k(\alpha,\beta)$.

 \smallskip

"$\chi^*_k \Rightarrow \chi_{k+1}$":
 We follow Definition~\ref{def:point-bisim} and Proposition~\ref{prop:ml-bisim-types}.
 \begin{align*}
                    & T^{\alpha,\beta}_{w,v} \vDash \chi_{k+1}(\alpha,\beta)                                                               \\
  \Leftrightarrow\; & T^{\alpha,\beta}_{w,v} \vDash \chi_0(\alpha,\beta)\land
  \Box\chi^*_k(\alpha,\beta)\tag{Definition of $\chi_{k+1}$}                                                                               \\
  \Leftrightarrow\; & w \bis_0 v \text{ and }T^{\alpha,\beta}_{w,v} \vDash \Box\chi^*_k(\alpha,\beta)\tag{Induction hypothesis}            \\
  \Leftrightarrow\; & w \bis_0 v \text{ and }RT^{\alpha,\beta}_{Rw,Rv} \vDash \chi^*_k(\alpha,\beta)\tag{Proposition~\ref{prop:of-scopes}} \\
  \Leftrightarrow\; & w \bis_0 v \text{ and }Rw \bis_k Rv\tag{Induction hypothesis}                                                        \\
  \Leftrightarrow\; & w \bis_{k+1} v\text{.}\tag*{({Proposition~\ref{prop:ml-bisim-types}})}
 \end{align*}
 It is routine to check that the formulas are constructible in logarithmic space from $\alpha$, $\beta$, $\Phi$ and $k$, and that $\md(\chi_k) = \md(\chi^*_k) = k$.
\end{proof}

Let us stress that $\chi_k$ relies on disjoint scopes to be present in the structure, and it is open whether the property $\size{\type{T}_k} \leq 1$ is polynomially definable without these.
Incidentally, the related property $\size{\type{Rw}_k} \leq 1$ of points $w$ was recently studied by Hella and Vilander~\cite{HellaV16}, and was proven to be expressible in $\ML$, but only by formulas of non-elementary size.
However, they proved that it is definable in exponential size in \emph{2-dimensional modal logic} $\ML^2$ (for an introduction to $\ML^2$, see Marx and Venema~\cite{marx1997multi}).
Roughly speaking, $\ML^2$ is evaluated by traversing over \emph{pairs} of points independently.
The relationship between $\ML^2$ and $\MTL$ is unclear: Arguably, pairs of points are a special case of teams.
But on the other hand, the modalities in $\MTL$ do not act on the points in a team independently, as in $\ML^2$, but instead always proceed to a successor team "synchronously".
As a consequence, it is also open whether $\MTL$ can define one of the above properties by a formula of elementary size.

\section{Enforcing a canonical model}%
\label{sec:staircase}

In this section, we approach the canonical models of $\MTL$ from a lower bound perspective.
Here, we devise an $\MTL_k$-formula that is satisfiable but permits \emph{only} $k$-canonical models.

For $k = 0$, that is propositional team logic, Hannula et al.~\cite{Hannula2015} defined the $\PTL$-formula
\begin{align*}
 \max(X) \dfn \negg \bigvee_{\mathclap{x \in X}}\dep{x}
\end{align*}
and proved that $T \vDash \max(\Phi)$ if and only if $T$ is $0$-canonical, \ie, contains all Boolean assignment over $\Phi$.
We generalize this for all $k$, \ie, construct a satisfiable formula $\canon_k$ that has only $k$-canonical models.

\subsection{Staircase models}

Our approach is to express $k$-canonicity by inductively enforcing $i$-canonical sets of worlds for $i = 0,\ldots,k$ located in different "height" inside the model.
For this purpose, we employ distinct scopes $\fraka_0,\ldots,\fraka_k$ ("\emph{stairs}"), and introduce a specific class of models:

\begin{defi}\label{def:staircase}
 Let $k,i \geq 0$ and let $(\calK,T)$ be a Kripke structure with team, $\calK = (W,R,V)$.
 Then $T$ is \emph{$k$-canonical with offset $i$} if for every $\tau \in \Delta_k$ there exists $w \in T$ with $\type{R^{i}w}_k = \{\tau\}$.
 $(\calK,T)$ is called \emph{$k$-staircase} if for all $i \in \{0,\ldots,k\}$ we have that $T_{\fraka_i}$ is $i$-canonical with offset $k-i$.
\end{defi}

\begin{figure*}[b!]\centering
\begin{tikzpicture}[scale=.9]
\tikzset{world/.style={draw,circle,inner sep=.3mm,black,fill}}
\tikzset{dummy/.style={circle,inner sep=.3mm}}
\tikzset{team/.style={draw,rounded corners,thick,inner sep=1.6mm}}

\node at (0.1,-1) {$\fraka_0$};
\node at (1.35,-1) {$\fraka_1$};
\node at (4.1,-1) {$\fraka_2$};
\node at (10,-1) {$\fraka_3$, $2^{2^{2^{2^\size{\Phi}}}} = 16 = \size{\Delta_3}$ elements};

\foreach \j in {0,1,2,3} {
\node[world] (w0\j) at (0,\j) {};
}
\path
(w00) edge[-] (w01)
(w01) edge[-] (w02)
(w02) edge[-] (w03);

\draw[dashed] (.5,3.5) -- (.5,-0.5);

\node[team,fill=red,fill opacity=.4,fit=(w03),draw=red,draw opacity=.6] (canon0) {};

\foreach \i in {2,3} {
\foreach \j in {0,1,2} {
\node[world] (w\i\j) at ({\i*0.5},\j) {};
}
\path
(w\i0) edge[-] (w\i1)
(w\i1) edge[-] (w\i2);
}
\node[world] (w33) at (1.5,3) {};
\path (w32) edge (w33);
\node[dummy] (d1) at (1,3) {};

\draw[dashed] (2,3.5) -- (2,-0.5);

\node[team,fill=red,fill opacity=.4,fit=(w33),draw=red,draw opacity=.6] {};
\node[team,fill=red,fill opacity=.4,fit=(d1),draw=red,draw opacity=.6] {};
\node[team,fill=red,fill opacity=.25,fit=(w22)(w32),draw=red,draw opacity=.4] (canon1) {};

\begin{scope}[xshift=1mm]
\foreach \i in {5,7,9,11} {
\foreach \j in {0,1} {
\node[world] (w\i\j) at ({\i*0.5+0.1},\j) {};
\node[dummy] (ldw\i\j) at ({\i*0.5-0.1},\j) {};
\node[dummy] (rdw\i\j) at ({\i*0.5+0.3},\j) {};
}
\path
(w\i0) edge[-] (w\i1);
}
\node[world] (v7a) at (3.35,2) {};
\node[dummy] (d7x) at (3.8,2) {};
\path (w71) edge (v7a);
\node[world] (v9b) at (4.8,2) {};
\node[world] (v9c) at (4.8,3) {};
\path (w91) edge (v9b);
\path (v9b) edge (v9c);
\node[world] (v11a) at (5.35,2) {};
\node[world] (v11b) at (5.8,2) {};
\node[world] (v11c) at (5.8,3) {};
\path (w111) edge (v11a);
\path (w111) edge (v11b);
\path (v11b) edge (v11c);

\node[dummy] (d2) at (4.35,2) {};

\node[dummy] (d4) at (2.4,2) {};
\node[dummy] (d4x) at (2.8,2) {};
\node[dummy] (d5) at (2.5,3) {};

\draw[dashed] (6.4,3.5) -- (6.4,-0.5);

\node[team,fill=red,fill opacity=.25,fit=(d4)(d4x),draw=red,draw opacity=.4] {};
\node[team,fill=red,fill opacity=.25,fit=(v7a)(d7x),draw=red,draw opacity=.4] {};
\node[team,fill=red,fill opacity=.25,fit=(d2)(v9b),draw=red,draw opacity=.4] {};
\node[team,fill=red,fill opacity=.25,fit=(v11a)(v11b),draw=red,draw opacity=.4] {};
\node[team,fill=red,fill opacity=.1,fit=(ldw51)(rdw111),draw=red,draw opacity=.3] (canon2) {};
\end{scope}

\node[world] (u0) at (7,0) {};
\node[dummy] (ue) at (12.8,0) {};
\node[world] (u1) at (8,0) {};
\node[world] (u2) at (8,1) {};
\path (u1) edge (u2);
\node at (8.8,0) {$\cdots$};
\node at (8.8,1) {$\cdots$};

\begin{scope}[xshift=7cm]
\node[world] (u3) at (4,0) {};
\foreach \i in {5,7,9,11} {
\node[world] (x\i1) at ({\i*0.5},1) {};
\path
(u3) edge[-] (x\i1);
}
\node[world] (v7a) at (3.2,2) {};
\path (x71) edge (v7a);
\node[world] (v9b) at (4.8,2) {};
\node[world] (v9c) at (4.8,3) {};
\path (x91) edge (v9b);
\path (v9b) edge (v9c);
\node[world] (v11a) at (5.3,2) {};
\node[dummy] (dv11a) at (5.8,1) {};
\node[world] (v11b) at (5.8,2) {};
\node[world] (v11c) at (5.8,3) {};
\path (x111) edge (v11a);
\path (x111) edge (v11b);
\path (v11b) edge (v11c);
\end{scope}

\node[dummy] (d6) at (7,1) {};

\node[team,fill=red,fill opacity=.1,fit=(d6),draw=red,draw opacity=.3] {};
\node[team,fill=red,fill opacity=.1,fit=(u2),draw=red,draw opacity=.3] {};
\node[team,fill=red,fill opacity=.1,fit=(x51)(dv11a),draw=red,draw opacity=.3] {};
\node[team,fill=red,fill opacity=.1,fit=(u0)(ue),draw=red,draw opacity=.2] (canon3) {};

\draw[thick, red, rounded corners] (-.3,2.65) -- (.65,2.65) -- (.65,1.65) -- (2.15,1.65) -- (2.15,.65) -- (6.65,.65) -- (6.65,-.35) -- (13,-.35);

\node[above=3.2cm of canon3] {3-canonical};
\node[above=2.3cm of canon2] {2-canonical};
\node[above=1.4cm of canon1,xshift=1mm] {1-canonical};
\node[above=.5cm of canon0,xshift=-2mm] {0-c.};

\draw [decorate,decoration={brace,amplitude=10pt,mirror},xshift=0pt,yshift=0pt]
(-.4,2.6) -- (-.4,.3) node [black,midway,xshift=-1cm,yshift=.5mm]
{Offset};

\node at (-1,-1) {Scope:};

\node[dummy] (teamend) at (12.8,0) {};

\node[team,fit=(w00)(teamend),fill=black,fill opacity=.1,draw opacity=.3] {};
\node at (-.5,0) {$T$};
\end{tikzpicture}
\caption{Visualization of the $3$-staircase for $\Phi = \emptyset$, where the subteam $T_{\fraka_i}$ is $i$-canonical with offset $3-i$.}\label{fig:staircase}
\end{figure*}

As an example, a $3$-staircase for $\Phi = \emptyset$ is depicted in Figure~\ref{fig:staircase}.
Observe that it is a \emph{directed forest}, \ie, it is acyclic and all worlds are either \emph{roots} (\ie, without predecessor) or have exactly one predecessor.
Moreover, it has bounded \emph{height}, where the height of a directed forest is the greatest number $h$ such that every path traverses at most $h$ edges.

\begin{prop}\label{prop:canon-sat}
 For each $k\geq 0$, there is a finite $k$-staircase $(\calK,T)$ such that $\fraka_0,\ldots,\fraka_k$ are disjoint scopes in $\calK$, and $\calK$ is a directed forest with height at most $k$ and its set of roots being exactly $T$.
\end{prop}
\begin{proof}
 See Figure~\ref{fig:staircase}.
\end{proof}

Observe that in such a model, $T_{\fraka_k}$ is $k$-canonical with offset $0$, which is simply $k$-canonical:

\begin{cor}[Finite tree model property of $\MTL$]\label{cor:height-model}
 Every satisfiable $\MTL_k$-formula has a finite model $(\calK,T)$ such that $\calK$ is a directed forest with height at most $k$ and its set of roots being exactly $T$.
\end{cor}

\subsection{Enforcing canonicity}

In the rest of the section, we illustrate how a $k$-staircase can be enforced in $\MTL$ inductively.

For $\Phi = \emptyset$, the inductive step---obtaining $(k+1)$-canonicity from $k$-canonicity---is done by the formula $\qaa{\alpha} \, \qe{\beta} \, \Box \chi^*_{k}(\alpha,\beta)$.
The idea is that this formula states that for every \emph{subteam} $T' \subseteq T_\alpha$ there exists a \emph{point} $w \in T_\beta$ such that $\type{RT'}_k = \type{Rw}_k$.
Intuitively, every possible set of types is captured as the image of some point in $T_\beta$.
As a consequence, if $T_\alpha$ is $k$-canonical with offset $1$, then $T_\beta$ will be $(k+1)$-canonical.

Note that the simpler formula $\Box^k\max(\Phi)$ expresses $0$-canonicity of $R^{k}T$, but not $0$-canonicity of $T$ with offset $k$ (consider, \eg, a singleton $T$).
Instead, we use the formula
\begin{align*}
 \max_i \dfn \; & \top \lor (\Diamond^i\top \land \negg \bigvee_{\mathclap{p \in \Phi}} (\Diamond^i p \ovee \Diamond^i \neg p))\text{.}
\end{align*}
It states not only that $R^{i}T$ is $0$-canonical, but also that $R^{i}w$ contains exactly one propositional assignment for each $w \in T$, which together yields $0$-canonicity with offset $i$.
\begin{lem}\label{lem:canon0}
 $T \vDash \max_i$ iff $T$ is $0$-canonical with offset $i$.
\end{lem}
\begin{proof}
 By the distributive law $\varphi \lor (\psi_1 \ovee \psi_2) \equiv (\varphi \lor \psi_1) \ovee (\varphi \lor \psi_2)$, the duality $\negg(\psi_1 \ovee\psi_2)\equiv \negg \psi_1 \land \negg \psi_2$, and the definition
 $\E \psi = \negg \neg \psi$,
 \[
  \negg \bigvee_{\mathclap{p \in \Phi}} (\Diamond^i p \ovee \Diamond^i \neg p) \equiv \negg \bigovee_{P \subseteq \Phi} \Big( \bigvee_{p \in P} \Diamond^i p \lor \bigvee_{\mathclap{p \in \Phi \setminus P}}
  \Diamond^i \neg p \Big) \equiv \bigwedge_{P \subseteq \Phi}\!\! \E
  \Big(  \bigwedge_{p \in P} \Box^i \neg p \land \bigwedge_{\mathclap{p \in \Phi \setminus P}} \Box^i  p\Big)
  \text{.}
 \]
 The rightmost formula now states that for all types $\tau \in \Delta_0$ (each represented by a subset of $\Phi$, cf.\ Proposition~\ref{prop:det-types}), there exists a world $w \in T$ such that $\type{R^{i}w}^{\Phi}_0 \subseteq \{ \tau\}$.
 Likewise, $T \vDash \Diamond^i\top$ iff $R^{i}w \neq \emptyset$ for every $w \in T$.
\end{proof}

Based on this, $k$-canonicity with offset $i$ is now recursively defined as $\rho^i_k$:
\begin{align*}
 \rho^i_0(\beta) \, & \dfn \; \beta \hook \max_i                                                                                                     \\
 \rho^i_{k+1}(\alpha,\beta)
                    & \,\dfn \; \qaa{\alpha} \, \qee{\beta} \left( \rho^i_0(\beta) \land \Box^{i} \qa{\beta} \; \Box \chi^*_{k}(\alpha,\beta)\right) \\
 \canon_{k} \,      & \dfn \;  \rho^k_0(\fraka_0) \land \bigwedge_{m = 1}^k \rho^{k-m}_{m}(\fraka_{m-1},\fraka_m)
\end{align*}

\begin{thm}\label{thm:canon-formula}
 Let $k \geq 0$ and $\calK$ be a structure with disjoint scopes $\fraka_0,\ldots,\fraka_k$.
 Then $(\calK,T) \vDash \canon_k$ if and only if $(\calK,T)$ is a $k$-staircase.
 Moreover, $\canon_{k}$ is an $\MTL_k$-formula constructible in space $\bigO{\log(\size{\Phi} + k)}$.
\end{thm}
\begin{proof}
 Similar to Theorem~\ref{thm:main-bisim}, the construction of the above formula in logspace is straightforward.
 We proceed with the correctness of the formula.
 Suppose that $\fraka_0,\ldots,\fraka_k$ are disjoint scopes in $\calK$.
 We show the following by induction on $0 \leq i \leq k$:
 Assuming that $T_\alpha$ is $k$-canonical with offset $i+1$, it holds that $T_\beta$ is $(k+1)$-canonical with offset $i$ if and only if $T \vDash \rho^i_{k+1}(\alpha,\beta)$.
 With the induction basis done in Lemma~\ref{lem:canon0}, the inductive step is proved by the following equivalence:
 \begin{align*}
                     & T_\beta\text{ is }(k+1)\text{-canonical with offset }i                                                                                                                                                       \\
  \Leftrightarrow \; &
  \forall \tau \in \Delta_{k+1} \colon \exists w \in T_\beta \colon \type{R^{i}w}_{k+1} = \{ \tau \}                                                                                                                                  \\
  \intertext{Using the inverse of the bijection $h\colon \tau \mapsto (\tau \cap \Phi,\calR\tau)$ from Proposition~\ref{prop:det-types}, we can equivalently quantify over $\pow{\Delta_k}$ and $\pow{\Phi}$:
  }
  \Leftrightarrow \; &
  \forall \Delta' \subseteq \Delta_k \colon \forall \Phi' \subseteq \Phi  \colon \exists w \in T_\beta \colon  \type{R^{i}w}_{k+1} = \{h^{-1}(\Phi',\Delta')\}                                                                        \\
  \Leftrightarrow \; &
  \forall \Delta' \subseteq \Delta_k \colon \forall \Phi' \subseteq \Phi  \colon \exists w \in T_\beta \colon R^{i}w \neq \emptyset \text{ and } \forall v \in R^{i}w \colon \type{v}_{k+1} = h^{-1}(\Phi',\Delta')                     \\
  \intertext{By Lemma~\ref{lem:types-equal-const}, $V^{-1}(v) = \Phi' \text{ and }\type{Rv}_k = \Delta'$ is equivalent to $\type{v}_{k+1} = h^{-1}(\Phi',\Delta')$:}
  \Leftrightarrow \; &
  \forall \Delta' \subseteq \Delta_k \colon \forall \Phi' \subseteq \Phi  \colon \exists w \in T_\beta \colon R^{i}w \neq \emptyset                                                                                                   \\
                     & \qquad \text{and } \forall v \in R^{i}w \colon V^{-1}(v) = \Phi' \text{ and }\type{Rv}_k = \Delta'                                                                                                             \\
  \intertext{Again by Proposition~\ref{prop:det-types}, $h \colon \tau \mapsto \tau \cap \Phi$ is a bijection from $\Delta_0$ to $\pow{\Phi}$:}
  \Leftrightarrow \; &
  \forall \Delta' \subseteq \Delta_k \colon \forall \tau_0 \in \Delta_0  \colon \exists w \in T_\beta \colon R^{i}w \neq \emptyset                                                                                                    \\
                     & \qquad \text{and } \forall v \in R^{i}w \colon V^{-1}(v) = \tau_0 \cap \Phi \text{ and }\type{Rv}_k = \Delta'                                                                                                  \\
  \intertext{Once more by Lemma~\ref{lem:types-equal-const}:}
  \Leftrightarrow \; &
  \forall \Delta' \subseteq \Delta_k \colon \forall \tau_0 \in \Delta_0  \colon \exists w \in T_\beta \colon R^{i}w \neq \emptyset                                                                                                    \\
                     & \qquad \text{and }\forall v \in R^{i}w \colon \type{v}_{0} = \tau_0 \text{ and }\type{Rv}_k = \Delta'                                                                                                          \\
  \Leftrightarrow \; &
  \forall \Delta' \subseteq \Delta_k \colon \forall \tau_0 \in \Delta_0  \colon \exists w \in T_\beta \colon \type{R^{i}w}_0 = \{ \tau_0\}  \text{ and }\forall v \in R^{i}w \colon \type{Rv}_k = \Delta'                               \\
  \intertext{Since $T_\alpha$ is assumed $k$-canonical with offset $i+1$, for every $\tau' \in \Delta_k$ there exists $u \in T_\alpha$ such that $\type{R^{i+1}u}_k = \{ \tau'\}$.
   Accordingly, for every \emph{set} $\Delta' \subseteq \Delta_k$ there exists $S \subseteq T_\alpha$ such that $\type{R^{i+1}S}_k = \Delta'$:}
  \Leftrightarrow \; &
  \forall S \subseteq T_\alpha \colon \forall \tau_0 \in \Delta_0 \colon \exists w \in T_\beta \colon \type{R^{i}w}_0 = \{\tau_0\} \text{ and } \forall v \in R^{i}w \colon \type{Rv}_k = \type{R^{i+1}S}_k                             \\
  \intertext{For each $S$, gather the respective $w$ in a team $U \subseteq T_\beta$:}
  \Leftrightarrow \; &
  \forall S \subseteq T_\alpha \colon \exists U \subseteq T_\beta \colon \left( \forall \tau_0 \in \Delta_0 \colon \exists w \in U \colon \type{R^{i}w}_0 = \{\tau_0\}\right)                                                         \\
                     & \qquad\qquad  \text { and } \forall v \in R^{i}U \colon \type{Rv}_k = \type{R^{i+1}S}_k                                                                                                                        \\
  \Leftrightarrow \; &
  \forall S \subseteq T_\alpha : \exists U \subseteq T_\beta \colon U\text{ is $0$-canonical with offset }i                                                                                                                         \\
                     & \qquad\qquad  \text { and } \forall v \in R^{i}U \colon \type{Rv}_k = \type{R^{i+1}S}_k                                                                                                                        \\
  \intertext{By the base case $k = 0$, and since $U = {(T^{\alpha,\beta}_{S,U})}_\beta$:}
  \Leftrightarrow \; &
  \forall S \subseteq T_\alpha \colon \exists U \subseteq T_\beta \colon T^{\alpha,\beta}_{S,U} \vDash \rho^i_0(\beta) \text{ and }\forall v \in R^{i}U \colon \type{Rv}_k = \type{R^{i+1}S}_k                                        \\
  \intertext{By Theorem~\ref{thm:main-bisim}:}
  \Leftrightarrow \; &
  \forall S \subseteq T_\alpha \colon \exists U \subseteq T_\beta \colon T^{\alpha,\beta}_{S,U} \vDash \rho^i_0(\beta) \text{ and } \forall v \in R^{i}U \colon {(R^{i+1}T)}^{\alpha,\beta}_{R^{i+1}S,Rv} \vDash \chi^*_k(\alpha,\beta) \\
  \intertext{By Proposition~\ref{prop:of-scopes} (5.):}
  \Leftrightarrow \; &
  \forall S \subseteq T_\alpha \colon \exists U \subseteq T_\beta \colon T^{\alpha,\beta}_{S,U} \vDash \rho^i_0(\beta) \text{ and } \forall v \in R^{i}U \colon {(R^{i}T)}^{\alpha,\beta}_{R^{i}S,v} \vDash \Box \chi^*_k(\alpha,\beta) \\
  \intertext{By Proposition~\ref{prop:quantifiers} applied to ${(R^{i}T)}^{\alpha,\beta}_{R^{i}S,R^{i}U}$:}
  \Leftrightarrow \; &
  \forall S \subseteq T_\alpha \colon \exists U \subseteq T_\beta \colon T^{\alpha,\beta}_{S,U} \vDash \rho^i_0(\beta) \text{ and } {(R^{i}T)}^{\alpha,\beta}_{R^{i}S,R^{i}U}  \vDash \qa{\beta} \Box \chi^*_k(\alpha,\beta)                \\
  \intertext{Again by Proposition~\ref{prop:of-scopes} (5.) and Proposition~\ref{prop:quantifiers}:}
  \Leftrightarrow \; &
  \forall S \subseteq T_\alpha \colon \exists U \subseteq T_\beta \colon T^{\alpha,\beta}_{S,U}\vDash \rho^i_0(\beta) \text{ and } R^i\left(T^{\alpha,\beta}_{S,U}\right)  \vDash \qa{\beta} \Box \chi^*_k(\alpha,\beta)            \\
  \Leftrightarrow \; &
  \forall S \subseteq T_\alpha \colon \exists U \subseteq T_\beta \colon  T^{\alpha,\beta}_{S,U} \vDash \rho^i_0(\beta) \land \Box^i \qa{\beta} \Box \chi^*_k(\alpha,\beta)                                                         \\
  \Leftrightarrow \; &
  T \vDash \qaa{\alpha} \, \qee{\beta} \, (\rho^i_0(\beta) \land \Box^i \qa{\beta} \Box \chi^*_k(\alpha,\beta))                                                                                                                     \\
  \Leftrightarrow \; &
  T \vDash \rho^i_{k+1}(\alpha,\beta)\text{.}\tag*{\qedhere}
 \end{align*}
\end{proof}

\subsection{Enforcing scopes}

As the next step, we lift the restriction of the $\fraka_i$ being scopes \emph{a priori}.
In a sense, this condition is definable in $\MTL$ as well.
For this, let $\Psi \subseteq \PS$ be disjoint from $\Phi$.
Then the formula below ensures that $\Psi$ is a set of disjoint scopes "up to height $k$".
\begin{align*}
 \scopes_{k}(\Psi) & \dfn \bigwedge_{\substack{x,y \in \Psi \\x\neq y }}\neg (x \land y) \land \bigwedge_{i=1}^{k} \Big((x \land \Box^i x) \lor (\neg x \land \Box^i \neg x )\Big)\text{.}
\end{align*}

The definition up to height $k$ is sufficient for our purposes, which follows from the next lemma.
\begin{lem}\label{lem:model-restrict-height}
 If $\varphi \in \MTL_k$, then $\varphi$ is satisfiable if and only if $\varphi \land \Box^{k+1}\bot$ is satisfiable.
\end{lem}
\begin{proof}
 As the direction from right to left is trivial, suppose $\varphi$ is satisfiable.
 By Corollary~\ref{cor:height-model}, it then has a model $(\calK,T)$ that is a directed forest of height at most $k$.
 But then $(\calK,T) \vDash \Box^{k+1}\bot$, since $R^{k+1}T = \emptyset$ and $(\calK,\emptyset)$ satisfies all $\ML$-formulas, including $\bot$.
\end{proof}

\begin{thm}\label{thm:bounded-canon}
 $\canon_k \land \scopes_{k}(\{\fraka_0,\ldots,\fraka_k\}) \land \Box^{k+1}\bot$ is satisfiable, but has only $k$-staircases as models.
\end{thm}
\begin{proof}
 By combining Proposition~\ref{prop:canon-sat}, Theorem~\ref{thm:canon-formula} and Lemma~\ref{lem:model-restrict-height}, the formula is satisfiable.
 Since in every model $(\calK,T)$ the propositions $\fraka_0,\ldots,\fraka_k$ must be disjoint scopes due to $\Box^{k+1} \bot$ and $\scopes_k$, we can apply Theorem~\ref{thm:canon-formula}.
\end{proof}

As for bisimilarity, it is open whether $(\Phi,k)$-canonicity can be defined in $\MTL^\Phi_k$ \emph{efficiently} without restricting the models to those with scopes.
Note that the results of this section alone do not imply that the brute force algorithm given in Theorem~\ref{thm:k-membership} is optimal, as there could possibly be a satisfiability algorithm that does not need to construct a model.
To show proper complexity theoretic hardness, we need to encode non-elementary computations in such models, to which we will proceed in the next sections.

\section{Defining an order on types}%
\label{sec:order}

In the previous section, we enforced $k$-canonicity with a formula, \ie, such that $\size{\Delta_k}$ different types are contained in the team.
In order to encode computations of length $\size{\Delta_k}$, we additionally need to be able to talk about an ordering of $\Delta_k$.

Let us call any finite strict linear ordering simply an \emph{order}.
We specify an order $\prec_k$ on $\Delta_k$,
and
 analogously to team bisimilarity, an order $\prec^*_k$ on $\pow{\Delta_k}$.
To begin with, let us first agree on some arbitrary order $<$ on $\Phi$, say, $p_1 < p_2 < \cdots < p_{\size{\Phi}}$.
Furthermore, if $\sqsubset$ is some order on $X$, then the \emph{lexicographic order} $\sqsubset^*$ on $\pow{X}$ is defined by
\begin{align*}
 X_1 \sqsubset^* X_2 \text{ iff }\exists x \in X_2 \setminus X_1 \text{ such that }\forall x' \in X \colon (x \sqsubset x') \Rightarrow (x' \in X_1 \Leftrightarrow x' \in X_2)\text{.}
\end{align*}
For example, let $X = \{0,1\}$ and $0 \sqsubset 1$.
Then $\emptyset \sqsubset^* \{ 0\} \sqsubset^* \{1\} \sqsubset^* \{0,1\}$.
The order $\prec_k$ depends on the propositions true in a world, and otherwise recursively on the lexicographic order of the image team:
\begin{align*}
 \tau \prec_0 \tau' \;     & \Leftrightarrow\; \tau \cap \Phi <^* \tau' \cap \Phi\text{,}                                                                                                \\
 \tau \prec_{k+1} \tau' \; & \Leftrightarrow\; \tau \cap \Phi <^* \tau' \cap \Phi \text{ or }(\tau \cap \Phi = \tau' \cap \Phi\text{ and }\calR \tau \;\prec_{k}^*\; \calR\tau')\text{.}
\end{align*}
It is easy to verify by induction that $\prec_k$ and $\prec^*_k$ are orders on $\Delta_k$ and $\pow{\Delta_k}$, respectively.

The next step is to prove that $\prec_k$ and $\prec^*_k$ are (efficiently) definable in $\MTL_k$.
For this, we pursue the same approach as for $\chi_k$ and $\chi^*_k$ in Section~\ref{sec:encoding}, and show that $\prec_k$ and $\prec^*_k$ are definable in formulas $\zeta_k$ and $\zeta^*_k$ in a mutually recursive fashion.
Since order is a binary relation, the formulas below are once more parameterized by two scopes.
\begin{align*}
 \zeta_0(\alpha,\beta)     & \dfn \bigvee_{p \in \Phi} \Big[ (\alpha \hook \neg p) \land (\beta \hook p) \land \bigwedge_{\substack{q \in \Phi                                                      \\q < p}} (\alpha \lor \beta) \hook \dep{q}\Big]\\[1.5mm]
 \zeta_{k+1}(\alpha,\beta) & \dfn \zeta_0(\alpha,\beta) \,\ovee\, \chi_0(\alpha,\beta) \land \, \Box \zeta^*_k(\alpha,\beta)                                                                        \\[2mm]
 \zeta^*_{k}(\alpha,\beta) & \dfn \qe{\fraka_k} \left(\qe{\beta} \chi_k(\fraka_k,\beta) \right) \land \left(\negg \qe{\alpha} \chi_k(\fraka_k,\alpha)\right)                                        \\
                           & \qquad\land \Big( \big (\chi^*_k(\alpha,\beta) \land (\alpha \lor \beta) \big) \,\lor \, \big(\qa{\alpha\lor\beta} \negg \zeta_k(\fraka_k,\alpha\lor\beta) \big) \Big)
\end{align*}

Note that we make use of the scopes $\fraka_0,\ldots,\fraka_k$ in the formula, and in the following we restrict ourselves to $k$-staircase models.
Moreover, in the subformula $\zeta_k(\fraka_k,\alpha\lor\beta)$, we use the fact that $\alpha\lor\beta$ is a scope whenever $\alpha,\beta$ are scopes.

We require the next lemma for the correctness of $\zeta_k$ and $\zeta^*_k$.
Intuitively, it states that $\MTL_k$ is invariant under substitution of "locally equivalent" $\ML$-formulas.
\begin{lem}\label{lem:substitution}
 Let $\alpha,\beta \in \ML$ and $\varphi \in \MTL_k$.
 Let $T$ be a team such that $R^{i}T \vDash \alpha \leftrightarrow \beta$ for all $i\in\{0,\ldots,k\}$.
 Then $T \vDash \varphi$ if and only if $T \vDash \mathrm{Sub}(\varphi,\alpha,\beta)$, where $\mathrm{Sub}(\varphi,\alpha,\beta)$ is the formula obtained from $\varphi$ by substituting every occurrence of $\alpha$ with $\beta$.
\end{lem}
\begin{proof}
By straightforward induction; see the appendix.
\end{proof}

The following theorem states that in the class of $k$-staircase models (see the previous section) $\zeta_k$ and $\zeta^*_k$ define the required orders.

\begin{thm}\label{thm:order-theorem}
 Let $k \geq 0$, and let $(\calK,T)$ be a $k$-staircase with disjoint scopes
 $\alpha,\beta,\fraka_0,\ldots,\fraka_{k}$.
 If $w \in T_\alpha$ and $v \in T_\beta$, then
 \begin{alignat*}{4}
   &  & T^{\alpha,\beta}_{w,v} & \;\vDash\; \zeta_k(\alpha,\beta) \;   &  & \text{ if and only if } \; & \type{w}_k        & \prec_k \type{v}_k\text{,}         \\
   &  & T                      & \;\vDash\; \zeta^*_k(\alpha,\beta) \; &  & \text{ if and only if } \; & \type{T_\alpha}_k & \prec^*_k \type{T_\beta}_k\text{.}
 \end{alignat*}
 Furthermore, both $\zeta_k(\alpha,\beta)$ and $\zeta^*_k(\alpha,\beta)$ are $\MTL_k$-formulas that are constructible in space $\bigO{\log(k + \size{\Phi} + \size{\alpha}+\size{\beta})}$.
\end{thm}
We first give a rough idea of the proof, and after a series of required lemmas fully prove the theorem.
The definition of $\zeta_{k+1}$ simply follows the definition of $\prec_{k+1}$.
Furthermore, the formula $\zeta^*_k$ implements the lexicographic order $\prec^*_k$ as follows.
As shown in Figure~\ref{fig:order}, we first choose some $z  \in T_{\fraka_k}$ that acts as an \emph{pivot} to determine if $\type{T_\alpha}_k \prec^*_k \type{T_\beta}$, in the sense that it is the $\prec_k$-maximal type in which $T_\alpha$ and $T_\beta$ differ.\footnote{Since the pivot is selected from $T_{\fraka_k}$, at this point it is crucial that the underlying structure is a $k$-staircase.}
The first line of $\zeta^*_k$ indeed expresses that $\type{z}_k \in \type{T_\beta}_k \setminus \type{T_\alpha}_k$.

The disjunction in the second line intuitively states that we then can "split off" the subteam of $T_\alpha \cup T_\beta$ consisting of the elements $\prec_k$-greater than $z$ (the solid green area in Figure~\ref{fig:order}), while $\chi^*_k$ ensures that they agree on the contained types (this reflects the part after the quantifier in the definition of $\sqsubset^*$).
To achieve this, the subformula $\qa{\alpha\lor\beta} \negg \zeta_k(\fraka_k,\alpha\lor\beta)$ stipulates that any "remaining" elements from $T_\alpha \cup T_\beta$ possess only types not $\prec_k$-greater than $\type{z}_k$ (the dashed green area in the figure).

Here, Lemma~\ref{lem:substitution} is applied, as it ensures that after processing $\qa{\alpha\lor\beta}$ the formula $\zeta_k(\fraka_k,\alpha\lor\beta)$ in fact behaves as either $\zeta_k(\fraka_k,\alpha)$ or $\zeta_k(\fraka_k,\beta)$; and hence behaves correctly by induction hypothesis.

\begin{figure}[t]\centering
\begin{tikzpicture}[node distance=2mm,scale=.85]
\tikzset{world/.style={draw,circle,inner sep=.5mm,black,fill}}
\tikzset{dummy/.style={inner sep=.6mm}}
\tikzset{team/.style={draw,rounded corners,thick,inner sep=2mm}}

\foreach \i in {1,3} {
    \node[world] (w\i) at (\i,0) {};
}
\foreach \i in {2,6,4,5} {
    \node[world,cross out] (w\i) at (\i,0) {};
}
\foreach \i in {1,...,5} {
    \node[opacity=.5] at ({\i+0.55},0) {$\scriptstyle\succ_k$};
}

\node[right= of w6] (tb) {$\blue{T_\beta} \;\hat=\; \mathbf{10\underline{1}}000$};

\begin{scope}[yshift=4cm]
\foreach \i in {1,4,5} {
    \node[world] (v\i) at (\i,0) {};
}
\foreach \i in {2,3,6} {
    \node[world,cross out] (v\i) at (\i,0) {};
}
\foreach \i in {1,...,5} {
    \node[opacity=.5] at ({\i+0.55},0) {$\scriptstyle\succ_k$};
}

\node[right= of v6] (ta) {$\red{T_\alpha} \;\hat=\; \mathbf{10\underline{0}}110$};
\end{scope}

\node[dummy] at (7,2.2) (lpos) {};
\node[dummy] at (8.5,2.2) (rpos) {};
\node[world] at (7.5,2.2) (pos) {};
\node[team,fit=(lpos)(rpos)] {};
\node[team,fit=(pos),fill=black,fill opacity=.05,draw=black,draw opacity=.4] {};

\node[team,fit=(w1)(w6),fill=blue,fill opacity=.1,draw=blue] (teamA){};
\node[team,fit=(v1)(v6),fill=red,fill opacity=.1,draw=red] (teamB){};
\node[below right=2mm and -2mm of rpos,color=black] (tc) {$T_{\fraka_k}$};
\node[right=2mm of pos,color=black] {$z$};

\draw[dashed,thick,draw=green!50!black]
(ta.east) to[out=0,in=0] node[color=green!50!black,right] {$\prec^*_k$\;?} (tb.east)
;

\node[team,fit=(w1)(v2),fill=green,fill opacity=.05,draw=green!40!black,draw opacity=.4,inner sep=1mm] (comp1){};

\node[team,fit=(w3)(v6),fill=green,fill opacity=.05,draw=green!40!black,draw opacity=.4,inner sep=1mm,dashed] (comp2){};

\draw[dotted,thick,draw=green!50!black,->]
(pos) edge[bend right=10] node[color=green!50!black,above,pos=.5] {$\succ_k$} (comp1)
(pos) edge[bend right=10] node[color=green!50!black,above,pos=.6] {$\bis_k\;$} (w3)
(pos) edge[bend right=15] node[color=green!50!black,above,pos=.6] {$\quad\preceq_k$} ([yshift=1cm]comp2.east)
;
\end{tikzpicture}
\caption{The pivot $z \in T_{\fraka_k}$ determines that $\type{T_\alpha}_k \prec^*_k \type{T_\beta}_k$.
The subteam of $T_{\alpha\lor\beta}$ of worlds $\prec_k$-greater than $z$ must satisfy $\chi^*_k(\alpha,\beta)$.}\label{fig:order}
\end{figure}

\begin{defi}
 Let $k \geq 0$.
 Let $\alpha,\beta$ be disjoint scopes and $T$ a team in a Kripke structure.
 Then $\alpha$ and $\beta$ are called \emph{$\prec_k$-comparable in $T$} if for all $w \in T_\alpha,v\in T_\beta$
 \begin{align*}
  T^{\alpha,\beta}_{w,v} & \vDash \zeta_k(\alpha,\beta) \text{ iff } \type{w}_k \prec_k \type{v}_k\text{ and} \\
  T^{\alpha,\beta}_{w,v} & \vDash \zeta_k(\beta,\alpha) \text{ iff } \type{v}_k \prec_k \type{w}_k\text{.}
 \end{align*}
 Likewise, $\alpha$ and $\beta$ are \emph{$\prec^*_k$-comparable in $T$} if
 \begin{align*}
  T & \vDash \zeta^*_k(\alpha,\beta) \text{ iff }\type{T_\alpha}_k \prec^*_k \type{T_\beta}_k\text{ and}             \\
  T & \vDash \zeta^*_k(\beta,\alpha) \text{ iff }\type{T_\beta}_k \prec^*_k \type{T_\alpha}_k\text{.}\tag*{\qedhere}
 \end{align*}
\end{defi}

\medskip

The next lemma shows that the correctness of $\prec^*_k$ follows from that of $\prec_k$.

\begin{lem}\label{lem:order1}
 Suppose that $(\calK,T)$ is a $k$-staircase with disjoint scopes $\alpha,\beta,\fraka_0,\ldots,\fraka_{k}$.
 If both $\alpha$ and $\beta$ are $\prec_k$-comparable to $\fraka_k$ in all subteams $S$ of the form $T_{\fraka_0}\cup \cdots \cup T_{\fraka_{k-1}} \subseteq S \subseteq T$, then $\alpha$ and $\beta$ are $\prec_k^*$-comparable in $T$.
\end{lem}
\begin{proof}
Assuming $\calK,T,\alpha,\beta,\fraka_0,\ldots,\fraka_k$ as above, the proof is split into the following claims.

 \begin{claim}[a]
  The disjoint scopes $\alpha \lor \beta$ and $\fraka_k$ are $\prec_k$-comparable in any team $S$ that satisfies $T_{\fraka_0}\cup \cdots \cup T_{\fraka_{k-1}} \subseteq S \subseteq T$.
 \end{claim}
 \begin{cproof}
  Let $w \in S_{\alpha\lor\beta}$ and $v \in S_{\fraka_k}$.
  \Wloss $w \in S_\alpha$ (the case $w \in S_\beta$ works analogously).
  Then
  \begin{align*}
                      & S^{\alpha\lor\beta,\fraka_k}_{w,v} \vDash \zeta_k(\alpha\lor\beta,\fraka_k)                                                                                                                                                       \\
   \Leftrightarrow \; & S^{\alpha,\beta,\fraka_k}_{w,\emptyset,v} \vDash \zeta_k(\alpha\lor\beta,\fraka_k) \tag{Since $S^{\alpha\lor\beta,\fraka_k}_{w,v} = S^{\alpha,\beta,\fraka_k}_{w,\emptyset,v}$}                                                   \\
   \Leftrightarrow \; & S^{\alpha,\beta,\fraka_k}_{w,\emptyset,v} \vDash \zeta_k(\alpha,\fraka_k)\tag{By Lemma~\ref{lem:substitution}, as $\bigcup_{i=0}^k R^{i}S^{\alpha,\beta,\fraka_k}_{w,\emptyset,v} \vDash \alpha \leftrightarrow (\alpha\lor\beta)$} \\ % chktex 1
   \Leftrightarrow \; & \type{w}_k \prec_k \type{v}_k\text{.}\tag{By assumption of the lemma}
  \end{align*}
  The case $\zeta_k(\fraka_k,\alpha \lor \beta)$ is symmetric.
 \end{cproof}

 For the remaining proof, we omit the subscript $k$ when referring to types and $\prec$.
 Furthermore, for all $\tau \in \Delta_k$, let $\type{T}^{\tau}$ denote the restriction of $\type{T}$ to types $\tau'$ such that $\tau' \succ \tau$.
 Intuitively, these types are the "more significant positions" for the lexicographic ordering.
 In the next claim, we essentially show that the second line in the definition of $\zeta^*_k(\alpha,\beta)$ can be expressed as a statement of the form $\type{T_\alpha}^\tau = \type{T_\beta}^\tau$.
 \begin{claim}[b]
  Let $T$ be a team and $\tau \in \Delta_k$.
  Then $\type{T_\alpha}^\tau = \type{T_\beta}^\tau$ if and only if there exists $S \subseteq T_{\alpha\lor\beta}$ such that $\type{S_\alpha} = \type{S_\beta}$ and $\type{w} \nsucc \tau$ for all $w \in T_{\alpha\lor\beta} \setminus S$.
 \end{claim}
 \begin{cproof}
  "$\Rightarrow$":
  Let $S \dfn \{ v \in T_{\alpha\lor\beta} \mid \type{v} \succ \tau  \}$.
  Then $\type{S_\alpha} = \type{T_\alpha}^\tau = \type{T_\beta}^\tau = \type{S_\beta}$.
  Moreover, for every $w \in T_{\alpha\lor\beta} \setminus S$ clearly $\type{w} \nsucc \tau$ holds.

  "$\Leftarrow$":
  Assume that $S$ exists as stated in the claim.
  By symmetry, we only prove $\type{T_\alpha}^\tau \subseteq \type{T_\beta}^\tau$.
  Consequently, let $w \in T_\alpha$ such that $\type{w} \in \type{T_\alpha}^\tau$.
  Then $\type{w} \succ \tau$ by definition.
  But then $w \notin T_{\alpha\lor\beta}\setminus S$.
  However, we have $w \in T_\alpha$, hence $w \in T_{\alpha\lor\beta}$, which only leaves the possibility $w \in S$.
  Combining $w \in S$ and $w \in T_\alpha$ yields $w \in S_\alpha$, which by assumption also implies $\type{w} \in \type{S_\beta}$.
  As $\type{S_\beta} \subseteq \type{T_\beta}$ and $\type{w} \succ \tau$, the membership $\type{w} \in \type{T_\beta}^\tau$ follows.
 \end{cproof}

 \begin{claim}[c]
   $\alpha$ and $\beta$ are $\prec^*_k$-comparable in $T$.
 \end{claim}
 \begin{cproof}
 Due to symmetry, we prove only that $T \vDash \zeta^*_k(\alpha,\beta)$ iff $\type{T_\alpha}_k \prec^*_k \type{T_\beta}_k$.
 \begin{align*}
                     & \type{T_\alpha} \prec^* \type{T_\beta}                                                                                                                                                                              \\
  \Leftrightarrow \; &
  \exists \tau \in \type{T_\beta} \setminus \type{T_\alpha} \colon \forall \tau' \in \Delta, \tau \prec \tau' \colon\tau' \in \type{T_\alpha} \Leftrightarrow \tau' \in \type{T_\beta}                                                     \tag{Definition of $\prec^*_k$}\\
  \Leftrightarrow \; &
  \exists \tau \in \type{T_\beta} \setminus \type{T_\alpha} \colon \type{T_\alpha}^{\tau} = \type{T_\beta}^{\tau}                                                                                                                    \tag{Definition of $\type{\cdot}^\tau$}      \\
  \intertext{Since $T_{\fraka_k}$ is $k$-canonical, for every $\tau \in \Delta$ there exists $z \in T_{\fraka_k}$ of type $\tau$:}
  \Leftrightarrow \; &
  \exists z \in T_{\fraka_k} \colon \type{T_\alpha}^{\type{z}} = \type{T_\beta}^{\type{z}} \text{ and }  \type{z} \in \type{T_\beta} \setminus \type{T_\alpha}                                                                             \\
  \Leftrightarrow \; &
  \exists z \in T_{\fraka_k} \colon \type{T_\alpha}^{\type{z}} = \type{T_\beta}^{\type{z}} \text{ and } \exists x \in T_\beta \colon \type{z} = \type{x} \text{ and }\nexists y \in T_\alpha \colon \type{z} = \type{y}                    \\
  \intertext{As $\alpha,\beta$ and $\fraka_k$ are disjoint, we have $T_\alpha = O_\alpha$, where $O \dfn T^{\fraka_k}_z$, and likewise $T_\beta = O_\beta$:}
  \Leftrightarrow \; &
  \exists z \in T_{\fraka_k} \colon \type{O_\alpha}^{\type{z}} = \type{O_\beta}^{\type{z}} \text{ and } \exists x \in O_\beta \colon \type{z} = \type{x} \text{ and }\nexists y \in O_\alpha \colon \type{z} = \type{y}                    \\
  \Leftrightarrow \; &
  \exists z \in T_{\fraka_k} \colon \exists x \in O_\beta \colon \type{z} = \type{x} \text{ and }\nexists y \in O_\alpha \colon \type{z} = \type{y}                                                                                       \\
                     & \text{ and } \exists S \subseteq O_{\alpha\lor\beta} \colon \type{S_\alpha} = \type{S_\beta} \text{ and } \forall w \in O_{\alpha\lor\beta}\setminus S \colon \type{z} \nprec \type{w}                 \tag{by Claim (b)}              \\
  \intertext{Clearly $S$ is a subteam of $O_{\alpha\lor\beta}$ if and only if it is a subteam of $O$ and satisfies $\alpha \lor \beta$:}
  \Leftrightarrow \; &
  \exists z \in T_{\fraka_k} \colon \exists x \in O_\beta \colon \type{z} = \type{x} \text{ and }\nexists y \in O_\alpha \colon \type{z} = \type{y}                                                                                        \\
                     & \text{ and } \exists S \subseteq O \colon  \type{S_\alpha} = \type{S_\beta} \text{ and } S \vDash \alpha \lor \beta \text{ and } \forall w \in O_{\alpha\lor\beta}\setminus S \colon \type{z} \nprec \type{w}       \\
  \intertext{Letting $U = O \setminus S$, we have $O_{\alpha\lor\beta} \setminus S = U_{\alpha\lor\beta}$:}
  \Leftrightarrow \; &
  \exists z \in T_{\fraka_k} \colon \exists x \in O_\beta \colon \type{z} = \type{x} \text{ and }\nexists y \in O_\alpha \colon \type{z} = \type{y} \text{ and } \exists S\subseteq O \colon                                               \\
                     & \type{S_\alpha} = \type{S_\beta} \text{ and } S \vDash \alpha\lor\beta  \text{ and } \exists U \subseteq O \colon U = O \setminus S \text{ and }  \forall w \in U_{\alpha\lor\beta} \colon \type{z} \nprec \type{w} \\
  \intertext{Clearly, the property $\forall w \in U_{\alpha\lor\beta} : \type{z} \nprec \type{w}$ is preserved when taking subteams of $U$. Hence, $U = O \setminus S$ satisfies it if and only if some (not necessarily proper) superteam $U'$ of $O \setminus S$ does:}
  \Leftrightarrow \; &
  \exists z \in T_{\fraka_k} \colon \exists x \in O_\beta \colon \type{z} = \type{x} \text{ and }\nexists y \in O_\alpha \colon \type{z} = \type{y}                                                                                        \\
                     & \text{ and } \exists S \subseteq O \colon \type{S_\alpha} = \type{S_\beta} \text{ and } S \vDash \alpha\lor\beta                                                                                                    \\
                     & \qquad \text{ and } \exists U' \subseteq O \colon  U' \supseteq O \setminus S  \text{ and }  \forall w \in U'_{\alpha\lor\beta} \colon \type{z} \nprec \type{w}                                                     \\
  \intertext{By Theorem~\ref{thm:main-bisim}:}
  \Leftrightarrow \; &
  \exists z \in T_{\fraka_k} \colon O \vDash (\qe{\beta} \chi_k(\fraka,\beta)) \land (\negg \qe{\alpha} \chi_k(\fraka,\alpha))\text{ and } \exists S \subseteq O \colon                                                                    \\
                     & S \vDash (\alpha\lor\beta) \land \chi^*_k(\alpha,\beta) \text{ and } \exists U' \subseteq O \colon U' \supseteq O \setminus S  \text{ and }  \forall w \in U'_{\alpha\lor\beta} \colon \type{z} \nprec \type{w}     \\
  \intertext{Note that $T_{\fraka_0},\ldots, T_{\fraka_{k-1}}$ are retained in $O$.
   Moreover, $S \subseteq O_{\alpha\lor\beta}$, which implies that they are still subteams of $O \setminus S$ and hence of $U'$.
   But by Claim (a), $\alpha \lor \beta$ and $\fraka_k$ are then $\prec_k$-comparable scopes in $U'$ and we can replace $\type{z} \nprec \type{w}$:}
  \Leftrightarrow \; &
  \exists z \in T_{\fraka_k} \colon O \vDash (\qe{\beta} \chi_k(\fraka,\beta)) \land (\negg \qe{\alpha} \chi_k(\fraka,\alpha))                                                                                                             \\
                     & \text{ and } \exists S \subseteq O \colon S \vDash (\alpha\lor\beta) \land \chi^*_k(\alpha,\beta)                                                                                                                   \\
                     & \qquad  \text{ and } \exists U' \subseteq O \colon U' \supseteq  O \setminus S   \text{ and }  \forall w \in U'_{\alpha\lor\beta} \colon {(U')}^{\alpha\lor\beta}_{w} \vDash \negg \zeta_k(\fraka_k,\alpha\lor\beta)  \\
  \intertext{Recalling that $O = T^{\fraka_k}_z$, and by Proposition~\ref{prop:quantifiers}, we obtain:}
  \Leftrightarrow \; &
  \exists z \in T_{\fraka_k} \colon T^{\fraka_k}_z \vDash (\qe{\beta} \chi_k(\fraka,\beta)) \land (\negg \qe{\alpha} \chi_k(\fraka,\alpha))                                                                                                \\
                     & \text{ and } \exists S \subseteq T^{\fraka_k}_z \colon S \vDash (\alpha\lor\beta) \land \chi^*_k(\alpha,\beta)                                                                                                      \\
                     & \qquad  \text{ and } \exists U' \subseteq T^{\fraka_k}_z \colon U' \supseteq  T^{\fraka_k}_z \setminus S  \text{ and } U' \vDash \qa{\alpha\lor\beta} \negg \zeta_k(\fraka_k,\alpha\lor\beta)                       \\
  \Leftrightarrow \; &
  T \vDash \qe{\fraka_k} (\qe{\beta} \chi_k(\fraka,\beta)) \land (\negg \qe{\alpha} \chi_k(\fraka,\alpha))                                                                                                                                 \\
                     & \land \big((\alpha\lor\beta) \land \chi^*_k(\alpha,\beta)\big) \lor \big(\qa{\alpha \lor \beta} \negg \zeta_k(\fraka_k,\alpha\lor\beta)\big)                                                                        \\
  \Leftrightarrow \; & T \vDash \zeta^*(\alpha,\beta)\text{.}\tag*{\altqed\oldqed}
 \end{align*}
 \let\altqed\relax
\end{cproof}
\let\qed\relax
\end{proof}

In the next lemma, we prove the converse direction of Lemma~\ref{lem:order1}.

\begin{lem}\label{lem:order2}
 Let $k > 0$, and let $(\calK,T)$ be a $k$-staircase with disjoint scopes $\alpha,\beta,\fraka_0,\ldots,\fraka_{k-1}$.
 Then $\alpha$ and $\beta$ are $\prec_k$-comparable in every subteam $S$ of $T$ that contains $T_{\fraka_0} \cup \cdots \cup T_{\fraka_{k-1}}$.
\end{lem}

\begin{proof}
 The proof is by induction on $k$.
 Disjoint scopes $\alpha$ and $\beta$ are always $\prec_0$-comparable, which can be easily seen in $\zeta_0$.
 For the inductive step to $k + 1$, assume $(\calK,T)$ and $S$ as above, and let $\calK = (W, R, V)$.
Let $O \dfn S^{\alpha,\beta}_{w,v}$ with $w \in S_\alpha,v \in S_\beta$ arbitrary.

 \begin{claim}[a]
  $\alpha$ and $\beta$ are $\prec^*_k$-comparable in $RO$.
 \end{claim}
 \begin{cproof}
In the inductive step, now $\fraka_0,\ldots,\fraka_k,\alpha,\beta$ are disjoint scopes.
Additionally, $(\calK,RT)$ is a $k$-staircase.
  In particular, in the induction step $\alpha$ and $\beta$ are disjoint from $\fraka_k$.
  For this reason, $(\calK,RO)$ is a $k$-staircase as well, as ${(RO)}_{\fraka_0 \lor \cdots \lor \fraka_{k}} = {(RT)}_{\fraka_0 \lor \cdots \lor \fraka_{k}}$.

  Hence, by induction hypothesis, for every team $U$ such that $RO_{\fraka_0}\cup \cdots \cup RO_{\fraka_{k-1}} \subseteq U \subseteq RO$, we obtain that $\fraka_k$ and $\alpha$ are $\prec_k$-comparable in $U$, as well as $\fraka_k$ and $\beta$.
  Consequently, we can apply Lemma~\ref{lem:order1}, which proves the claim.
 \end{cproof}

We proceed with the induction step.
Again by symmetry, we only show that $O \vDash \zeta_{k+1}(\alpha,\beta)$ iff $\type{w}_{k+1} \prec_{k+1} \type{v}_{k+1}$.
 We distinguish three cases \wrt $\prec_0$:
 \begin{itemize}
  \item If $\type{w}_0 \prec_0 \type{v}_0$, then $O \vDash \zeta_0(\alpha,\beta)$ by the induction basis.
        As the former implies $\type{w}_{k+1} \prec_{k+1} \type{v}_{k+1}$ and the latter  $O \vDash \zeta_{k+1}(\alpha,\beta)$, the equivalence holds.\smallskip
  \item If $\type{w}_0 \succ_0 \type{v}_0$, then $\type{w}_{k+1} \nprec_{k+1} \type{v}_{k+1}$.
        Moreover, $O \nvDash \zeta_0(\alpha,\beta)$ by induction basis.
        Additionally, $O \nvDash \chi_0(\alpha,\beta)$ by Theorem~\ref{thm:main-bisim}.
        Consequently, both sides of the equivalence are false.

  \item If $\type{w}_0 = \type{v}_0$, then $O \vDash \chi_0(\alpha,\beta)$ by Theorem~\ref{thm:main-bisim}, but $O \nvDash \zeta_0(\alpha,\beta)$ by induction basis.
        Consequently, $O \vDash \zeta_{k+1}(\alpha,\beta)$ iff $O \vDash \Box\zeta^*_k(\alpha,\beta)$.
        Also, $\type{w}_{k+1} \prec_{k+1} \type{v}_{k+1}$ iff $\calR\type{w}_{k+1} \prec^*_k \calR\type{v}_{k+1}$.
The following equivalence concludes the proof:
        \begin{align*}
                            & \calR\type{w}_{k+1} \prec^*_k \calR\type{v}_{k+1}
         \\
         \Leftrightarrow \; &
         \type{Rw}_k \prec^*_k \type{Rv}_k\tag{By Proposition~\ref{prop:det-types}}
         \\
         \Leftrightarrow \; &
         RO \vDash \zeta^*_k(\alpha,\beta)\tag{By Claim (a)}             \\
         \Leftrightarrow \; &
         O \vDash \Box\zeta^*_k(\alpha,\beta)\text{.}\tag*{\qedhere}
        \end{align*}
 \end{itemize}
\end{proof}

\noindent
With the above lemmas we are now in the position to prove Theorem~\ref{thm:order-theorem}:

\begin{proof}[Proof of Theorem~\ref{thm:order-theorem}]
 First, it is straightforward to construct $\zeta_k$ and $\zeta^*_k$ in logarithmic space.
 For the correctness, let $(\calK,T)$ be a model with disjoint scopes $\alpha,\beta,\fraka_0,\ldots,\fraka_k$ as in the theorem.
 By Lemma~\ref{lem:order2} it immediately follows that $\alpha$ and $\beta$ are $\prec_k$-comparable in $T$.
 The second part, that $\alpha$ and $\beta$ are $\prec^*_k$-comparable in $T$, follows from the combination of Lemma~\ref{lem:order1} and~\ref{lem:order2}.
\end{proof}

\section{Encoding non-elementary computations}%
\label{sec:reduction}

We combine all the previous sections and extend Theorem~\ref{thm:k-membership} and Corollary~\ref{cor:membership} by their matching lower bounds:

\begin{thm}\label{thm:mtl-main}
\hfill
\begin{itemize}
  \item $\SAT(\MTL)$ and $\VAL(\MTL)$ are complete for $\TOWERPOLY$.
  \item If $k \geq 0$, then $\SAT(\MTL_k)$ and $\VAL(\MTL_k)$ are complete for $\ATIMEALT{\exp_{k+1}}$.
\end{itemize}
\end{thm}

\noindent
The above complexity classes are complement-closed, and additionally $\MTL$ and $\MTL_k$ are syntactically closed under negation.
For this reason, it suffices to prove the hardness of $\SAT(\MTL)$ and $\SAT(\MTL_k)$, respectively.
Moreover, the case $k = 0$ is equivalent to $\SAT(\PTL)$ being $\AEXPPOLY$-hard, which was proven by Hannula et al.~\cite{ptl2017}.
Their reduction also works in logarithmic space.
Consequently, the result boils down to the following lemma:
\begin{lem}\label{lem:reduction}
  \hfill
  \begin{itemize}
\item If $L \in \TOWERPOLY$, then $L \leqlogm \SAT(\MTL)$.
\item If $k \geq 1$ and $L \in \ATIMEALT{\exp_{k+1}}$, then $L \leqlogm \SAT(\MTL_{k})$.
\end{itemize}
\end{lem}

\noindent
We devise for each $L$ a reduction $x \mapsto \varphi_x$ such that $\varphi_x$ is a formula that is satisfiable if and only if $x \in L$.
By assumption, there exists a single-tape alternating Turing machine $M$ that decides $L$ (for $L \in \TOWERPOLY$, \wloss $M$ is alternating as well).

Let $M$ have states $Q$, which is the disjoint union of $Q_\exists$ (\emph{existential states}), $Q_\forall$ (\emph{universal states}), $Q_\text{acc}$ (\emph{accepting states}) and $Q_\text{rej}$ (\emph{rejecting states}).
Also, $Q$ contains some initial state $q_0$.
Let $M$ have a finite tape alphabet $\Gamma$ with blank symbol $\flat \in \Gamma$, and
a transition relation $\delta$.

\smallskip

We design $\varphi_x$ in a fashion that forces its models $(\calK,T)$ to encode an accepting computation of $M$ on $x$.
Let us call any legal sequence of configurations of $M$ (not necessarily starting with the initial configuration) a \emph{run}.
Then, similarly as in Cook's theorem~\cite{cook}, we encode runs as square "grids" with a vertical "time" coordinate and a horizontal "space" coordinate in the model, \ie, each row of the grid represents a configuration of $M$.

\Wloss $M$ never leaves the input to the left, and there exists $N$ that is an upper bound on both the length of a configuration and the runtime of $M$.
Formally, a run of $M$ is then a function $C \colon {\{1,\ldots,N\}}^2 \to \Gamma \cup (Q \times \Gamma)$,
Here, $C(i,j) = c$ for $c \in \Gamma$ means that the $i$-th configuration (\ie, after $M$ performed $i-1$ transitions) contains the symbol $c$ in its $j$-th cell.
The same holds if $C(i,j) = (q,c)$ for $(q,c)\in Q \times \Gamma$, but then additionally the machine is in the state $q$ with its head visiting the $j$-th cell in the $i$-th configuration.
As an example, for a run $C$ from $M$'s initial configuration we have $C(1,1) = (q_0,x_1)$, $C(1,i) = x_i$ for $2 \leq i \leq n$, and $C(1,i) = \flat$ for $n < i \leq N$.

\smallskip

Due to the semantics of $\MTL$, such a run must be encoded in $(\calK,T)$ very carefully.
We let the team $T$ contain $N^2$ worlds $w_{i,j}$ in which the respective value of $C(i,j)$ is encoded as a propositional assignment.
However, we cannot simply pursue the standard approach of assembling a large $N\times N$-grid in the edge relation $R$ in order to compare successive configurations; by Corollary~\ref{cor:height-model}, we cannot force the model to contain $R$-paths longer than $\size{\varphi_x}$.
Instead, to define grid neighborship, we let $w_{i,j}$ encode $i$ and $j$ in its \emph{type}.
More precisely, we use the linear order $\prec_k$ on $\Delta_k$ we defined with the $\MTL_k$-formula $\zeta_k$ in the previous section.
Then, instead of using $\Box$ and $\Diamond$, we examine the grid by letting $\zeta_k$ judge whether a given pair of worlds is deemed (horizontally or vertically) adjacent.

\subsection{Encoding runs in a team}

Next, we discuss how runs $C \colon {\{1,\ldots,N\}}^2 \to \Gamma \cup (Q\times \Gamma)$ are encoded in $T$.
Given a world $w \in T$, we partition the image $Rw$ with two special propositions $\frakt \notin \Phi$ ("timestep") and $\frakp \notin \Phi$ ("position").
Then we assign to $w$ the pair $\ell(w) \dfn (i,j)$ such that $\type{{(Rw)}_\frakt}_{k-1}$ is the $i$-th element, and $\type{{(Rw)}_\frakp}_{k-1}$ is the $j$-th element in the order $\prec^*_{k-1}$.
We call the pair $\ell(w)$ the \emph{location} of $w$ (in the grid).

Accordingly, we fix $N \dfn \size{\pow{\Delta^\Phi_{k-1}}}$.
For the case of fixed $k$, $M$ has runtime bounded by $\exp_{k+1}(g(n))$ for a polynomial $g$.
Then taking $\Phi \dfn \{ p_1, \ldots, p_{g(n)} \}$ yields a sufficiently large coordinate space, as
\begin{align*}
 \exp_{k+1}(g(n)) = \exp_{k+1} & (\size{\Phi}) = 2^{\exp_{k-1}\left(2^\size{\Phi}\right)} \leq 2^{\exp^*_{k-1}\left(2^\size{\Phi}\right)} = 2^{\size{\Delta^\Phi_{k-1}}} = N
\end{align*}
by Proposition~\ref{prop:number-of-types}.
For runtime $\exp_{g(n)}(1)$ of $M$, we let $\Phi \dfn \emptyset$ and precompute $k \dfn g(\size{x}) + 1$, but otherwise proceed identically.

\medskip

Next, let $\Xi$ be a constant set of propositions disjoint from $\Phi$ that encodes the range of $C$ via some bijection $c \colon \Xi \to \Gamma \cup (Q \times \Gamma)$.
If a world $w$ satisfies exactly one proposition $p$ of those in $\Xi$, then by slight abuse of notation we write $c(w)$ instead of $c(p)$.
Intuitively, $c(w) \in \Gamma \cup (Q \times \Gamma)$ is the content of the grid cell represented by $w$.

\medskip

Using $\ell$ and $c$, the function $C$ can be encoded into a team $T$ as follows.
First, a team $T$ is called \emph{grid} if every point in $T$ satisfies exactly one proposition in $\Xi$, and if every location $(i,j) \in {\{1,\ldots,N\}}^2$ occurs as $\ell(w)$ for some point $w \in T$.
Moreover, a grid $T$ is called \emph{pre-tableau} if for every location $(i,j)$ \emph{and} every element $p \in \Xi$ there is some world $w \in T$ such that $\ell(w) = (i,j)$ and $w \vDash p$.
Finally, a grid $T$ is a \emph{tableau} if any two elements $w,w' \in T$ with $\ell(w) = \ell(w')$ also agree on $\Xi$, \ie, $c(w) = c(w')$.

Let us motivate the above definitions.
Clearly, the definition of a \emph{grid} $T$ means that $T$ captures the whole domain of $C$, and that $c$ is well-defined on the level of \emph{points}.
If $T$ is additionally a \emph{tableau}, then $c$ is also well-defined on the level of \emph{locations}.
In other words, a tableau $T$ induces a function $C_T\colon {\{1,\ldots,N\}}^2 \to \Gamma \cup (Q \times \Gamma)$ via $C_\alpha(i,j) \dfn c(w)$, where $w \in T$ is arbitrary such that $\ell(w) = (i,j)$.\label{p:tableau-function}

A \emph{pre-tableau} can be seen as the union of all possible $C$.
In particular, given any pre-tableau, the definition ensures that arbitrary tableaus can be obtained from it by the means of subteam quantification $\qee{}$ (cf.\ p.\ \pageref{p:quantifiers}).

A tableau $T$ is \emph{legal} if $C_T$ is a run of $M$, \ie, if every row is a configuration of $M$, and if every pair of two successive rows represents a valid $\delta$-transition.

\smallskip

The idea of the reduction is now to capture the alternating computation of $M$ by nesting polynomially many quantifications (via $\qee{}$ and $\qaa{}$) of legal tableaus, of which each one is the continuation of the computation of the previous one.

\subsection{Accessing two components of locations}

An discussed earlier, we choose to represent a location $(i,j)$ in a point $w$ as a pair $(\Delta',\Delta'')$ by stipulating that $\Delta' = \type{{(Rw)}_\frakt}_{k-1}$ and $\Delta'' = \type{{(Rw)}_\frakp}_{k-1}$.
To access the two components of a encoded location independently, we introduce the operator
\[
 |^\alpha_\frakq\,\psi \dfn (\alpha \land \neg \frakq) \lor ((\alpha \hook \frakq) \land \psi)\text{,}
\]
where $\frakq \in \{\frakt,\frakp\}$ and $\alpha \in \ML$.
It is easy to check that $T \vDash |^\alpha_\frakq\,\psi$ iff $T^\alpha_{T_\frakq} \vDash \psi$.

In order to \emph{compare} the locations of grid cells, for each component $\frakq \in \{\frakt,\frakp\}$ we define the following formulas:
$\psi^\frakq_\prec(\alpha,\beta)$ tests whether the location in $T_\alpha$ is less than the one in $T_\beta$ \wrt its $\frakq$-component (assuming singleton teams $T_\alpha$ and $T_\beta$).
Analogously, $\psi^\frakq_\equiv(\alpha,\beta)$ checks for equality of the respective component:
\begin{align*}
 \psi^\frakq_{\prec}(\alpha,\beta) \dfn \;  & \Box \,|^{\alpha}_\frakq |^{\beta}_\frakq \zeta_{k-1}^{*}(\alpha,\beta) \\
 \psi^\frakq_{\equiv}(\alpha,\beta) \dfn \; & \Box \,|^{\alpha}_\frakq |^{\beta}_\frakq \chi^*_{k-1}(\alpha,\beta)
\end{align*}
For this purpose, $\psi^\frakq_{\prec}$ is built upon the formula $\zeta^*_{k-1}$ from Theorem~\ref{thm:order-theorem}, while $\psi^\frakq_{\equiv}$ checks for equality with the help of $\chi^*_{k-1}$ from Theorem~\ref{thm:main-bisim}.

\medskip

\begin{claim}[a]
 Let $\calK$ be a structure with a team $T$ and disjoint scopes $\alpha$ and $\beta$.
 Suppose $w \in T_\alpha$ and $v \in T_\beta$, where $\ell(w) = (i_w,j_w)$ and $\ell(v) = (i_v,j_v)$. Then:
 \begin{align*}
  T^{\alpha,\beta}_{w,v} \vDash \psi^\frakt_\equiv(\alpha,\beta) & \;\Leftrightarrow\;i_w = i_v  \\
  T^{\alpha,\beta}_{w,v} \vDash \psi^\frakp_\equiv(\alpha,\beta) & \;\Leftrightarrow\; j_w = j_v
 \end{align*}
 Moreover, if $\alpha,\beta,\fraka_0,\ldots,\fraka_{k}$ are disjoint scopes in $\calK$ and $(\calK,T)$ is a $k$-staircase, then:
 \begin{align*}
  T^{\alpha,\beta}_{w,v} \vDash \psi^\frakt_\prec(\alpha,\beta) & \;\Leftrightarrow \;i_w < i_v \\
  T^{\alpha,\beta}_{w,v} \vDash \psi^\frakp_\prec(\alpha,\beta) & \;\Leftrightarrow\; j_w < j_v
 \end{align*}
\end{claim}
\begin{cproof}
 Let us begin with $\psi^\frakt_\equiv$ ($\psi^\frakp_\equiv$ works identically):
 \begin{align*}
  i_w = i_v \;
  \Leftrightarrow\; &
  \type{{(Rw)}_\frakt}_{k-1} = \type{{(Rv)}_\frakt}_{k-1} \tag{By Definition}                                                     \\
  \Leftrightarrow\; &
  RT^{\alpha,\beta}_{{(Rw)}_{\,\frakt},{(Rv)}_{\,\frakt}} \vDash \chi^*_{k-1}(\alpha,\beta) \tag{By Theorem~\ref{thm:main-bisim}} \\
  \Leftrightarrow\; &
  {\Big(RT^{\alpha,\beta}_{Rw,Rv}\Big)}^{\alpha,\beta}_{RT_\frakt,RT_\frakt} \vDash \chi^*_{k-1}(\alpha,\beta)                  \\
  \Leftrightarrow\; &
  RT^{\alpha,\beta}_{Rw,Rv} \vDash |^\alpha_\frakt|^\beta_\frakt \chi^*_{k-1}(\alpha,\beta)                                   \\
  \Leftrightarrow\; &
  T^{\alpha,\beta}_{w,v} \vDash \Box \, |^\alpha_\frakt|^\beta_\frakt \chi^*_{k-1}(\alpha,\beta) \tag{Proposition~\ref{prop:of-scopes}}
 \end{align*}
 Similarly for $\psi^\frakt_{\prec}$ ($\psi^\frakp_\prec$ again works identically):
 \begin{align*}
  i_w < i_v \;
  \Leftrightarrow\; &
  \type{{(Rw)}_\frakt}_{k-1} \prec^*_{k-1} \type{{(Rv)}_\frakt}_{k-1} \tag{By Definition}                                             \\
  \Leftrightarrow\; &
  RT^{\alpha,\beta}_{{(Rw)}_{\,\frakt},{(Rv)}_{\,\frakt}} \vDash \zeta^*_{k-1}(\alpha,\beta) \tag{By Theorem~\ref{thm:order-theorem}} \\
  \Leftrightarrow\; &
  {\Big(RT^{\alpha,\beta}_{Rw,Rv}\Big)}^{\alpha,\beta}_{T_\frakt,T_\frakt} \vDash \zeta^*_{k-1}(\alpha,\beta)                       \\
  \Leftrightarrow\; &
  RT^{\alpha,\beta}_{Rw,Rv} \vDash |^\alpha_\frakt|^\beta_\frakt \zeta^*_{k-1}(\alpha,\beta)                                      \\
  \Leftrightarrow\; &
  T^{\alpha,\beta}_{w,v} \vDash \Box \, |^\alpha_\frakt|^\beta_\frakt \zeta^*_{k-1}(\alpha,\beta) \tag*{(Proposition~\ref{prop:of-scopes})\;\altqed}
 \end{align*}\let\altqed\relax
\end{cproof}

\subsection{Defining grids, pre-tableaus, and tableaus}

Next, we aim at constructing formulas that check whether a given team is a grid, pre-tableau, or a tableau, respectively.

First, to check that every location $(i,j) \in {\{1,\ldots,N\}}^2$ of the grid occurs as $\ell(w)$ of some $w \in T$, we quantify over all corresponding pairs $(\Delta',\Delta'') \in \pow{\Delta_{k-1}}^2$.
To cover all these sets of types we can quantify, for instance, over the images of all points of $T_{\fraka_k}$.
However, as subteam quantifiers $\qee{},\qe{},\qaa{},\qa{}$ cannot pick \emph{two} subteams from the same scope, we enforce a $k$-canonical copy $\fraka'_k$ of $\fraka_k$ in the spirit of Theorem~\ref{thm:canon-formula}:
\begin{align*}
 \canon' \dfn \; & \rho^k_0(\fraka_0) \land \bigwedge_{m = 1}^k \rho^{k-m}_{m}(\fraka_{m-1},\fraka_m) \land \rho^0_{k}(\fraka_{k-1},\fraka'_k)
\end{align*}

\begin{claim}[b]
 If $\fraka_0,\ldots,\fraka_k,\fraka'_k$ are disjoint scopes in $\calK$, then
 $(\calK,T) \vDash \canon'$ if and only if $(\calK,T)$ is a $k$-staircase and $T_{\fraka'_k}$ is $k$-canonical.
 Moreover, $\canon' \land \scopes_{k}(\{\fraka_0,\ldots,\fraka_k,\fraka'_k\}) \land \Box^{k+1}\bot$ is satisfiable, but is only satisfied by $k$-staircases $(\calK,T)$ in which both $T_{\fraka_k}$ and $T_{\fraka'_k}$ are $k$-canonical.
 Furthermore, both formulas are constructible in space $\bigO{\log(\size{\Phi} + k)}$.
\end{claim}
\begin{cproof}
 Proven similarly to Theorem~\ref{thm:canon-formula} and~\ref{thm:bounded-canon}.
\end{cproof}

The next formula checks whether a given team is a grid.
More precisely, the subformula $\psi_\text{pair}$ compares the $\frakt$-component of the selected location in $\alpha$ to the image of the world quantified in $\fraka_k$, and its $\frakp$-component to $\fraka'_k$, respectively.
That every world satisfies exactly one element of $\Xi$ is guaranteed by $\psi_\text{grid}$ as well.
\begin{align*}
 \psi_\text{grid}(\alpha) \dfn \; & \Big(\alpha \hook \bigvee_{e \in \Xi}e \land \bigwedge_{\substack{e' \in \Xi                                                                            \\ e'\neq e}} \neg e'\Big) \land \qa{\fraka_k} \+ \qa{\fraka'_k} \+ \qe{\alpha} \,\psi_\text{pair}(\alpha)\\
 \psi_\text{pair}(\alpha) \dfn \; & \Box \left[ \big(\,|^\alpha_\frakt\; \chi^*_{k-1}(\fraka_{k},\alpha)\big) \land \big(\,|^\alpha_\frakp \; \chi^*_{k-1}(\fraka'_{k},\alpha)\big) \right]
\end{align*}

\medskip

In the following and all subsequent claims, we always assume that $T$ is a team in a Kripke structure $\calK$ such that $(\calK,T)$ satisfies $\canon' \land \Box^{k+1}\bot$.
Moreover, all stated scopes are always assumed pairwise disjoint in $\calK$ (as we can enforce this later in the reduction with $\scopes_k(\cdots)$).

\begin{claim}[c]
 $T \vDash \psi_{\text{grid}}(\alpha)$ if and only if $T_\alpha$ is a grid.
\end{claim}
\begin{cproof}
 Clearly $T \vDash \alpha \hook \bigvee_{e \in \Xi}e \land \bigwedge_{e' \in \Xi, e'\neq e} \neg e'$ if and only if every world $w \in T_\alpha$ satisfies exactly one element of $\Xi$.
 Consequently, for the proof it remains to show the following equivalence:
 \begin{align*}
                     & \forall (i,j) \in {\{1,\ldots,N\}}^2 \colon \exists w \in T_\alpha \colon \ell(w) = (i,j)                                                                                                                                          \\
  \Leftrightarrow \; &
  \forall \Delta',\Delta'' \subseteq \Delta_{k-1} \colon \exists w \in T_\alpha : \type{{(Rw)}_\frakt}_{k-1} = \Delta' \text{ and } \type{{(Rw)}_\frakp}_{k-1} = \Delta''                                                                                   \\
  \intertext{By $k$-canonicity of $\fraka_k,\fraka'_k$ due to Claim (b):}
  \Leftrightarrow \; &
  \forall v \in T_{\fraka_k},v' \in T_{\fraka'_k} \colon \exists w \in T_\alpha \colon  \type{{(Rw)}_\frakt}_{k-1} = \type{Rv}_{k-1} \text{ and } \type{{(Rw)}_\frakp}_{k-1} = \type{Rv'}_{k-1}                                                             \\
  \intertext{By Theorem~\ref{thm:main-bisim}:}
  \Leftrightarrow \; &
  \forall v \in T_{\fraka_k},v' \in T_{\fraka'_k} \colon \exists w \in T_\alpha \colon  RT^{\alpha,\fraka_k,\fraka'_k}_{{(Rw)}_\frakt,Rv,Rv'} \vDash \chi^*_{k-1}(\fraka_{k},\alpha)                                                                      \\
                     & \qquad \qquad \text{ and } RT^{\alpha,\fraka_k,\fraka'_k}_{{(Rw)}_\frakp,Rv,Rv'} \vDash \chi^*_{k-1}(\fraka'_{k},\alpha)                                                                                                           \\
  \Leftrightarrow \; &
  \forall v \in T_{\fraka_k},v' \in T_{\fraka'_k} \colon \exists w \in T_\alpha :  {\big(RT^{\alpha,\fraka_k,\fraka'_k}_{Rw,Rv,Rv'}\big)}^{\alpha}_{RT_\frakt} \vDash \chi^*_{k-1}(\fraka_{k},\alpha)                                                     \\
                     & \qquad \qquad \text{ and } {\big(RT^{\alpha,\fraka_k,\fraka'_k}_{Rw,Rv,Rv'}\big)}^\alpha_{RT_\frakp} \vDash \chi^*_{k-1}(\fraka'_{k},\alpha)                                                                                       \\
  \Leftrightarrow \; &
  \forall v \in T_{\fraka_k},v' \in T_{\fraka'_k}
  \colon \exists w \in T_\alpha \colon  RT^{\alpha,\fraka_k,\fraka'_k}_{Rw,Rv,Rv'} \vDash \,|^\alpha_{\frakt}\chi^*_{k-1}(\fraka_{k},\alpha) \land \,|^\alpha_{\frakp}\chi^*_{k-1}(\fraka'_{k},\alpha)                                                  \\
  \intertext{By Proposition~\ref{prop:of-scopes}:}
  \Leftrightarrow \; &
  \forall v \in T_{\fraka_k},v' \in T_{\fraka'_k} \colon \exists w \in T_\alpha \colon  T^{\alpha,\fraka_k,\fraka'_k}_{w,v,v'} \vDash \Box \,|^\alpha_{\frakt}\chi^*_{k-1}(\fraka_{k},\alpha) \land \,|^\alpha_{\frakp}\chi^*_{k-1}(\fraka'_{k},\alpha) \\
  \intertext{By Proposition~\ref{prop:quantifiers}:}
  \Leftrightarrow \; &
  T \vDash  \qa{\fraka_k}\qa{\fraka'_k}\qe{\alpha} \Box \,|^\alpha_{\frakt}\chi^*_{k-1}(\fraka_{k},\alpha) \land \,|^\alpha_{\frakp}\chi^*_{k-1}(\fraka'_{k},\alpha)                                                                                    \\
  \Leftrightarrow \; &
  T \vDash  \qa{\fraka_k}\qa{\fraka'_k}\qe{\alpha} \, \psi_{\text{pair}}(\alpha)\tag*{\altqed}
 \end{align*}
 \let\altqed\relax
\end{cproof}

\medskip

With slight modifications it is straightforward to define pre-tableaus:
\begin{align*}
 \psi_\text{pre-tableau}(\alpha) \dfn \; & \psi_\text{grid}(\alpha) \land \qa{\fraka_k} \+ \qa{\fraka'_k} \+ \bigwedge_{e \in \Xi} \qe{\alpha} \big(\psi_\text{pair}(\alpha) \land (\alpha \hook e)\big)
\end{align*}

\begin{claim}[d]
 $T \vDash \psi_{\text{pre-tableau}}(\alpha)$ if and only if $T_\alpha$ is a pre-tableau.
\end{claim}
\begin{cproof}
 Proven similarly to Claim (c).
\end{cproof}

\medskip

The other special case of a grid, that is, a \emph{tableau}, requires a more elaborate approach to define in $\MTL$.
The difference to a grid or pre-tableau is that we have to quantify over all \emph{pairs} $(w,w')$ of points in $T$, and check that they agree on $\Xi$ if $\ell(w) = \ell(w')$.
However, as discussed before, while $\qa{}$ can quantify over all points in a team, it cannot quantify over pairs.

As a workaround, we consider not only a tableau $T_\alpha$, but also a \emph{second} tableau that acts as a copy of $T_\alpha$.
Formally, for grids $T_\alpha,T_\beta$, let $T_\alpha \approx T_{\beta}$ denote that for all pairs $(w,w') \in T_\alpha \times T_{\beta}$ it holds that $\ell(w) = \ell(w')$ implies $c(w) = c(w')$.

As $\approx$ is symmetric and transitive, $T_\alpha \approx T_\beta$ in fact implies both $T_\alpha \approx T_\alpha$ and $T_\beta \approx T_\beta$, and hence that both $T_\alpha$ and $T_\beta$ are tableaus such that $C_{T_\alpha} = C_{T_\beta}$, where $C_{T_\alpha},C_{T_\beta} \colon {\{1,\ldots,N\}}^2 \to \Gamma \cup (Q \times \Gamma)$ are the induced runs as discussed on p.\ \pageref{p:tableau-function}.
\begin{align*}
 \psi_\text{tableau}(\alpha) \dfn \; & \psi_\text{grid}(\alpha) \land \, \qee{\gamma_0} \,\psi_\text{grid}(\gamma_0) \land \psi_\approx(\alpha,\gamma_0)                                                                     \\[2mm]
 \psi_\approx(\alpha,\beta) \dfn     & \qa{\alpha}\qa{\beta} \; \Big( \big(\psi^\frakt_\equiv(\alpha,\beta) \land \psi^\frakp_\equiv(\alpha,\beta) \big) \timp \bigovee_{e \in \Xi} ((\alpha \lor \beta) \hook e) \Big)
\end{align*}

In the following claim (and in the subsequent ones), we use the scopes $\gamma_0,\gamma_1,\gamma_2,\ldots$ as "auxiliary pre-tableaus".
Later, we will also use them as domains to quantify extra locations or rows from.
(The index of $\gamma_i$ is incremented whenever necessary to avoid quantifying from the same scope twice.)
For this reason, from now on we always assume, for sufficiently large $i$, that $T_{\gamma_i}$ is a pre-tableau.
This can be later enforced in the reduction with $\psi_\text{pre-tableau}(\gamma_i)$.

\begin{claim}[e]
  \hfill
  \begin{enumerate}
    \item $T \vDash \psi_\text{tableau}(\alpha)$ if and only if $T_\alpha$ is a tableau.
    \item For grids $T_\alpha,T_\beta$, it holds $T \vDash \psi_\approx(\alpha,\beta)$ if and only if $T_\alpha \approx T_\beta$.
  \end{enumerate}
\end{claim}
\begin{cproof}
 (2) follows straightforwardly from Claim (a).
 Let us consider (1).
 As $\psi_\text{tableau}$ implies $\psi_\text{grid}$, and by Claim (c), we can assume that $T_\alpha$ is a grid.

 Suppose that the formula is true.
 Then there exists $S \subseteq T_{\gamma_0}$ such that $T^\alpha_S \vDash \psi_\text{grid}(\gamma_0)$.
 By Claim (c), then $S$ is a grid as well.
 Moreover, $T_\alpha \approx S$ by (2).
 As argued above, this implies that $T_\alpha$ (and $S$) is a tableau.

 For the other direction, suppose that $T_\alpha$ is a tableau.
 Then it defines a function $C_{T_\alpha}$.
 Since $T_{\gamma_0}$ is a pre-tableau, we can pick a subteam $S$ of it that contains for each $(i,j) \in {\{1,\ldots,N\}}^2$ exactly those worlds $w$ with $\ell(w) = (i,j)$ such that $c(w) = C_{T_\alpha}(i,j)$.
 Then $T_\alpha \approx S$, and $\psi_\text{tableau}$ is true, with the quantifier $\qee{\gamma_0}$ witnessed by $S$.
\end{cproof}

\subsection{From tableaus to runs}

To ascertain that a tableau contains a run of $M$, we have to check whether each row indeed is a configuration of $M$---in other words, exactly one cell of each row contains a pair $(q,a) \in Q \times \Gamma$---and whether consecutive configurations obey the transition relation $\delta$ of $M$.

For this, in the spirit of Cook's theorem~\cite{cook} it suffices to consider all \emph{legal windows} in the grid, \ie, cells that are adjacent as follows, where $e_1,\ldots,e_6 \in \Gamma \cup (Q \times \Gamma)$:
\[
 \begin{tabular}{|c|c|c|} % chktex 44
  \hline % chktex 44
  $e_1$ & $e_2$ & $e_3$ \\
  \hline % chktex 44
  $e_4$ & $e_5$ & $e_6$ \\
  \hline % chktex 44
 \end{tabular}
\]

\medskip

If, say, $(q,a,q',a',R) \in \delta$---$M$ switches to state $q'$ from $q$, replacing $a$ on the tape by $a'$, and moves to the right---then the windows obtained by setting
$e_1 = e_4 = b$, $e_2 = (q,a)$, $e_5 = a'$, $e_3 = b'$, $e_6 = (q',b')$ are legal for all $b,b' \in \Gamma$.
Using this scheme, $\delta$ is completely represented by a constant finite set $\mathsf{win} \subseteq \Xi^6$ of tuples $(e_1,\ldots,e_6)$ that represent the allowed windows in a run of $M$.

Let us next explain how adjacency of cells is expressed.
Suppose that two points $w \in T_\alpha$ and $v \in T_{\beta}$ are given.
That $v$ is the immediate ($\frakt$- or $\frakp$-)successor of $w$ then means that no element of the order exists between them. % chktex 36
Simultaneously, $w$ and $v$ have to agree on the other component of their location, which is expressed by the first conjunct below.
If $\frakq \in \{\frakt,\frakp\}$ and $\overline{\frakq} \in \{\frakt,\frakp\} \setminus \{\frakq\}$,
 we define:
\begin{align*}
 \psi^\frakq_{\text{succ}}(\alpha,\beta) \dfn \; & \psi^{\overline{\frakq}}_{\equiv}(\alpha,\beta)\land \psi^\frakq_{\prec}(\alpha,\beta) \land \negg \qe{\gamma_0} \left(\psi^\frakq_{\prec}(\alpha,\gamma_0) \land \psi^\frakq_{\prec}(\gamma_0,\beta) \right)
\end{align*}

\begin{claim}[f]
 If $w \in T_\alpha$ and $v \in T_\beta$, then:
 \begin{align*}
  T^{\alpha,\beta}_{w,v} \vDash \psi^\frakt_\text{succ}(\alpha,\beta) & \Leftrightarrow \exists i,j \in \{1,\ldots,N\} \colon \ell(w) = (i,j) \text{ and }\ell(v) = (i+1,j) \\
  T^{\alpha,\beta}_{w,v} \vDash \psi^\frakp_\text{succ}(\alpha,\beta) & \Leftrightarrow \exists i,j \in \{1,\ldots,N\} \colon \ell(w) = (i,j) \text{ and }\ell(v) = (i,j+1)
 \end{align*}
\end{claim}
\begin{cproof}
 Let us consider only $\frakq = \frakt$, as the case $\frakq = \frakp$ is proven analogously.
 Assume that the formula $\psi^\frakt_\text{succ}(\alpha,\beta)$ is true in $T^{\alpha,\beta}_{w,v}$.
 By Claim (a), $\psi^{\frakp}_{\equiv}$ holds if and only if there is a unique $j$ such that $\ell(w) = (i,j)$ and $\ell(v) = (i',j)$, for some $i,i'$;
 in other words, if $w$ and $v$ agree on their $\frakp$-component.

 Next, consider the sets $\Delta_w \dfn \type{{(Rw)}_\frakt}_{k-1}$ and $\Delta_v \dfn \type{{(Rv)}_\frakt}_{k-1}$ which correspond to the $\frakt$-components of $\ell(w)$ and $\ell(v)$.
 Suppose that $\Delta_w$ is the $i$-th element of $\prec^*_{k-1}$.
 By $\psi^\frakt_\prec$ and Claim (a), then clearly $\Delta_v$ is the $i'$-th element for some $i' > i$.

 Suppose for the sake of contradiction that also $i' > i + 1$, and let then instead $\Delta' \subseteq \Delta_{k-1}$ be the $(i+1)$-th element of $\prec^*_{k-1}$.
 As $T_{\gamma_0}$ is a pre-tableau, it contains a world $z$ such that $\ell(z) = (i+1,j)$.
 But then $\psi^\frakq_{\prec}(\alpha,\gamma_0) \land \psi^\frakq_{\prec}(\gamma_0,\beta)$ is true in $T^{\alpha,\beta,\gamma_0}_{w,v,z}$, contradiction to $\psi^\frakt_\text{succ}$.
 Consequently, $i' = i + 1$.
 The direction from right to left is shown similarly.
\end{cproof}

To check all windows in the tableau $T_\alpha$, we need to simultaneously quantify elements from \emph{six} tableaus $T_{\gamma_1},\ldots,T_{\gamma_6}$ that are copies of $T_\alpha$.
For this purpose, we define
\begin{align*}
 \qeet{\gamma_i}{\alpha} \, \varphi \dfn \; & \qee{\gamma_i} \; \psi_\text{grid}(\gamma_i) \land \psi_\approx(\alpha,\gamma_i) \land \varphi\text{.}
\end{align*}
Intuitively, under the premise that $T_{\gamma_i}$ is a pre-tableau and $T_\alpha$ is a tableau, it "copies" the tableau $T_\alpha$ into $T_{\gamma_i}$ by shrinking $T_{\gamma_i}$ accordingly.
This is proven analogously to Claim (e).
The next formula states that the picked points are arranged as in the picture:
\begin{align*}
 \psi_\text{window}(\gamma_1,\ldots,\gamma_6) \dfn \; &
 \psi^\frakt_{\text{succ}}(\gamma_1,\gamma_4) \land
 \psi^\frakt_{\text{succ}}(\gamma_2,\gamma_5)  \land
 \psi^\frakt_{\text{succ}}(\gamma_3,\gamma_6) \; \land                                                                                                  \\
                                                      & \psi^\frakp_{\text{succ}}(\gamma_1,\gamma_2) \land \psi^\frakp_{\text{succ}}(\gamma_2,\gamma_3)
\end{align*}

\[
 \begin{tabular}{|c|c|c|} % chktex 44
  \hline % chktex 44
  $T_{\gamma_1}$ & $T_{\gamma_2}$ & $T_{\gamma_3}$ \\
  \hline % chktex 44
  $T_{\gamma_4}$ & $T_{\gamma_5}$ & $T_{\gamma_6}$ \\
  \hline % chktex 44
 \end{tabular}
\]

\medskip

\noindent{}The formula defining legal tableaus follows.
\begin{align*}
 \psi_{\text{legal}}(\alpha) & \dfn \psi_{\text{tableau}}(\alpha) \land \qeet{\gamma_1}{\alpha}\cdots \qeet{\gamma_6}{\alpha} \; \vartheta_1 \land \vartheta_2 \land \vartheta_3
 \intertext{We check that at most cell per row contains a state of $M$:}
 \vartheta_1                 & \dfn \qa{\gamma_1}\qa{\gamma_2} \Big(\psi^\frakt_{\equiv}(\gamma_1,\gamma_2) \land \psi^\frakp_{\prec}(\gamma_1,\gamma_2)\big) \timp                                                                                                                                                                                         \\
                             & \bigwedge_{\mathclap{(q_1,a_1),                                                                                                                                                                   (q_2,a_2) \in Q \times \Gamma}} \negg\big((\gamma_1 \hook c^{-1}(q_1,a_1)) \land (\gamma_2 \hook c^{-1}(q_2,a_2)\big)\Big) \\
 \intertext{We also check that every row contains some state. For this, $\qa{\gamma_1}$ fixes some row and $\qe{\gamma_2} \psi^\frakt_\equiv(\gamma_1,\gamma_2)$ searches that particular row for a state:}
 \vartheta_2                 & \dfn \qa{\gamma_1} \qe{\gamma_2} \; \psi^\frakt_\equiv(\gamma_1,\gamma_2) \land \bigovee_{\mathclap{(q,a) \in Q \times \Gamma}}\; (\gamma_2 \hook c^{-1}(q,a))                                                                                                                                                                            \\
 \intertext{Finally, every window must obey the transition relation:}
 \vartheta_3                 & \dfn \qa{\gamma_1} \cdots \qa{\gamma_6}  \; \Big( \psi_\text{window}(\gamma_1,\ldots,\gamma_6) \timp  \bigovee_{(e_1,\ldots,e_6) \in \mathrm{win}} \;\bigwedge_{i=1}^6 (\gamma_i \hook e_i) \Big)
\end{align*}

\begin{claim}[g]
 $T \vDash \psi_\text{legal}(\alpha)$ iff $T_\alpha$ is a legal tableau, \ie, iff $C_{T_\alpha}$ exists and is a run of $M$.
\end{claim}
\begin{cproof}
 Suppose that the formula holds.
 We show that $T_\alpha$ is a legal tableau; the other direction is proven similarly.

 Due to Claim (e), there are tableaus $S_1 \subseteq T_{\gamma_1}$, \ldots, $S_6 \subseteq T_{\gamma_6}$ that are copies of $T_\alpha$ such that $\vartheta_1 \land \vartheta_2 \land \vartheta_2$ holds in $T^{\gamma_1,\ldots,\gamma_6}_{S_1,\ldots,S_6}$.

 Due to Claim (a), the subformula $\vartheta_1$ ensures the following:
 For all $w \in S_1, w' \in S_2$, $\ell(w) = (i,j)$, $\ell(w') = (i',j')$, if $i = i'$ and $j < j'$ hold, then it is not the case that both $c(w) = (q,a)$ and $c(w') = (q',a')$ for any state symbols $q,q' \in Q$.
 Since $C_{S_1} = C_{S_2} = C_{T_\alpha}$,
 this is precisely the case if each row of $C_{T_\alpha}$ contains at most one state symbol.

 Conversely, again by Claim (a), the subformula $\vartheta_2$ states that for every cell $w \in S_1$ there is another cell $w' \in S_2$ in the same row that carries a state symbol: in other words, every row of $C_{T_\alpha}$ contains at least one state symbol.

 Finally, $\vartheta_3$ relies on Claim (f) and states for every choice of singletons $w_1,\ldots,w_6$ in $S_1,\ldots,S_6$, assuming that they are arranged as a window, that there exists a tuple $(e_1,\ldots,e_6)\in \mathrm{win}$ such that $w_i \in S_i$ satisfies $c(w_i) = e_i$.
 As we showed that $C_{T_\alpha}$ contains in each row a configuration of $M$, this implies that  $C_{T_\alpha}$ exists and is a run of $M$.
\end{cproof}

\subsection{From runs to a computation}

To encode the initial configuration on input $x = x_1 \cdots x_n$ in a tableau, we access the first $n$ cells of the first row and assign the respective letter of $x$, as well as the initial state, to the first cell.
Moreover, we assign $\flat$ to all other cells in that row.
For each $\frakq \in \{\frakt,\frakp\}$, we can check whether the location of a point in $T_\alpha$ is minimal in its $\frakq$-component:
\begin{align*}
 \psi^\frakq_\text{min}(\alpha) \dfn \; & \negg \qe{\gamma_0} \psi^\frakq_\prec(\gamma_0,\alpha)                                                                                                                                                                        \\
 \intertext{This enables us to fix the first row of the configuration:}
 \psi_{\text{input}}(\alpha) \dfn \;    & \qeet{\gamma_1}{\alpha}\cdots \qeet{\gamma_{n+1}}{\alpha} \;  \qe{\gamma_1} \cdots \qe{\gamma_n} \,\psi^\frakt_{\text{min}}(\gamma_1) \land \psi^\frakp_{\text{min}}(\gamma_1) \land \big(\gamma_1 \hook c^{-1}(q_0,x_1)\big) \\
                                        & \quad \bigwedge_{i = 2}^n \psi^\frakp_\text{succ}(\gamma_{i-1},\gamma_i)  \; \land \big(\gamma_i \hook c^{-1}(x_i)\big)                                                                                                       \\
                                        & \qquad \land \qa{\gamma_{n+1}} \Big(\big(\psi^\frakt_\equiv(\gamma_n,\gamma_{n+1}) \land \psi^\frakp_\prec(\gamma_n,\gamma_{n+1})\big)  \timp \big(\gamma_{n+1} \hook c^{-1}(\flat)\big)\Big)
\end{align*}

\begin{claim}[h]
 Let $T_\alpha$ be a tableau.
 Then $T \vDash \psi_\text{input}(\alpha)$ if and only if
 \begin{enumerate}
  \item $C_{T_\alpha}(1,1) = (q_0,x_1)$,
  \item $C_{T_\alpha}(1,i) = x_i$ for $2 \leq i \leq n$,
  \item $C_{T_\alpha}(1,i) = \flat$ for $n < i \leq N$.
 \end{enumerate}
\end{claim}
\begin{cproof}
 Suppose that the formula holds.
 After processing the quantifiers $\qeet{\gamma_1}{\alpha}\cdots\qeet{\gamma_{n+1}}{\alpha}$, for all $m \in \{1,\ldots,n+1\}$ the team $T_{\gamma_m}$ is a tableau such that $C_{T_{\gamma_m}} = C_{T_\alpha}$. (Obviously this requires these teams to be pre-tableaus beforehand.)
 For this reason, we can freely replace $C_{T_\alpha}(i,j)$ with $C_{T_{\gamma_m}}(i,j)$ when proving the properties (1)--(3). % chktex 36

 In the second line of the formula, we make sure that $c(w) = (q_0,x_1)$ holds for least one point $w \in C_{T_{\gamma_1}}$ of location $\ell(w) = (1,1)$.
 That $\ell(w) = (1,1)$ holds follows from Claim (a), $\psi^\frakq_\text{min}$, and the assumption that $T_{\gamma_0}$ is a pre-tableau (which it still is after processing $\qeet{\gamma_1}{\alpha}\cdots\qeet{\gamma_{n+1}}{\alpha}$).
 In particular, $C_{T_{\gamma_1}}(1,1) = (q_0,x_1)$.

 The third line works similarly: for $2 \leq i \leq n$, it assigns $x_i$ to $C_{T_{\gamma_i}}(1,i)$ and hence to $C_{T_\alpha}(1,i)$.
 Note that $\psi^\frakp_\text{succ}$ also preserves the position in "$\frakp$-direction", \ie, it is not necessary to repeat it for every cell of the first row.
 Finally, the last two lines state that every other location $(1,j')$ with $j' > n$ contains $\flat$.
 The other direction is again similar.
\end{cproof}

Until now, we ignored the fact that $M$ (polynomially often) alternates.
To simulate this, we alternatingly quantify polynomially many tableaus, each containing a part of the computation of $M$.
Each of these tableaus possesses a \emph{tail configuration},
which is the configuration where $M$ either accepts, rejects, or alternates.
Formally, a number $i \in \{1,\ldots,N\}$ is a \emph{tail index} of $C$ if there exists $j$ such that either
\begin{enumerate}
 \item $C(i,j)$ has an accepting or rejecting state,
 \item or $C(i,j)$ has an existential state and and there are $i' < i$ and $j'$ with a universal state in $C(i',j')$,
 \item or $C(i,j)$ has a universal state and there are $i' < i$ and $j'$ with an existential state in $C(i',j')$.
\end{enumerate}
The least such $i$ is called \emph{first tail index}, and the corresponding configuration is the \emph{first tail configuration}.
The idea is that we can split the computation of $M$ into multiple tableaus if any tableau (except the initial one) contains a run that continues from the previous tableau's first tail configuration.

We formalize the above as follows.
Assume that $T_\alpha$ is a tableau, and that $T_{\beta}$ marks a single row $i$ by being a singleton $\{w\}$ with $\ell(w) = (i,j)$ for some $j$.
Then the formula $\psi_\text{tail}(\alpha,\beta)$ below will be true if and only if the $i$-th row of $C_{T_\alpha}$ is a tail configuration.
With
\begin{align*}
 \xstate{Q'}(\beta) & \dfn \bigovee_{\mathclap{(q,a) \in Q' \times \Gamma}} (\beta \hook c^{-1}(q,a))\text{,}
\end{align*}
we check if a given singleton $T_{\beta} = \{w\}$ encodes an accepting, rejecting, existential, universal, or any state by setting $Q'$ to $Q_\mathrm{acc}$, $Q_\mathrm{rej}$, $Q_{\exists}$, $Q_\forall$ or $Q$, respectively.
We define $\psi_\text{tail}$:
\begin{align*}
 \psi_\text{tail}(\alpha,\beta) \;&\dfn
 \qeet{\gamma_0}{\alpha}\, \qe{\alpha}                        \psi^\frakt_\equiv(\alpha,\beta) \land \xstate{Q}(\alpha)  \land \Big[ \xstate{Q_\mathrm{acc}}(\alpha) \ovee \xstate{Q_\mathrm{rej}}(\alpha) \; \ovee \\
  \qe{\gamma_0}& \Big(\psi^\frakt_\prec(\gamma_0,\alpha) \land
  \big(\xstate{Q_\exists}(\alpha) \land \xstate{Q_\forall}(\gamma_0)) \ovee (\xstate{Q_\forall}(\alpha) \land \xstate{Q_\exists}(\gamma_0) \big) \Big) \Big]                                                            \\
 \psi_\text{first-tail}(\alpha,\beta) \;& \dfn  \psi_\text{tail}(\alpha,\beta) \land \negg \qe{\gamma_1} \Big(\psi^\frakt_\prec(\gamma_1,\beta) \land \psi_\text{tail}(\alpha,\gamma_1)\Big)
\end{align*}

\begin{claim}[i]
 Suppose that $T_\alpha$ is a tableau, $T_\beta =\{ w\}$, and $\ell(w) = (i,j)$.
 Then $T \vDash \psi_\text{tail}(\alpha,\beta)$ if and only if $i$ is a tail index of $C_{T_\alpha}$.
 Moreover, $T \vDash \psi_\text{first-tail}(\alpha,\beta)$ if and only if $i$ is the first tail index of $C_{T_\alpha}$.
\end{claim}
\begin{cproof}
 Since $T_{\gamma_1}$ is a pre-tableau and hence contains all locations in rows $i' < i$, it is easy to see that the proof for $\psi_{\text{first-tail}}$ boils down to that of $\psi_\text{tail}$.
 Consequently, let us consider $\psi_\text{tail}$.

 First, due to $\qeet{\gamma_0}{\alpha}$, we can assume that $T_{\gamma_0}$ is a tableau that is a copy of $T_\alpha$, \ie, $C_{T_\alpha} = C_{T_{\gamma_0}}$.
 Here, it is required for the inner quantification in the definition of a tail index.

 The first line of the formula reduces $T_\alpha$ to a singleton that is (due to $\psi^\frakt_\equiv$) in row $i$.
 Furthermore, it carries a state $q$ of $M$ due to $\xstate{Q}(\alpha)$.
 The further examination of this state will determine if $i$ is a tail index.
 Now, $q$ is exactly one of accepting, rejecting, existential, or universal.
 If $q \in Q_\mathrm{acc} \cup Q_\mathrm{rej}$, then $i$ is a tail index by definition.

 Otherwise we quantify over the states $q'$ of all (copies of) earlier rows in $T_\alpha$, using $\qe{\gamma_0} \psi^\frakt_\prec(\gamma_0,\alpha)$, and search for a universal state if $q$ is existential and vice versa, which as well, if it exists, proves by definition that $i$ is a tail index.
\end{cproof}

\medskip

Formally, given a run $C$ of $M$ that has a tail configuration, $C$ \emph{accepts} if the state $q$ in its first tail configuration is in $Q_\mathrm{acc}$, $C$ \emph{rejects} if that $q$ is in $Q_\mathrm{rej}$, and $C$ \emph{alternates} otherwise.
That a run of the form $C_{T_\alpha}$ accepts or rejects is expressed by
\begin{align*}
 \psi_\text{acc}(\alpha) & \dfn \qeet{\gamma_2}{\alpha} \; \qe{\gamma_2} \; \xstate{Q_\mathrm{acc}}(\gamma_2) \land \psi_\text{first-tail}(\alpha,\gamma_2)\text{,} \\
 \psi_\text{rej}(\alpha) & \dfn \qeet{\gamma_2}{\alpha} \; \qe{\gamma_2} \; \xstate{Q_\mathrm{rej}}(\gamma_2) \land \psi_\text{first-tail}(\alpha,\gamma_2)\text{.}
\end{align*}

In this formula, first the tableau $T_\alpha$ is copied to $T_{\gamma_2}$ to extract with $\qe{\gamma_2}$ the world carrying an accepting/rejecting state, while $\psi_\text{first-tail}(\alpha,\gamma_2)$ ensures that no alternation or rejecting/accepting state occurs at some earlier point in $C_{T_\alpha}$.

If the first tail configuration of the run contains an alternation, and if the run was existentially quantified, then it should be continued in a universally quantified tableau, and vice versa.
The following formula expresses, given two tableaus $T_\alpha,T_\beta$, that $C_{T_\beta}$ is a \emph{continuation} of $C_{T_\alpha}$, \ie, that the first configuration of $C_{T_\beta}$ equals the first tail configuration of $C_{T_\alpha}$.
In other words, if $i$ is the first tail index of $C_{T_\alpha}$, then $C_{T_\alpha}(i,j) = C_{T_\beta}(1,j)$ for all $j \in \{1, \ldots,N \}$.
\begin{align*}
 \psi_\text{cont}(\alpha,\beta) \dfn \; & \qe{\gamma_2} \, \psi_\text{first-tail}(\alpha,\gamma_2) \land \qa{\alpha}\qa{\beta}                                                                                                                           \\
                                        & \quad \Big[ \Big(\psi^\frakt_\text{min}(\beta) \land \psi^\frakt_\equiv(\alpha,\gamma_2) \land \psi^\frakp_\equiv(\alpha,\beta)\Big)  \timp \Big( \bigovee_{e \in \Xi} (\alpha \lor \beta) \hook e \Big) \Big]
\end{align*}
The above formula first obtains the first tail index $i$ of $C_{T_\alpha}$ and stores it in a singleton $y \in T_{\gamma_2}$.
Then for all worlds $w \in T_\alpha$ and $v \in T_\beta$, where $v$ is $\frakt$-minimal (\ie, in the first row) and $w$ is in the same row as $y$, and which additionally agree on their $\frakp$-component, the third line states that $w$ and $v$ agree on $\Xi$.
Altogether, the $i$-th row of $C_{T_\alpha}$ and the first row of $C_{T_\beta}$ then have to coincide.

\medskip

$M$ performs at most $r(n) - 1$ alternations for some polynomial $r$.
Then we require $r = r(n)$ tableaus, which we call $\alpha_1,\ldots,\alpha_r$.
In the following, the formula $\psi_{\text{run},i}$ describes the behaviour of the $i$-th run, \ie, the part of the computation after $i-1$ alternations.
\Wloss $r$ is even and $q_0 \in Q_\exists$.
We may then define the final run by
\begin{align*}
 \psi_{\text{run},r} & \dfn \qaa{\alpha_r} \Big[\Big( \psi_\text{legal}(\alpha_r)  \land \psi_\text{cont}(\alpha_{r-1},\alpha_r)\Big)  \timp \Big(\negg\psi_\text{rej}(\alpha_r) \land \psi_\text{acc}(\alpha_r)\Big)\Big]\text{.}
\end{align*}
For $1 < i < r$ and even $i$, let
\begin{align*}
 \psi_{\text{run},i} \dfn \; & \qaa{\alpha_i} \;\Big[ \Big(\psi_\text{legal}(\alpha_i) \land  \psi_\text{cont}(\alpha_{i-1},\alpha_i)\Big)\timp \Big(\negg\psi_\text{rej}(\alpha_i) \land \big( \psi_\text{acc}(\alpha_i) \ovee \psi_{\text{run},{i+1}} \big)\Big)\Big]
 \intertext{and for $1 < i < r$ and odd $i$}
 \psi_{\text{run},i} \dfn\;  & \qee{\alpha_i} \Big[ \psi_\text{legal}(\alpha_i)\land \psi_\text{cont}(\alpha_{i-1},\alpha_i) \land \negg\psi_\text{rej}(\alpha_i) \land \Big(\psi_\text{acc}(\alpha_i) \ovee \psi_{\text{run},{i+1}} \Big)\Big]\text{.}
\end{align*}
Analogously, the initial run is described by
\begin{align*}
 \psi_{\text{run},1} \dfn & \; \qee{\alpha_1} \Big( \psi_\text{legal}(\alpha_1)\land \psi_\text{input}(\alpha_1) \land \negg\psi_\text{rej}(\alpha_1) \land  \Big(\psi_\text{acc}(\alpha_1) \ovee \psi_{\text{run},{2}} \Big)\Big)
\end{align*}

\medskip

We are now in the position to state the full reduction.
Let us gather all relevant scopes in the set $\Psi \subseteq \PS$:
\begin{align*}
 \Psi \dfn \; & \{ \fraka_i \mid 0 \leq i \leq k \} \cup \{ \fraka'_k\}\; \cup\;
 \{ \gamma_i \mid 0 \leq i \leq n+1 \} \; \cup \; \{ \alpha_i \mid 1 \leq i \leq r \}
\end{align*}
The scopes that accommodate pre-tableaus are
\begin{align*}
 \Psi' \dfn \; & \{ \gamma_i \mid 0 \leq i \leq n+1 \}  \cup \{ \alpha_i \mid 1 \leq i \leq r \}\text{.}
\end{align*}
\Wloss $n \geq 5$, as $\gamma_1,\ldots,\gamma_6$ are always required in the construction.
The reduction now maps $x$ to
\begin{align*}
 \varphi_x \dfn \; & \canon' \land \scopes_k(\Psi)\land \bigwedge_{p\in \Psi'} \psi_\text{pre-tableau}(p) \land \psi_\text{run,1}\text{.}
\end{align*}

It is easy to see that this formula is an $\MTL_k$-formula that is logspace-constructible from $x$ and $k$, where $k$ itself is either constant or a polynomial in $\size{x}$ and hence logspace-computable.
By Lemma~\ref{lem:model-restrict-height}, $\varphi_x$ is satisfiable if and only if $\varphi_x \land \Box^{k+1}\bot$ is satisfiable.
For this reason, we conclude the reduction with the following proof.

\begin{proof}[Proof of Lemma~\ref{lem:reduction}]
 It remains to argue that $\varphi_x \land \Box^{k+1}\bot$ is satisfiable if and only if $M$ accepts $x$.
 For the sake of simplicity, assume $r = 2$.
 The cases $r > 2$ are proven analogously.

 "$\Rightarrow$":
 Suppose $(\calK,T) \vDash \varphi_x \land \Box^{k+1}\bot$.
 Similarly as in Theorem~\ref{thm:bounded-canon}, the $p \in \Psi$ are disjoint scopes due to $\scopes_k(\Psi)$.
 Moreover, by $\canon'$ and Claim (b), $(\calK,T)$ is then a $k$-staircase in which $T_{\fraka_k}$ and $T_{\fraka'_k}$ both are $k$-canonical teams.
 Due to Claim (d) and the large conjunction in $\varphi_x$, $T_{\alpha_1},T_{\alpha_2},T_{\gamma_1},\ldots,T_{\gamma_{n+1}}$ are then pre-tableaus.

 As the formula $\psi_\text{run,1}$ holds, by Claim (g) and (h), $T_{\alpha_1}$ has a subteam $S_1$ that is a legal tableau and starts with $M$'s initial configuration on $x$.
 In particular, $C_{S_1}$ exists.
 Moreover, either $\psi_\text{acc}$ holds (\ie, $C_{S_1}$ and hence $M$ is accepting) or $\psi_\text{run,2}$ holds (\ie, if $C_{S_1}$ alternates).
 Consider the latter case.
 Then for all legal tableaus $S_2 \subseteq T_{\alpha_2}$ such that $C_{S_2}$ is a continuation of $C_{S_1}$ it holds that $C_{S_2}$ is accepting.
 However, as $T_{\alpha_2}$ is a pre-tableau, every run is of the form $C_{S_2}$ for some $S_2 \subseteq T_{\alpha_2}$.
 Consequently, $M$ accepts $x$.

 "$\Leftarrow$":
 Suppose $M$ accepts $x$.
 First of all, due to Claim (b), the formula $\canon' \land \scopes_k(\{\fraka_0,\ldots,\fraka_k,\fraka'_k\})\land \Box^{k+1}\bot$ has a model $(\calK,T)$.
 Moreover, we can freely add a pre-tableau $T_p$ for each $p \in \Psi$ to satisfy the large conjunction in $\varphi_x$.
 By labeling the propositions in $\Psi$ correctly (as disjoint scopes), we ensure that $\scopes_k(\Psi)$ holds as well.

 It remains to demonstrate $T \vDash \psi_\text{run,1}$.
 As $M$ accepts $x$, there exists a run $C_1$ starting from $M$'s initial configuration such that either $C_1$ accepts, or, for all runs $C_2$ continuing $C_1$, $C_2$ accepts.

 Since $T_{\alpha_1}$ is a pre-tableau, it also contains a subteam $S_1$ such that $S_1$ is a legal tableau and $C_{S_1} = C_1$.
 We choose $S_1$ as witness for $\qee{\alpha_1}$.
 If $C_1$ itself accepts, then $\psi_\text{acc}(\alpha_1)$ and hence $\psi_\text{run,1}$ is satisfied.
 Otherwise we consider $\psi_\text{run,2}$.
 Suppose that $S_2 \subseteq T_{\alpha_2}$ is picked as a subteam by $\qaa{\alpha_2}$.
 If it forms a legal tableau and $C_{S_2}$ is a continuation of $C_1$, then $C_2$ must be accepting since $M$ accepts $x$ by assumption.
 But this implies that $\psi_\text{acc}(\alpha_2)$ is true for any such $S_2$.
 Consequently, $\psi_\text{run,2}$ and hence $\psi_\text{run,1}$ is true.
\end{proof}

\section{Hardness under strict semantics and on restricted frame classes}%
\label{sec:variants}

\subsection{Lax and strict semantics}

In this section, we further generalize the hardness result of the previous section.

Team-semantical connectives can be evaluated either in so-called \emph{standard} or \emph{lax semantics}, or alternatively in \emph{strict semantics}.
In Section~\ref{sec:mtl-prelim}, we defined $\MTL$ with lax semantics.
In strict semantics, the connectives $\lor$ and $\Diamond$ are replaced by their counterparts $\lor_s$ and $\Diamond_s$:
\begin{alignat*}{3}
  & (\calK,T) \vDash \psi \lor_s \theta  &  & \Leftrightarrow\;\exists S, U \subseteq T \text{ such that }T = S \cup U\text{, }S \cap U = \emptyset\text{, }(\calK,S) \vDash \psi\text{, and }(\calK,U)\vDash \theta\text{,} \\
  & (\calK,T)\vDash \Diamond_s \psi      &  & \Leftrightarrow\;(\calK, S)\vDash \psi \text{ for some strict successor team }S\text{ of }T\text{,}
\end{alignat*}
where a \emph{strict successor team} of $T$ is a successor team $S \subseteq RT$ for which there exists a surjective $f \colon T \to S$ satisfying $f(w) \in Rw$ for all $w \in T$.
Intuitively, in the lax disjunction the teams of the splitting may overlap, while in the strict disjunction they are disjoint.
Likewise, a lax successor team may contain multiple successor of any $w \in T$, while in a strict successor team we pick exactly one successor for each $w \in T$.

\smallskip

An $\MTL$-formula $\varphi$ is \emph{downward closed} if $(\calK,T) \vDash \varphi$ implies $(\calK,S) \vDash \varphi$ for all $S \subseteq T$.
For example, every $\ML$-formula is downward closed, as is the constancy atom $\dep{\alpha} = \alpha \ovee \neg \alpha$ or generally any monotone Boolean combination of $\ML$-formulas.
On such formulas, strict and lax semantics are equivalent:

\begin{prop}\label{prop:strict-lax}
  Let $\varphi,\psi \in \MTL$ such that $\varphi$ is downward closed.
  Then $\varphi \lor \psi \equiv \varphi \lor_s \psi$ and $\Diamond\varphi \equiv \Diamond_s \varphi$.
\end{prop}
\begin{proof}
  Clearly $\varphi \lor_s \psi$ entails $\varphi \lor \psi$ and $\Diamond_s \varphi$ entails $\Diamond \varphi$.
  If conversely $T \vDash \varphi \lor \psi$ via subteams $S,U \subseteq T$ such that $S \cup U = T$, $S \vDash \varphi$ and $U \vDash \psi$, then we instead split $T$ into the subteams $U$ and $T \setminus U$.
  Since $T \setminus U \subseteq S$ and $\varphi$ is downward closed, this proves $T \vDash \varphi \lor_s \psi$.

  Likewise, suppose $T \vDash \Diamond \varphi$ via some successor team $S$ of $T$.
  Assuming the axiom of choice, there is some function $f \colon T \to S$ such that $f(w) \in Rw$ for each $w \in T$.
  The team $\{f(w) \mid w \in T\}\subseteq S$ is now a strict successor team of $T$ and satisfies $\varphi$ due to downward closure.
\end{proof}

Due to Proposition~\ref{prop:strict-lax}, the distinction between strict and lax semantics was traditionally unnecessary for many team logics such as the original \emph{dependence logic}~\cite{vaananen_dependence_2007,vaananen_modal_2008}, as it has only downward closed formulas.
The distinction between strict and lax semantics was first made in the context of first-order team logic by Galliani~\cite{Galliani12}.
It has some interesting consequences, for instance first-order inclusion logic in strict semantics is as expressive as existential second-order logic~\cite{GallianiHK13} (see also Hannula and Kontinen~\cite{HannulaK15}).

With modal team logic, strict semantics was studied, \eg, by Hella et al.~\cite{minc,HellaKMV15,HellaKMV17}.
In the works that explicitly study strict semantics, the underlying (first-order or modal) team logic was enriched by not downward closed constructs such as the \emph{inclusion atom} $\subseteq$ or \emph{exclusion atom} $\mid$, or the \emph{independence atom} $\perp$.

In this article, where we consider team-wide negation $\negg$ as part of the logic, the distinction between strict and lax semantics becomes apparent already for simple formulas such as $\E \top \lor \E \top \not \equiv \E \top \lor_s \E \top$, where the former defines non-emptiness, but the latter means that the team contains at least two points.

We prove that our hardness results also hold in strict semantics.
Let the logics $\MTL(\lor_s,\Box)$ and $\MTL_k(\lor_s,\Box)$ be defined like $\MTL$ and $\MTL_k$, but with $\lor_s$ instead of $\lor$ and without $\Diamond$ and $\Diamond_s$ (\ie, only using the modality $\Box$).

\begin{thm}\label{thm:complexity-strict}$\SAT(\calL)$ and $\VAL(\calL)$ are hard for $\TOWERPOLY$ if $\calL = \MTL(\lor_s,\Box)$, and hard for
$\ATIMEALT{\exp_{k+1}}$ if $\calL = \MTL_k(\lor_s,\Box)$ and $k \geq 0$.
  \end{thm}
\begin{proof}
  An analysis of the proof of Lemma~\ref{lem:reduction} yields that the $\MTL$-formula $\varphi_x$ produced in the reduction can be easily adapted to strict semantics.
First, observe that $\Diamond$ occurs only in the subformula $\max_i$, which is by Proposition~\ref{prop:strict-lax} equivalent to
\[
\top \lor_s \Bigl(\neg\Box^i\bot \land \negg\mathop{{\bigvee}_{\!\!s}}\limits_{\!p\in\Phi} (\neg\Box^i p \ovee \neg \Box^i \neg p)\Bigr)\text{,}
\]
since $\Diamond\alpha \equiv \neg \Box \neg \alpha$ and $\neg \Box^i p \ovee \neg \Box^i \neg p$ is a downward closed formula.
A quick check reveals that all other instances of $\lor$ in $\varphi_x$ are subject to Proposition~\ref{prop:strict-lax} as well, except of the occurrence in the second line of $\zeta^*_k$.
Here, the critical part of the correctness proof is the choice of the subteam $U'$ in Claim (c) of Lemma~\ref{lem:order1}.
In strict semantics, the only possibility becomes $U' = U = O \setminus S$, for which the proof works identically.
Finally, for the case $k= 0$, a similar check of the proof for $\PTL$~\cite[Theorem 4.9]{ptl2017} reveals that there also every $\lor$ can be replaced by $\lor_s$ due to Proposition~\ref{prop:strict-lax}.
\end{proof}

Note that the corresponding upper bound via the construction of a canonical model (viz.\ Theorem~\ref{thm:canonical}) does not apply to strict semantics.
The reason for this is the failure of Proposition~\ref{prop:mtl-bisim-types}:
In strict semantics, $\MTL_k$-formulas are not invariant under $k$-team-bisimulation in general.

As an example, consider the formula $\varphi \dfn \E\top \lor_s \E\top$.
It states that the team contains at least two points.
However, for every finite $\Phi \subseteq \PS$ and $k \geq 0$ it is easy to find a team $T$ of two points and a singleton $S$ that is $(\Phi,k)$-bisimilar to it, while $T \vDash \varphi$ and $S \nvDash \varphi$.

A possible approach could be to define a bisimulation relation that respects the multiplicity of types in a team, and to define a corresponding canonical model, but this is beyond the scope of this paper.

\subsection{Restricted frame classes}

A natural restriction in the context of modal logic is to focus on a specific subclass of Kripke frames, which is useful for instance for modeling belief or temporal systems.
(For an introduction to frame classes, consider, \eg, Fitting~\cite{handbook2}.)
Let $F = (W,R)$ denote a frame.
Prominent frame classes include
\begin{description}
  \item[$\mathsf{K}$] all frames,
  \item[$\mathsf{D}$] serial frames ($w \in W \Rightarrow Rw \neq \emptyset$),
  \item[$\mathsf{T}$] reflexive frames ($w \in W \Rightarrow w \in Rw$),
  \item[$\mathsf{K4}$] transitive frames ($u \in Rv, v \in Rw, w \in W \Rightarrow u \in Rw$),
  \item[$\mathsf{D4}$] serial and transitive frames,
  \item[$\mathsf{S4}$] reflexive and transitive frames.
\end{description}

In this section, we consider these classes from a complexity theoretic perspective, and show that the lower bounds of $\MTL$ hold when restricted to these classes.
Given a frame class $\calF$ and a fragment $\calL$ of $\MTL$, let $\SAT(\calL,\calF)$ denote the set of all $\calL$-formulas that are satisfied in a model $(W,R,V,T)$ where $(W,R)$ is a frame in $\calF$.
Define $\VAL(\calL,\calF)$ analogously.

We prove the team-semantical analog of Ladner's theorem, which states that classical modal satisfiability and validity are $\PSPACE$-hard problem for any frame class between $\mathsf{S4}$ and $\mathsf{K}$~\cite[Theorem 3.1]{Ladner77}.
Note that this includes all the frame classes stated above.

\begin{thm}\label{thm:complexity-s4}
  Let $\calF$ be a frame class such that $\mathsf{S4} \subseteq \calF \subseteq \mathsf{K}$.
Then $\SAT(\MTL,\calF)$ and $\VAL(\MTL,\calF)$ are hard for $\TOWERPOLY$,
and $\SAT(\MTL_k,\calF)$ and $\VAL(\MTL_k,\calF)$ are hard for $\ATIMEALT{\exp_{k+1}}$, for $k \geq 0$.
\end{thm}
\begin{proof}
We give the proof for $\SAT(\MTL_k) \leqlogm \SAT(\MTL_k, \calF)$.
Let $\varphi \in \MTL_k$.
The idea is to introduce new propositions $\ell_0,\ldots,\ell_k \notin \Prop(\varphi)$ that mark the layers of different height in a structure, and to modify the formula such that all edges except between consecutive layers $i$ and $i+1$ are ignored.
(Here, we make the assumption that $\calK$ is a acyclic, which relies on Corollary~\ref{cor:height-model} and hence indirectly on Proposition~\ref{prop:mtl-bisim-types}).

Given a $\Phi\cup\{\ell_0,\ldots,\ell_k\}$-structure $\calK = (W, R, V)$, let $\calK^\circ \dfn (W, R^\circ, V)$ be the structure where only such edges are retained, \ie,
\[
R^\circ = R \cap \bigcup_{i = 0}^{k-1} (V(\ell_i) \times V(\ell_{i+1}))\text{.}
\]
On the side of formulas, the reduction is $\varphi \mapsto \ell_0 \land \varphi^0$, where $\varphi^i$ is inductively as follows.
The non-modal connectives are ignored, \ie, $p^i \dfn p$ for $p \in \Phi$, ${(\psi\land\theta)}^i \dfn \psi^i\land\theta^i$, ${(\negg\psi)}^i \dfn \negg\psi^i$, ${(\psi\lor\theta)}^i \dfn \psi^i \lor \theta^i$.
For the modalities, let ${(\Diamond\psi)}^i \dfn \Diamond(\ell_{i+1} \land \psi^{i+1})$ and $(\Box\psi^i) \dfn \Box(\ell_{i+1} \hook \psi^{i+1})$.
Intuitively, $\varphi^i$ is meant to be evaluated in layer $i$, and we make sure that successor teams always are contained in the next layer $i+1$.

For the correctness of the reduction, we will first show the following claim.

\begin{claim}
  For all $i \in \{ 0,\ldots,k \}$ and $T \subseteq V(\ell_i)$, it holds that $(\calK,T) \vDash \varphi^i$ iff $(\calK^\circ,T) \vDash \varphi$.
\end{claim}
\begin{cproof}
  This is proved by a straightforward induction on the formula size:
\begin{itemize}
  \item Atomic propositions are clear.
  The Boolean connectives and splitting follow straightforwardly from the induction hypothesis (as subteams of $T$ are again in $V(\ell_i)$).
\item Let $\varphi = \Diamond\psi$.
Suppose $(\calK,T) \vDash \varphi^i$, \ie, $(\calK,S) \vDash \ell_{i+1} \land \psi^{i+1}$ for some $R$-successor team $S$ of $T$.
Then by induction hypothesis $(\calK^\circ,S) \vDash \psi$, as $S \subseteq V(\ell_{i+1})$.
$S$ is an $R^\circ$-successor team of $T$ as well, since $(w,v) \in R \Leftrightarrow (w,v) \in R^\circ$ for every $(w,v) \in V(\ell_i) \times V(\ell_{i+1})$.
This proves $(\calK^\circ,T) \vDash \varphi$.

Conversely, if $(\calK^\circ,T) \vDash \varphi$, then $(\calK^\circ,S) \vDash \psi$ for some $R^\circ$-successor team $S$ of $T$.
However, any $R^\circ$-successor team of $T$ is a subset of $V(\ell_{i+1})$.
As a consequence, $(\calK,S) \vDash \ell_{i+1}$.
Moreover, by induction hypothesis, $(\calK,S)\vDash \psi^{i+1}$.
This yields $(\calK,T) \vDash \varphi^i$, since $S$ is trivially also a $R$-successor team of $T$.
\item Let $\varphi = \Box \psi$.
Then $(\calK,T) \vDash \varphi^i$ iff $(\calK,RT) \vDash (\ell_{i+1} \hook \psi^{i+1})$ iff $(\calK,RT \cap V(\ell_{i+1})) \vDash \psi^{i+1}$ iff $(\calK^\circ,RT \cap V(\ell_{i+1}) \vDash \psi$ by induction hypothesis.
It remains to show that $R^\circ T = RT \cap V(\ell_{i+1})$.
Clearly, $R^\circ T \subseteq RT$ and $R^\circ T \subseteq V(\ell_{i+1})$, since $R^\circ \subseteq R$, $R^\circ V(\ell_i) \subseteq V(\ell_{i+1})$, and $T \subseteq V(\ell_i)$.
Conversely, if $w \in RT \cap V(\ell_{i+1})$, then $(v,w) \in R$ for some $v \in T$.
As $(v,w) \in V(\ell_i) \times V(\ell_{i+1})$, then $(v,w) \in R^\circ$, hence $w \in R^\circ T$.\altqed%
\end{itemize}\let\altqed\relax
\end{cproof}

\noindent
Now, due to the above claim, if $\ell_0 \land \varphi^0$ is satisfiable, then clearly $\varphi$ is as well.
It remains to show that $\ell_0 \land \varphi^0$ has a reflexive and transitive model if $\varphi$ is satisfiable.
Suppose that the latter is satisfied in a $\Phi$-structure $(\calK,T)$.
By Corollary~\ref{cor:height-model}, we may assume that $(\calK,T)$ is a forest of height $k$ with the set of roots being $T$.
Then we label the new propositions $\ell_i$ such that $V(\ell_i) = R^{i}T$, \ie, $V(\ell_0) = T$, $V(\ell_1) = RT$ and so on.
As $\calK$ is a forest, note that the sets $T, RT, R^2T, \ldots$ are pairwise disjoint.
In other words, every world in $\calK$ has a unique distance $0 \leq i \leq k$ from $T$ and hence exactly one $\ell_i$ labeled.
This is required for the next part of the proof.

Let now $R^*$ be the reflexive transitive closure of $R$.
It remains to show ${(R^*)}^\circ = R$, since then we can again apply the previously proved claim and are done.
It is easy to see that $R \subseteq {(R^*)}^\circ$, since for every $(w,v) \in R$ there is some $i$ such that $w \in R^{i}T = V(\ell_i)$, consequently $(w,v) \in R^{i}T \times R^{i+1}T = V(\ell_i) \times V(\ell_{i+1})$.
For the other direction, suppose $(w,v) \in {(R^*)}^\circ$.
By definition of ${(R^*)}^\circ$, there is $i$ such that $w \in V(\ell_i)$, $v \in V(\ell_{i+1})$, and $v$ is reachable from $w$ by some $R$-path $(u_0,\ldots,u_n)$ where $w = u_0$ and $v = u_n$.
But since $u_0 \in R^{i}T$, for all $m$ it holds $u_m \in R^{i+m}T = V(\ell_{i+m})$.
As $V(\ell_{i+n}) \cap V(\ell_{i+1}) = \emptyset$ for $n \neq 1$, we conclude $n = 1$, so $(w,v) \in R$.
\end{proof}
 \section{Conclusion}

Theorem~\ref{thm:mtl-main} settles the complexity of $\MTL$ and proves that its satisfiability and validity problems are complete for the non-elementary complexity class $\TOWERPOLY$.
Moreover, the fragments $\MTL_k$ are proved complete for $\ATIMEALT{\exp_{k+1}}$, the levels of the elementary hierarchy with polynomially many alternations.

In our approach, we developed a notion of ($k$-)canonical models for modal logics with team semantics. % chktex 36
We showed that such models exist for $\MTL$ and $\MTL_k$, and that logspace-computable $\MTL_k$-formulas exist that are satisfiable, but only have $k$-canonical models.

Our lower bounds carry over to two-variable first-order team logic $\FO^2(\negg)$ and its fragment $\FO^2_k(\negg)$ of bounded quantifier rank $k$ as well~\cite{fo2-mfcs}.
While the former is $\TOWERPOLY$-complete, the latter is $\ATIMEALT{\exp_{k+1}}$-hard.
However, no matching upper bound for the satisfiability problem of $\FO^2_k(\negg)$ exists.

In the final section, we considered variants of the satisfiability problem for $\MTL$.
We showed that it is as hard as the original problem when $\MTL$ is interpreted in strict semantics, and in fact for $\Diamond$-free formulas with $\lor$ being interpreted either lax or strict.
Also, any restriction of the satisfiability problem to a frame class that includes at least the reflexive-transitive frames is as hard as the original problem.

\medskip

In future research, it could be useful to further generalize the concept of canonical models to other logics with team semantics.
Do logics such as $\FO^2_k(\negg)$ permit a canonical model in the spirit of $k$-canonical models for $\MTL_k$, and does this yield a tight upper bound on the complexity of their satisfiability problem?
How do $\MTL_k$ and $\FO^2_k(\negg)$ differ in terms of succinctness?

\smallskip

Other obvious open questions are the upper bounds for Theorem~\ref{thm:complexity-strict} and~\ref{thm:complexity-s4}, and also the combination of the above aspects, \eg, does the lower bound still hold in strict semantics on reflexive-transitive frames?
To solve these issues, the model theory of modal team logic has to be refined.
For example, what is the analog of Proposition~\ref{prop:mtl-bisim-types} for strict semantics?

\subsection*{Acknowledgements}
The author wishes to thank Heribert Vollmer, Irena Schindler and Arne Meier, as well as the anonymous referee, for numerous helpful comments and suggestions.

\bibliographystyle{plain}
\bibliography{main}

\clearpage

\appendix

\section{Proof details}

In the appendix, we include several propositions that have straightforward but lengthy proofs.

\subsection*{Proofs of Section~\ref{sec:canon}}

\begin{repprop}{prop:det-types}
  Let $\Phi \subseteq \PS$ be finite and $k \geq 0$.
  \begin{enumerate}
   \item $\type{w}^\Phi_k \cap \Phi = V^{-1}(w) \cap \Phi$ and $\type{Rw}^\Phi_k = \calR\type{w}^\Phi_{k+1}$, for all pointed structures $(W,R,V,w)$.\vspace{1mm}
   \item The mapping $h \colon \tau \mapsto \tau \cap \Phi$ is a bijection from $\Delta^\Phi_0$ to $\pow{\Phi}$.
   \item The mapping $h \colon \tau \mapsto (\tau \cap \Phi, \calR\tau)$ is a bijection from $\Delta^\Phi_{k+1}$ to $\pow{\Phi} \times \pow{\Delta^\Phi_k}$.
  \end{enumerate}
\end{repprop}
\begin{proof}[\unskip\nopunct]
 \begin{itemize}
  \item \emph{Proof of (1).}
        Assume $(W,R,V,v), \Phi \subseteq \PS$ and $k \geq 0$ as above.
        For all $p \in \Phi$, clearly $p \in \type{w}^\Phi_k$ iff $w \vDash p $ iff $p \in V^{-1}(w)$.
        Next, we show that $\type{Rw}^\Phi_k = \calR\type{w}^\Phi_{k+1}$.
        Let $\tau = \type{w}^\Phi_{k+1}$, and recall that $\calR\tau = \{ \tau' \in \Delta^\Phi_k \mid \{\alpha \mid \Box \alpha \in \tau\}\subseteq \tau' \}$.
        To prove $\type{Rw}^\Phi_k \subseteq \calR\tau$,
        let $\tau' \in \type{Rw}^\Phi_k$ be arbitrary.
        Then $\type{v}^\Phi_k = \tau'$ for some $v \in Rw$.
        Now, for all $\alpha \in \ML^\Phi_k$, $\Box\alpha \in \tau$ implies $w \vDash \Box \alpha$.
        In particular, $v \vDash \alpha$, \ie, $\alpha \in \tau'$.
        Hence, $\{\alpha \mid \Box \alpha \in \tau\}\subseteq \tau'$, which implies $\tau' \in \calR\tau$.

        For the converse direction, $\calR\tau\subseteq \type{Rw}^\Phi_k$, let $\tau' \in \calR\tau$ be arbitrary.
        By definition, $\{\alpha \mid \Box \alpha \in \tau\}\subseteq \tau'$.
        Since $\tau'$ is a $k$-type, it has a model $(\calK',v')$, and due to Proposition~\ref{prop:types}, $\type{\calK',v'}^\Phi_k = \tau'$.
        By Proposition~\ref{prop:hintikka}, there is a formula $\zeta \in \ML^\Phi_k$ such that $(\calK'',v'')\vDash \zeta$ if and only if $(\calK',v') \bis^\Phi_k (\calK'',v'')$.
        As $\tau$ is a $(k+1)$-type, either $\Diamond \zeta \in \tau$ or $\neg \Diamond \zeta \in \tau$.

        First, suppose $\neg \Diamond \zeta \in \tau$.
        Then $\Box \neg \zeta \in \tau$, hence $\neg \zeta \in \tau'$ by definition of $\tau'$.
        But as $(\calK',v')\vDash \tau'$, then both $(\calK',v') \nvDash \zeta$ and $(\calK',v') \vDash \zeta$, as $(\calK',v') \bis^\Phi_k (\calK',v')$.
        Contradiction, therefore $\Diamond \zeta \in \tau$.
        Consequently, $w$ has an $R$-successor $v$ such that $v \vDash \zeta$, \ie, $\tau' = \type{v}^\Phi_k \in \type{Rw}^\Phi_k$.

  \item \emph{Proof that $h$ in (2) and (3) is injective.}
        Let $\tau,\tau' \in \Delta^\Phi_k$ be arbitrary.
        Let $(\calK,w) = (W,R,V,w)$ be of type $\tau$, and $(\calK',w') = (W',R',V',w')$ of type $\tau'$.
        We first consider (2) and demonstrate that $h\colon \tau \mapsto \tau \cap \Phi$ injective.
        This follows from (1), as $\tau \cap \Phi = \tau' \cap \Phi$ implies $V^{-1}(w) = \tau \cap \Phi = \tau' \cap \Phi = V'^{-1}(w')$, \ie,  $(\calK,w) \bis^\Phi_0 (\calK',w')$.
        By Proposition~\ref{prop:types}, then $\tau = \tau'$.

        \smallskip

        For (3), let $k > 0$, and additionally suppose $\calR\tau = \calR\tau'$.
        Again by (1), we have $\type{\calK,Rw}^\Phi_{k-1} = \calR\tau = \calR\tau' = \type{\calK',R'w'}^\Phi_{k-1}$.
        By Proposition~\ref{prop:types}, $(\calK,Rw) \bis^\Phi_{k-1} (\calK',R'w')$ follows.
        Since $(\calK,w) \bis^\Phi_0 (\calK',w')$ holds as before, $(\calK,w) \bis^\Phi_k (\calK',w')$ by Proposition~\ref{prop:ml-bisim-types}.
        By Proposition~\ref{prop:types}, $\tau = \type{\calK,w}^\Phi_k = \type{\calK',w'}^\Phi_k = \tau'$.

        \medskip

  \item \emph{Proof that $h$ in (2) and (3) is surjective.}
        First, consider (2).
        We have to show that, for all $\Phi' \subseteq \Phi$, there exists a type $\tau \in \Delta^\Phi_0$ such that $\tau \cap \Phi = \Phi'$.
        Likewise, for (3) we have to show that for all $k \geq 0$, $\Phi' \subseteq \Phi$ and $\Delta' \subseteq \Delta^\Phi_k$, there exists a type $\tau \in \Delta^\Phi_{k+1}$ such that $\tau \cap \Phi = \Phi'$ and $\calR \tau = \Delta'$.
        We show the second statement, as the first one is shown analogously.
        The following model $(\calK,w)$ witnesses that there exists $\tau \in \Delta_{k+1}$ such that $\tau \cap \Phi = \Phi'$ and $\calR \tau = \Delta'$.
        First, recall that each $\tau' \in \Delta'$ has a model $(\calN_{\tau'},v_{\tau'})$ such that, by Proposition~\ref{prop:types}, $\type{\calN_{\tau'},v_{\tau'}}^\Phi_k = \tau'$.
        Define $\calK$ as the disjoint union of all $\calN_\tau$ and of a distinct point $w$, and let $V^{-1}(w) = \Phi'$.
        By (1), then $\type{w}^\Phi_{k+1} \cap \Phi = \Phi'$.
        Moreover, let $Rw = \{ v_{\tau'} \mid \tau' \in \Delta' \}$.
        Again due to (1), $\calR\type{w}^\Phi_{k+1} =  \type{Rw}^\Phi_k$.
        By definition, $\type{Rw}^\Phi_k = \type{\{ v_{\tau'} \mid \tau' \in \Delta' \}}^\Phi_k = \{ \type{v_{\tau'}}^\Phi_k \mid \tau' \in \Delta' \} = \Delta'$.\qedhere
 \end{itemize}
 \end{proof}

\begin{replem}{lem:poly-in-tower}
 For every polynomial $p$ there is a polynomial $q$ such that
 \[
    p(\exp^*_k(n)) \leq \exp_k( q((k+1)\cdot n))
 \]
 for all $k \geq 0$ and $n \geq 1$.
\end{replem}

We require the following inequalities.

\begin{lem}\label{lem:into-exp}
 Let $n,k,c \geq 0$.
 Then $c + \exp_k(n) \leq \exp_k(c+n)$.
 If also $n \geq 1$, then $c\cdot  \exp_k(n) \leq \exp_k(cn)$.
\end{lem}
\begin{proof}
 Induction on $k$, where $k = 0$ is trivial.
 For $k \geq 1$,
 \begin{align*}
  c + \exp_{k+1}(n) & = c + 2^{\exp_k(n)} \leq 2^c \cdot 2^{\exp_k(n)}\tag{As $c + a \leq 2^c \cdot a$ for $c\geq 0,a \geq 1$} \\ % chktex 1
                    & = 2^{c + \exp_k(n)}\leq 2^{\exp_k(c+n)}\tag{Induction hypothesis}                                        \\
                    & =\exp_{k+1}(c+n)\text{.}
 \end{align*}
 For the product, the cases $c = 0,1$ are trivial.
 For $c \geq 2$,
 \begin{align*}
  c \cdot \exp_{k+1}(n) & \leq 2^{c-1} \cdot 2^{\exp_k(n)}\tag{Since $c \geq 2$ implies $c \leq 2^{c-1}$}                   \\ % chktex 1
                        & = 2^{c - 1 + \exp_k(n)} \leq 2^{\exp_k(c - 1 +n)}\tag{By $+$ case}                                \\
                        & \leq 2^{\exp_k(cn)}= \exp_{k+1}(cn)\text{.}\tag*{(As $(c-1)+n \leq cn$ for $c,n \geq 1$)\qedhere} % chktex 1 chktex 36
 \end{align*}
\end{proof}

Recall that $\exp^*_0(n) \dfn n$ and $\exp^*_{k+1}(n) \dfn n \cdot 2^{\exp^*_k(n)}$.

\begin{lem}\label{lem:exp-star}
 Let $n,k \geq 0$.
 Then $\exp_{k}^*(n) \leq \exp_{k}((k+1)\cdot n)$.
\end{lem}
\begin{proof}
 Induction on $k$.
 For $k = 0$, $\exp_0^*(n) = n = \exp_0((0+1)\cdot n)$.
 For the inductive step,
 \begin{align*}
  \exp^*_{k+1}(n) & = n \cdot 2^{\exp^*_{k}(n)} \leq 2^{n} \cdot 2^{\exp^*_{k}(n)} = 2^{n + \exp^*_{k}(n)}   \\
                  & \leq2^{n + \exp_{k}((k+1)n)}\tag{Induction hypothesis}                                   \\
                  & \leq2^{\exp_{k}(n + (k+1)n)}=\exp_{k+1}((k+2)n)\tag*{(Lemma~\ref{lem:into-exp})\qedhere}
 \end{align*}
\end{proof}

The next inequality states that a polynomial can be "pulled inside" $\exp_k$:

\begin{lem}\label{lem:exp-poly}
 For every polynomial $p$ there is a polynomial $q$ such that $p(\exp_k(n)) \leq \exp_k(q(n)))$ for all $k \geq 0, n \geq 1$.
\end{lem}
\begin{proof}
 For every polynomial $p$ there are integers $c,d \geq 1$ such that $p(n) \leq c n^d$ for all $n\geq 1$.
 Let $q(n) \dfn cdn^d+c$.
 Then the case $k = 0$ is clear.
 For $k \geq 1$ and $n \geq 1$,
 \begin{align*}
  p(\exp_{k}(n)) & \leq c \cdot {\exp_k(n)}^d \leq 2^c \cdot {(2^{\exp_{k-1}(n)})}^d = 2^{c + d \cdot \exp_{k-1}(n)} \\
                 & \leq 2^{q(\exp_{k-1}(n))}\tag{As $q(n) \geq c+dn$}                                            \\ % chktex 1 chktex 36
                 & \leq 2^{\exp_{k-1}(q(n))}=\exp_{k}(q(n))\text{.}\tag*{(Lemma~\ref{lem:into-exp})\qedhere}
 \end{align*}
\end{proof}

Finally, we combine both lemmas:

\begin{proof}[Proof of~\ref{lem:poly-in-tower}]
 Let $p$ be a polynomial as above.
 \Wloss $p$ is non-decreasing.
 Then by Lemma~\ref{lem:exp-star}, $p(\exp^*_k(n)) \leq p\big(\exp_k((k+1)\cdot n)\big)$.
 Moreover, due to Lemma~\ref{lem:exp-poly}, there is a polynomial $q$ such that
 $p\big(\exp_k((k+1)\cdot n)\big) \leq \exp_k\big(q((k+1)\cdot n)\big)$.
\end{proof}

\subsection*{Proofs of Section~\ref{sec:encoding}}

\begin{repprop}{prop:of-scopes}
 Let $\alpha,\beta$ be disjoint scopes and $S,U,T$ teams in a Kripke structure $\calK = (W,R,V)$.
 Then the following laws hold:
 \begin{enumerate}
  \item Distributive laws: ${(T \cap S)}_\alpha = T_\alpha \cap S = T \cap S_\alpha = T_\alpha \cap S_\alpha$ and ${(T \cup S)}_\alpha = T_\alpha \cup S_\alpha$.
  \item Disjoint selection commutes: ${\big(T^\alpha_S\big)}^{\beta}_{U} = {\big(T^{\beta}_{U}\big)}^{\alpha}_S$.
  \item Disjoint selection is independent: ${\big({(T^\alpha_S)}^{\beta}_{U}\big)}_\alpha = T_\alpha \cap S$.
  \item Image and selection commute: ${(RT)}_\alpha = {\big(R(T_\alpha)\big)}_\alpha = R(T_\alpha)$
  \item Selection propagates: If $S \subseteq T$, then $R\big(T^\alpha_{S}\big) = {(RT)}^\alpha_{RS}$.
 \end{enumerate}
\end{repprop}
\begin{proof}
 \begin{enumerate}
  \item
        Observe that $X_\alpha = X \cap W_\alpha$.
        Hence, for the union ${(T\cup S)}_\alpha  = (T\cup S) \cap W_\alpha = (T \cap W_\alpha) \cup (S \cap W_\alpha) = T_\alpha \cup S_\alpha$ holds.
        For the intersection, likewise $(T\cap S) \cap W_\alpha = (T \cap W_\alpha) \cap S = T \cap (W_\alpha \cap S) = (T \cap W_\alpha) \cap (S \cap W_\alpha)$.

  \item
        Proved in the following equation.
        We use the fact that $X_{\gamma\land\gamma'} = {(X_{\gamma})}_{\gamma'} = {(X_{\gamma'})}_\gamma = X_{\gamma'\land\gamma}$ for all teams $X$ and scopes $\gamma,\gamma$.
        \begin{align*}
              & {\big(T^\alpha_S\big)}^{\beta}_{U}                                                                                                                                                                                     \\
         = \; & {\big(T_{\neg\alpha} \cup (T_\alpha \cap S)\big)}_{\neg\beta} \cup \Big({\big(T_{\neg\alpha} \cup (T_\alpha \cap S) \big)}_\beta \cap U\Big)                                                                             \\
         \intertext{Distributing all scopes according to (1):}
         = \; & T_{\neg\alpha\land\neg\beta} \cup \big(T_{\alpha\land\neg\beta} \cap S_{\neg\beta}\big) \cup \big(T_{\neg\alpha\land\beta} \cap U\big) \cup \big(T_{\alpha\land\beta} \cap S_{\beta} \cap U\big)                     \\
         \intertext{Replace $U$ by $U_{\neg\alpha}/U_\alpha$ due to the intersection law of (1):}
         = \; & T_{\neg\alpha\land\neg\beta} \cup \big(T_{\alpha\land\neg\beta} \cap S_{\neg\beta}\big) \cup \big(T_{\neg\alpha\land\beta} \cap U_{\neg\alpha}\big) \cup \big(T_{\alpha\land\beta} \cap S_\beta \cap U_{\alpha}\big) \\
         \intertext{Likewise, replace $S_{\neg\beta}/S_\beta$ by $S$:}
         = \; & T_{\neg\alpha\land\neg\beta} \cup \big(T_{\alpha\land\neg\beta} \cap S\big) \cup \big(T_{\neg\alpha\land\beta} \cap U_{\neg\alpha}\big) \cup \big(T_{\alpha\land\beta} \cap S \cap U_{\alpha}\big)                   \\
         \intertext{Reverse distribution of scopes:}
         = \; & {\big(T_{\neg \beta} \cup (T_\beta \cap U)\big)}_{\neg\alpha} \cup \Big({\big(T_{\neg \beta} \cup (T_\beta \cap U)\big)}_\alpha \cap S \Big)                                                                             \\
         =\;  & {\big(T^\beta_U\big)}^\alpha_S\text{.}
        \end{align*}

  \item
        By definition and application of (2), ${\big({(T^\alpha_S)}^{\beta}_{U}\big)}_\alpha$ equals
        \begin{align*}
             & {\Big[{\big(T_{\neg \beta} \cup (T_\beta \cap U)\big)}_{\neg\alpha} \cup \Big({\big(T_{\neg \beta} \cup (T_\beta \cap U)\big)}_\alpha \cap S \Big)\Big]}_\alpha  \\
         =\; & {\big(T_{\neg \beta} \cup (T_\beta \cap U)\big)}_{\neg\alpha\land\alpha} \cup {\Big({\big(T_{\neg \beta} \cup (T_\beta \cap U)\big)}_\alpha \cap S \Big)}_\alpha \\
         =\; & \;\emptyset \cup \Big({\big(T_{\neg \beta} \cup (T_\beta \cap U)\big)}_\alpha \cap S_\alpha \Big)                                                            \\
         =\; & \big(T_{\neg \beta\land\alpha}\cap S_\alpha\big) \cup \big(T_{\beta\land\alpha} \cap U_\alpha\cap S_\alpha\big)                                            \\
         \intertext{Since $\alpha$ and $\beta$ are disjoint:}
         =\; & \big(T_{\alpha}\cap S_\alpha\big)\cup (\emptyset \cap U_\alpha \cap S_\alpha) = T_\alpha \cap S\text{.}
        \end{align*}

  \item ${(RT)}_\alpha \subseteq {\big(R(T_\alpha)\big)}_\alpha$:
        Suppose $v \in {(RT)}_\alpha$.
        Then $v \in Rw$ for some $w\in T$.
        Moreover, $w \in T_\alpha$, since $\alpha$ is a scope.
        Hence $v \in R(T_\alpha)$.
        As $v \vDash \alpha$, $v \in {\big(R(T_\alpha)\big)}_\alpha$ follows.

        ${\big(R(T_\alpha)\big)}_\alpha \subseteq R(T_\alpha)$: Obvious.

        $R(T_\alpha) \subseteq {(RT)}_\alpha$:
        Again, let $v \in R(T_\alpha)$ be arbitrary.
        Then $v \in Rw$ for some $w \in T_\alpha$.
        Hence $v \in RT$.
        Since $v \vDash \alpha$ follows from $w \vDash \alpha$, we conclude $v \in {(RT)}_\alpha$.
  \item For "$\subseteq$", suppose $v \in R(T^\alpha_{S})$, \ie, $v \in Rw$ for some $w \in T^\alpha_{S}$.
        In particular, $v \in RT$.
        If $w \nvDash \alpha$, then $v \in RT_{\neg\alpha}$ and trivially $v \in {(RT)}^\alpha_{RS}$.
        If $w \vDash \alpha$, then necessarily $w \in S$.
        Moreover, $v \vDash \alpha$.
        Consequently, $v \in {RS}_{\alpha} \cap {RT}_\alpha$, hence $v \in {(RT)}^\alpha_{RS}$.

        For "$\supseteq$", suppose $v \in {(RT)}^\alpha_{RS} = RT_{\neg\alpha} \cup (RT_\alpha \cap RS)$.

        If $v \in RT_{\neg\alpha}$, then by (4) $v \in Rw$ for some $w \in T_{\neg\alpha}$.
        In particular, $w \in T^\alpha_S$, hence $v \in R\big(T^\alpha_S\big)$.

        If $v \in RT_\alpha \cap RS$, then by (1) $v \in RS_\alpha$.
        By (4) $v \in R(S_\alpha)$, in other words, $v \in Rw$ for some $w \in S_\alpha$.
        As $S \subseteq T$, then $w \in S_\alpha \cap T$, and in fact $w \in T_\alpha \cap S$ due to (1)
        Consequently, $w \in T^\alpha_S$ and $v \in R(T^\alpha_{S})$.\qedhere
 \end{enumerate}
\end{proof}

\subsection*{Proofs of Section~\ref{sec:order}}

\begin{replem}{lem:substitution}
 Let $\alpha,\beta \in \ML$ and $\varphi \in \MTL_k$.
 Let $T$ be a team such that $R^{i}T \vDash \alpha \leftrightarrow \beta$ for all $i\in\{0,\ldots,k\}$.
 Then $T \vDash \varphi$ if and only if $T \vDash \mathrm{Sub}(\varphi,\alpha,\beta)$, where $\mathrm{Sub}(\varphi,\alpha,\beta)$ is the formula obtained from $\varphi$ by substituting every occurrence of $\alpha$ with $\beta$.
\end{replem}
\begin{proof}
 Proof by induction on $k$ and the syntax on $\varphi$.
 \Wloss $\alpha$ occurs in $\varphi$.
 If $\varphi = \alpha$, then $\mathrm{Sub}(\varphi,\alpha,\beta) = \beta$, in which case the proof boils down to showing $T \vDash \alpha \Leftrightarrow T \vDash \beta$.
 However, this easily follows from $T\vDash \alpha \leftrightarrow \beta$ by the semantics for classical $\ML$-formulas.

 Otherwise, $\alpha$ is a proper subformula of $\varphi$.
 We distinguish the following cases.
 \begin{itemize}
  \item $\varphi = \neg \gamma$: Then $\mathrm{Sub}(\neg\gamma,\alpha,\beta) = \neg\mathrm{Sub}(\gamma,\alpha,\beta)$, and
        \begin{align*}
                           & T \vDash \mathrm{Sub}(\varphi,\alpha,\beta)                                                                                            \\
         \Leftrightarrow\; & T \vDash \neg \mathrm{Sub}(\gamma,\alpha,\beta)                                                                                        \\
         \Leftrightarrow\; & \forall  w \in T \colon \{ w \} \vDash \neg \mathrm{Sub}(\gamma,\alpha,\beta)                                                          \\
         \Leftrightarrow\; & \forall  w \in T \colon \{ w \} \vDash \neg \gamma\tag{Induction hypothesis, as $\{w\},Rw,\ldots \vDash \alpha \leftrightarrow \beta$} \\ % chktex 1
         \Leftrightarrow\; & T \vDash \neg \gamma                                                                                                                   \\
         \Leftrightarrow\; & T \vDash \varphi
        \end{align*}
  \item $\varphi = \negg \psi$: By induction hypothesis, $T \vDash \mathrm{Sub}(\varphi,\alpha,\beta)$ iff $T \vDash \negg \mathrm{Sub}(\psi,\alpha,\beta)$ iff $T \vDash \negg \psi$.
  \item $\varphi = \psi \land \theta$: Proved similarly to $\negg$.
  \item $\varphi = \psi \lor \theta$: First note that $\mathrm{Sub}(\psi \lor \theta,\alpha,\beta) = \mathrm{Sub}(\psi,\alpha,\beta) \lor \mathrm{Sub}(\theta,\alpha,\beta)$.
        Then:
        \begin{align*}
                           & T \vDash \mathrm{Sub}(\varphi,\alpha,\beta)                                                                           \\
         \Leftrightarrow\; & T \vDash \mathrm{Sub}(\psi,\alpha,\beta) \lor \mathrm{Sub}(\theta,\alpha,\beta)                                       \\
         \Leftrightarrow\; & \exists S,U \colon T = S \cup U, S \vDash \mathrm{Sub}(\psi,\alpha,\beta), U \vDash \mathrm{Sub}(\theta,\alpha,\beta)
         \intertext{By induction hypothesis, since $S,U,RS,RU,\ldots \vDash \alpha\leftrightarrow\beta$:}
         \Leftrightarrow\; & \exists S,U \colon T = S \cup U, S \vDash \psi, U \vDash \theta                                                       \\
         \Leftrightarrow\; & T \vDash \varphi
        \end{align*}
  \item $\varphi = \Box\psi$: We have $\mathrm{Sub}(\Box\psi,\alpha,\beta) = \Box\mathrm{Sub}(\psi,\alpha,\beta)$, hence
        \begin{align*}
                           & T \vDash \mathrm{Sub}(\varphi,\alpha,\beta)       \\
         \Leftrightarrow\; & T \vDash \Box \mathrm{Sub}(\psi,\alpha,\beta)     \\
         \Leftrightarrow\; & RT \vDash \mathrm{Sub}(\psi,\alpha,\beta)\text{.} \\
         \intertext{However, since $\psi \in \MTL_{k-1}$ and $RT,\ldots,R^{k-1}(RT) \vDash \alpha\leftrightarrow\beta$ holds by assumption, we obtain by induction hypothesis:}
         \Leftrightarrow\; & RT \vDash \psi                                    \\
         \Leftrightarrow\; & T \vDash \varphi
        \end{align*}
  \item $\varphi = \Diamond \psi$: As before, $\mathrm{Sub}(\Diamond\psi,\alpha,\beta) = \Diamond\mathrm{Sub}(\psi,\alpha,\beta)$.
        Then:
        \begin{align*}
                           & T \vDash \mathrm{Sub}(\varphi,\alpha,\beta)                                                 \\
         \Leftrightarrow\; & T \vDash \Diamond \mathrm{Sub}(\psi,\alpha,\beta)                                           \\
         \Leftrightarrow\; & \exists S \subseteq RT, T \subseteq R^{-1}S \colon S \vDash \mathrm{Sub}(\psi,\alpha,\beta) \\
         \intertext{Note that $S, RS, \ldots, R^{k-1}S$ are subteams of $RT, \ldots, R^{k}T$, respectively.
          For this reason, the teams $S, RS, \ldots, R^{k-1}S$ satisfy $\alpha \leftrightarrow \beta$ as well.
          As also $\psi \in \MTL_{k-1}$ holds, we obtain by induction hypothesis:}
         \Leftrightarrow\; & \exists S \subseteq RT, T \subseteq R^{-1}S \colon S \vDash \psi                            \\
         \Leftrightarrow\; & T \vDash \varphi\tag*{\qedhere}
        \end{align*}
 \end{itemize}
\end{proof}

\end{document}